\documentclass[pdflatex,sn-mathphys,Numbered,iicol]{sn-jnl}
\usepackage{graphicx}%
\usepackage{amsmath,amssymb,amsfonts}%
\usepackage{amsthm}%
\usepackage{mathrsfs}%
\usepackage[title]{appendix}%
\usepackage{xcolor}%
\usepackage{textcomp}%
\usepackage{manyfoot}%
\usepackage{booktabs}%
\usepackage{hyperref}
\usepackage{bookmark}

\usepackage[normalem]{ulem}
\usepackage[subrefformat=parens,labelformat=parens,caption=false]{subfig}
\usepackage{tikz}
\usepackage{tikz-cd}
\usetikzlibrary{arrows, arrows.meta, fit, positioning, shapes, backgrounds}
\usepackage[ruled,vlined]{algorithm2e}
\SetKwInput{KwInput}{Input}                
\SetKwInput{KwOutput}{Output}              
\usepackage{float}

\usepackage{multirow}
\usepackage{tabulary}

\newtheorem{axiom}{Axiom}
\newtheorem{assumption}[axiom]{Assumption}

\newtheorem{theorem}{Theorem}[section]
\newtheorem{lemma}[theorem]{Lemma}

\theoremstyle{remark}
\newtheorem{definition}[theorem]{Definition}
\newtheorem*{example}{Example}

\newtheorem*{remark}{Remark}

\newcommand\red[1]{{#1}}
\newcommand\redsout[1]{{}} 

\newcommand{\Z}{\mathbb{Z}}

\newcommand{\E}{\mathbb{E}}
\newcommand{\R}{\mathbb{R}}

\renewcommand{\P}{\mathbb{P}}

\newcommand{\eps}{\varepsilon}

\newcommand{\diam}{\textnormal{diam}}

\newcommand\supnorm[1]{\sup\limits_{#1 \in [0, T]}}
\newcommand{\chieq}{\overset{\mathrm{\chi}}{=}}

\begin{document}

\title[Topology-Driven Goodness-of-Fit Tests in Arbitrary Dimensions]{Topology-Driven Goodness-of-Fit Tests in Arbitrary Dimensions}


\author[1]{\fnm{Pawe{\l}} \sur{D{\l}otko}}\email{pawel.dlotko@impan.pl}

\author[1]{\fnm{Niklas} \sur{Hellmer}}\email{nhellmer@impan.pl}

\author[1]{\fnm{{\L}ukasz} \sur{Stettner}}\email{stettner@impan.pl}

\author[1]{\fnm{Rafa{\l}} \sur{Topolnicki}}\email{rafal.topolnicki@impan.pl}

\affil[1]{\orgdiv{Institute of Mathematics}, \orgname{Polish Academy of Sciences}, \orgaddress{\street{{\'S}niadeckich 8}, \city{Warsaw}, \postcode{00-656}, \country{Poland}}}

\abstract{This paper adopts a tool from computational topology, the Euler characteristic curve (ECC) of a sample, to perform one- and two-sample goodness of fit tests\red{. We call our procedure \emph{TopoTests}}. The presented tests work for samples \red{of} arbitrary dimension, having comparable power to the state-of-the-art tests in the one-dimensional case. It is demonstrated that the type I error of TopoTests can be controlled and their type II error vanishes exponentially with increasing sample size. Extensive numerical simulations of TopoTests are conducted to demonstrate their power \red{for samples of various sizes}.
}


\maketitle
        
\section{Introduction}
Goodness-of-fit (GoF) testing is one of the standard tasks in statistics.
The testing procedure can be stated in the one-sample or two-sample setting. 
In case of the one-sample problem, we observe a sample of $m$ independent realizations $\{x_1, \ldots, x_m \}$ of a $d$-dimensional 
random vector $X$ with an unknown distribution function $G$, i.e. $x_i \sim G$. 
The task is to test whether $G$ is equal
to \red{a} specific distribution $F$, i.e. we would like to test
\begin{equation}
    \label{eqn:OneSampleHypothesis}
    H_0: G = F~~\text{vs.}~~H_1: G \ne F.
\end{equation}
In the setting of the two-sample problem we are given two independent samples consisting of $m$ and $n$ ($m \ne n$ in general) independent realizations
of \redsout{an} $d$-dimensional random vectors $X$ and $Y$ with an unknown distribution function $F$ and $G$, respectively. This means $X=\{x_1, \ldots, x_m\}, x_i \sim F$ and $Y=\{y_1,\ldots,y_n\}, y_j \sim G$, while the hypothesis is the same as in (\ref{eqn:OneSampleHypothesis}).

In this paper, we consider a more general notion of equivalence, replacing the equal sign above by the relation of being \textit{Euler equivalent} (cf. Definition~\ref{def:EulerEquivalent}).

We are interested in the setting in which the underlying distribution is continuous.  
In this case, prominent GoF tests for samples from $\R$ rely on the empirical distribution function, see  \cite[chapter 4]{dagostino1986goodness}.
These include, in the one dimensional case, the Kolmogorov-Smirnov, Cram\'e{}r-von-Mises and Anderson-Darling tests. In higher dimensions, Kolmogorov-Smirnov leads to Fasano-Franceschini\cite{Fasano1987} and Peacock\cite{Peacock1983} tests; a general case was 
considered by Justel~\cite{Justel1997}. 
\red{A m}ultivariate version of Cram\'e{}r-von-Mises was proposed by Chiu and Liu\cite{Chiu2009}. 
Since those tests are based on empirical distribution function, their generalization to $\R^d$ for $d \geq 2$ is conceptually and computationally difficult. Moreover, we are not aware of an efficient implementation of a general goodness of fit tests for high dimensional samples.

To tackle this challenge we propose to replace the cumulative distribution function with \textit{Euler characteristic curves (ECCs)}  
\cite{gonzalez1977digital, richardsonEfficientClassificationUsing2014, Worsley1996Geometry}, a tool from computational topology that provides a signature of the considered sample.
To a given sample $X$, this notion associates a function $\chi(X)\colon [0,\infty) \to \Z$, which can serve as a stand-in for the empirical distribution function in arbitrary dimensions. Subsequently, for \red{one}-sample tests, inspired by the Kolmogorov-Smirnov test, we define the test statistic to be the supremum distance between the ECC of the sample and the expected ECC for the distribution. 
This topologically driven testing scheme will be referred to as \red{``}TopoTest'' for short.

The key characteristic of any goodness of fit test is its power, i.e. the type II error should be small,
under the requirement that the type I error is fixed at level $\alpha$.
We show that the proposed test satisfies this condition and that it performs very well in practical cases. In particular, even restricted to one dimensional samples, its power is comparable to those of the standard GoF tests.

The paper is organized as follows:
Section~\ref{sec:CompTop} reviews the necessary background from topology as well as the current work in the topic. In Section~\ref{sec:method} we present the theoretical justification of our method. In Section~\ref{sec:algorithms} the algorithms implementing proposed GoF tests are detailed. 
Sections~\ref{sec:NumericalExperiments_one_sample} and~\ref{sec:NumericalExperiments2sample} present the numerical experiments and comparison of the presented technique to existing methods. In particular, comparing to a higher dimensional version of the Kolmogorov-Smirnov test, we find that our procedure provides better power and takes less time to compute.
Finally, in Section~\ref{sec:discussion} the conclusions are drawn.

\subsection{Background}

\red{Since the seminal work of Edelsbrunner et al. \cite{edelsbrunner_topological_2002} and Carlsson \& Zomorodian \cite{zomorodian_computing_2005}, t}opological Data Analysis (TDA) is a fast growing interdisciplinary area 
combining tools and results of such diverse fields of science as algebra, topology, statistics and machine learning,
just to name few. 
\red{For a survey from a statistician's perspective, see \cite{wasserman_topological_2018}.}
One of the areas in which TDA can contribute to statistics is related to applications of topological
summaries of the data to hypothesis testing. 
Despite ongoing research and growing interest in TDA methods, attempts to construct statistical tests within the classical
Neyman-Person hypothesis testing paradigm based on persistent homology, the most popular topological summaries of data,
are limited because the distributions of test statistics under the null hypothesis are unknown. 
Therefore, the approaches that are most common in the literature utilize sampling and permutation based techniques  \cite{Cericola2016, Robinson2017, vejdemo-johanssonMultipleTestingPersistent2022}.
In this work, a different topological summary of the data, namely the Euler characteristic curve (ECC), is used to construct one-sample and two-sample statistical tests.
The application of ECCs is motivated by recent theoretical findings regarding the asymptotic distribution of ECC,
which enables us to construct tests in rigorous fashion. 
Since the finite sample distributions of ECCs remain unknown, extensive Monte Carlo simulations were conducted to investigate the properties and performances of the proposed tests.

\subsubsection*{Tools from Computational Topology}\label{sec:CompTop}

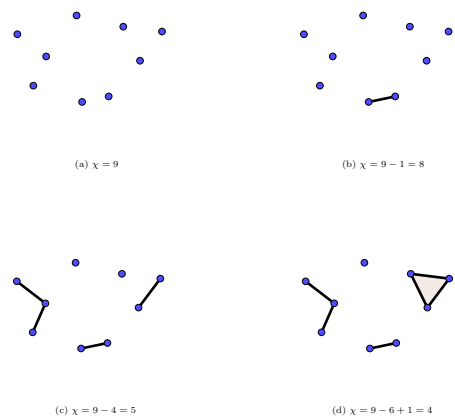
\begin{figure}
    \resizebox{\linewidth}{!}{
	\subfloat[$\chi=9$\label{fig:PointCloud}]{
		\definecolor{ududff}{rgb}{0.30196078431372547,0.30196078431372547,1}
\begin{tikzpicture}[line cap=round,line join=round,>=triangle 45,x=1cm,y=1cm]
\clip(-0.4815975434601234,0.4826342660750136) rectangle (6.687880545753719,6.195416277091689);
\begin{scriptsize}
\draw [fill=ududff] (0.96,3.66) circle (2.5pt);
\draw [fill=ududff] (2.51413,4.15604) circle (2.5pt);
\draw [fill=ududff] (3.74,3.86) circle (2.5pt);
\draw [fill=ududff] (4.75617,3.7317) circle (2.5pt);
\draw [fill=ududff] (1.71611,3.07302) circle (2.5pt);
\draw [fill=ududff] (4.17983,2.95902) circle (2.5pt);
\draw [fill=ududff] (1.38,2.3) circle (2.5pt);
\draw [fill=ududff] (2.6598,1.86966) circle (2.5pt);
\draw [fill=ududff] (3.35648,2.01533) circle (2.5pt);
\end{scriptsize}
\end{tikzpicture}%
	}
	\subfloat[$\chi = 9-1 = 8$\label{fig:VR1}]{
		\definecolor{ududff}{rgb}{0.30196078431372547,0.30196078431372547,1}
\begin{tikzpicture}[line cap=round,line join=round,>=triangle 45,x=1cm,y=1cm]
\clip(-0.4815975434601234,0.48263426607501353) rectangle (6.687880545753719,6.195416277091688);
\draw [line width=2pt] (2.6598,1.86966)-- (3.35648,2.01533);
\begin{scriptsize}
\draw [fill=ududff] (0.96,3.66) circle (2.5pt);
\draw [fill=ududff] (2.51413,4.15604) circle (2.5pt);
\draw [fill=ududff] (3.74,3.86) circle (2.5pt);
\draw [fill=ududff] (4.75617,3.7317) circle (2.5pt);
\draw [fill=ududff] (1.71611,3.07302) circle (2.5pt);
\draw [fill=ududff] (4.17983,2.95902) circle (2.5pt);
\draw [fill=ududff] (1.38,2.3) circle (2.5pt);
\draw [fill=ududff] (2.6598,1.86966) circle (2.5pt);
\draw [fill=ududff] (3.35648,2.01533) circle (2.5pt);
\end{scriptsize}
\end{tikzpicture}%
	}
    } \newline
    \resizebox{\linewidth}{!}{
	\subfloat[$\chi = 9-4 = 5$\label{fig:VR2}]{
		\definecolor{ududff}{rgb}{0.30196078431372547,0.30196078431372547,1}
\begin{tikzpicture}[line cap=round,line join=round,>=triangle 45,x=1cm,y=1cm]
\clip(-0.4815975434601234,0.48263426607501353) rectangle (6.687880545753719,6.195416277091688);
\draw [line width=2pt] (0.96,3.66)-- (1.71611,3.07302);
\draw [line width=2pt] (4.75617,3.7317)-- (4.17983,2.95902);
\draw [line width=2pt] (1.71611,3.07302)-- (1.38,2.3);
\draw [line width=2pt] (2.6598,1.86966)-- (3.35648,2.01533);
\begin{scriptsize}
\draw [fill=ududff] (0.96,3.66) circle (2.5pt);
\draw [fill=ududff] (2.51413,4.15604) circle (2.5pt);
\draw [fill=ududff] (3.74,3.86) circle (2.5pt);
\draw [fill=ududff] (4.75617,3.7317) circle (2.5pt);
\draw [fill=ududff] (1.71611,3.07302) circle (2.5pt);
\draw [fill=ududff] (4.17983,2.95902) circle (2.5pt);
\draw [fill=ududff] (1.38,2.3) circle (2.5pt);
\draw [fill=ududff] (2.6598,1.86966) circle (2.5pt);
\draw [fill=ududff] (3.35648,2.01533) circle (2.5pt);
\end{scriptsize}
\end{tikzpicture}%
	}
	\subfloat[$\chi = 9- 6 + 1 = 4$\label{fig:VR3}]{
		\definecolor{zzttqq}{rgb}{0.6,0.2,0}
\definecolor{ududff}{rgb}{0.30196078431372547,0.30196078431372547,1}
\begin{tikzpicture}[line cap=round,line join=round,>=triangle 45,x=1cm,y=1cm]
\clip(-0.4815975434601234,0.48263426607501353) rectangle (6.687880545753719,6.195416277091688);
\fill[line width=2pt,color=zzttqq,fill=zzttqq,fill opacity=0.10000000149011612] (3.74,3.86) -- (4.17983,2.95902) -- (4.75617,3.7317) -- cycle;
\draw [line width=2pt] (3.74,3.86)-- (4.75617,3.7317);
\draw [line width=2pt] (0.96,3.66)-- (1.71611,3.07302);
\draw [line width=2pt] (3.74,3.86)-- (4.17983,2.95902);
\draw [line width=2pt] (4.75617,3.7317)-- (4.17983,2.95902);
\draw [line width=2pt] (1.71611,3.07302)-- (1.38,2.3);
\draw [line width=2pt] (2.6598,1.86966)-- (3.35648,2.01533);
\begin{scriptsize}
\draw [fill=ududff] (0.96,3.66) circle (2.5pt);
\draw [fill=ududff] (2.51413,4.15604) circle (2.5pt);
\draw [fill=ududff] (3.74,3.86) circle (2.5pt);
\draw [fill=ududff] (4.75617,3.7317) circle (2.5pt);
\draw [fill=ududff] (1.71611,3.07302) circle (2.5pt);
\draw [fill=ududff] (4.17983,2.95902) circle (2.5pt);
\draw [fill=ududff] (1.38,2.3) circle (2.5pt);
\draw [fill=ududff] (2.6598,1.86966) circle (2.5pt);
\draw [fill=ududff] (3.35648,2.01533) circle (2.5pt);
\end{scriptsize}
\end{tikzpicture}%
	}
	}
	\caption{With increasing scale parameter, we draw in edges and triangles. 
	We keep track of the number of components, which is here $\#\text{points} - \#\text{edges} + \#\text{triangles}$.}
\end{figure}

To start with an example, let us consider the set $X$ of nine points in $\R^2$ (Figure \ref{fig:PointCloud}).
The most elementary way of assigning a numeric quantity to them is to simply count them.
This is a topological invariant, the number of connected components.
\red{Now if two points coincide, they should not be regarded as separate.
If they are very close together, say less than some given $\eps>0$ apart, we can also connect them.
}
So let us draw an edge between them (Figure \ref{fig:VR1}).
The number of connected components is now one less, suggesting we should subtract the number of edges from the number of points.
In order to formalize what we mean by points that are close to each other, we introduce a scale parameter $r\in \R^{\geq 0}$.
Then we draw edges between pairs of points whose distance is at most $r$.
Letting $r=0$ initially and increasing it, we draw more and more edges, thereby reducing the number of connected components (Figure \ref{fig:VR2}).
Once three points are within distance $r$ of each other, according to our intuition they should be considered as one connected component.
But we have three points and three edges, which yield a difference of zero.
To correct this mismatch with our intuition, we add the number of triangles (Figure \ref{fig:VR3}).
This procedure continues to higher dimensions: Once $k$ points are within distance $r$ of each other, we add\ $(-1)^{k\red{-1}}$.

These ideas will now be formalized.
For a textbook reference on these topics, we refer the reader to \cite{edelsbrunnerComputationalTopologyIntroduction2010}.
\begin{definition}
	An \textit{abstract simplicial complex} $K$ is a collection of nonempty sets which are closed under the subset operation:
	\[
		\tau \in K \text{ and } \sigma \subseteq \tau \Rightarrow \sigma \in K.
	\]
	The elements of $K$ are called \textit{simplices}.
	If $\sigma \subsetneq \tau \in K$, we say that $\sigma$ is a \textit{face} of $\tau$.
	The \textit{dimension of a simplex} $\sigma \in K$ is $\dim(\sigma) = \vert \sigma \vert -1$, where $\vert\cdot\vert$ denotes the cardinality of a set.
	The \textit{dimension of} $K$ is the the maximal dimension of any of its simplices.
\end{definition}

The construction of drawing edges, triangles etc. between points which are close to each other can be formalized in slightly different flavours.
Perhaps the simplest is the Vietoris-Rips construction:

\begin{definition}
	For a finite subset $X\subseteq \R^d$ and $r\ge 0$ define the \textit{Vietoris-Rips complex} at scale $r$ to be the abstract simplicial complex
	\[
	\mathcal{R}_r(X) = \left\{ \sigma \subseteq X \colon \diam(\sigma)\leq 2r \right\}
	\]
    \red{where $\diam$ is the diameter of the simplex $\diam(\sigma) = \max\{d(x,x')\colon x,x' \in \sigma, x\neq x'\}$.}
\end{definition}

A closely related notion is the \v{C}ech complex:
\begin{definition}
For a finite subset $X\subseteq \R^d$ and $r\ge 0$ define the \textit{\v{C}ech complex} at scale $r$ to be the abstract simplicial complex
\[
    \mathcal{C}_r(X) = \left\{ \sigma \subseteq X \colon \bigcap_{x\in\sigma} B_r(x) \neq \emptyset \right\},
\]
\red{w}here $B_r(x)$ is the closed ball of radius $r$ centered at $x$.
\end{definition}

Finally, the Alpha complex (which is the most useful in practice and used in our implementations), requires the following notion from computational geometry:

\begin{definition}
	Let $X\subseteq \R^d$ be a finite set.
	The \textit{Voronoi cell} of $x\in X$ is the subset of points in $\R^d$  that have $x$ as a closest point in $X$,
	\[
		V_X(x) = \{y\in \R^d \colon \forall x' \in X   \|y-x\| \leq \|y-x'\|\}.
	\]
\end{definition}
\begin{definition}
For a finite subset $X\subseteq \R^d$ and $r\ge 0$ define the \textit{Alpha complex} at scale $r$ to be the abstract simplicial complex
\[
\mathcal{A}_r(X) = \left\{ \sigma \subseteq X \colon \bigcap_{x\in\sigma} B_r(x)\cap V_X(x) \neq \emptyset \right\}.
\]
\end{definition}

For illustrations of the Alpha, \v{C}ech and Vietoris-Rips complex on a small sample\redsout{s}, consider 
Figures \ref{fig:Alpha}, \ref{fig:Cech} and \ref{fig:VR}, respectively.
We refer to $r$ as the \textit{scale parameter} or the \textit{filtration value}.
The latter name comes from the fact that for $r<r'$, the complex at scale $r$ is a subcomplex of the one at scale $r'$.

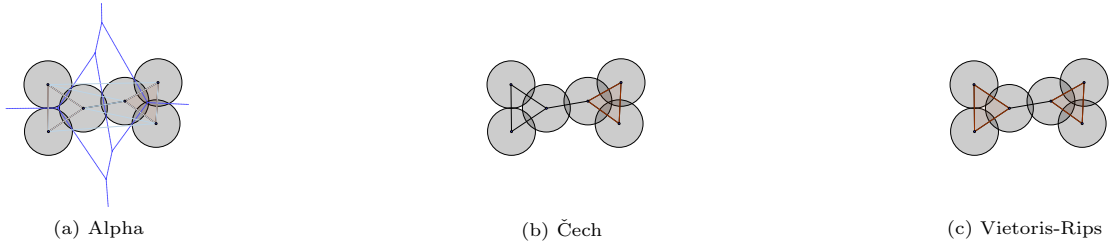
\begin{figure*}
	\definecolor{zzttqq}{rgb}{0.6,0.2,0}
	\definecolor{ududff}{rgb}{0.30196078431372547,0.30196078431372547,1}
	\definecolor{bcduew}{rgb}{0.7372549019607844,0.8313725490196079,0.9019607843137255}
	\definecolor{aqaqaq}{rgb}{0.6274509803921569,0.6274509803921569,0.6274509803921569}
	\subfloat[Alpha\label{fig:Alpha}]{\scalebox{0.15}{\begin{tikzpicture}[line cap=round,line join=round,>=triangle 45,x=1cm,y=1cm]
\clip(-12.616640864077574,-9.17579630531208) rectangle (12.698871070317466,8.967939674756556);
\draw [line width=2pt,fill=black,fill opacity=0.2] (-4.3455,1.80544) circle (2.1cm);
\draw [line width=2pt,fill=black,fill opacity=0.2] (-4.2993772823371055,-2.3460119162021873) circle (2.1cm);
\draw [line width=2pt,fill=black,fill opacity=0.2] (-1.2396040745723542,-0.28566211800883773) circle (2.1cm);
\draw [line width=2pt,fill=black,fill opacity=0.2] (2.404447434172601,0.34474341725927676) circle (2.1cm);
\draw [line width=2pt,fill=black,fill opacity=0.2] (5.310463194310983,1.989948106861429) circle (2.1cm);
\draw [line width=2pt,fill=black,fill opacity=0.2] (5.141330001921974,-1.6387276571208862) circle (2.1cm);
\fill[line width=2pt,color=zzttqq,fill=zzttqq,fill opacity=0.10000000149011612] (5.310463194310983,1.989948106861429) -- (5.141330001921974,-1.6387276571208862) -- (2.404447434172601,0.34474341725927676) -- cycle;
\draw [line width=2pt] (-4.3455,1.80544)-- (-4.2993772823371055,-2.3460119162021873);
\draw [line width=2pt] (-4.2993772823371055,-2.3460119162021873)-- (-1.2396040745723542,-0.28566211800883773);
\draw [line width=2pt] (-4.3455,1.80544)-- (-1.2396040745723542,-0.28566211800883773);
\draw [line width=2pt] (-1.2396040745723542,-0.28566211800883773)-- (2.404447434172601,0.34474341725927676);
\draw [line width=2pt] (2.404447434172601,0.34474341725927676)-- (5.310463194310983,1.989948106861429);
\draw [line width=2pt] (5.310463194310983,1.989948106861429)-- (5.141330001921974,-1.6387276571208862);
\draw [line width=2pt] (5.141330001921974,-1.6387276571208862)-- (2.404447434172601,0.34474341725927676);
\draw [line width=2pt,color=zzttqq] (5.310463194310983,1.989948106861429)-- (5.141330001921974,-1.6387276571208862);
\draw [line width=2pt,color=zzttqq] (5.141330001921974,-1.6387276571208862)-- (2.404447434172601,0.34474341725927676);
\draw [line width=2pt,color=zzttqq] (2.404447434172601,0.34474341725927676)-- (5.310463194310983,1.989948106861429);
\draw [line width=2pt,color=zzttqq] (-4.3455,1.80544)-- (-1.2396040745723542,-0.28566211800883773);
\draw [line width=2pt,color=zzttqq] (-1.2396040745723542,-0.28566211800883773)-- (-4.2993772823371055,-2.3460119162021873);
\draw [line width=2pt,color=zzttqq] (-4.2993772823371055,-2.3460119162021873)-- (-4.3455,1.80544);
\draw[line width=2pt,dotted,color=ududff] (0.37903400667212594,7.311473118421726) -- (0.35,9.03)(0.7633546765311897,-6.56237601443535) -- (0.94,-8.97)(4.397042446532841,0.2142432573593399) -- (8,0.05)(-3.479833713277305,-0.26092459971974896) -- (-8.02,-0.31)(-0.20801902640147274,4.5986735123027405) -- (1.2917298645142747,-4.070606662015118)(4.397042446532841,0.2142432573593399) -- (0.37903400667212594,7.311473118421726)(-3.479833713277305,-0.26092459971974896) -- (0.7633546765311897,-6.56237601443535)(-0.20801902640147274,4.5986735123027405) -- (-3.479833713277305,-0.26092459971974896)(1.2917298645142747,-4.070606662015118) -- (4.397042446532841,0.2142432573593399)(-0.20801902640147274,4.5986735123027405) -- (0.37903400667212594,7.311473118421726)(0.7633546765311897,-6.56237601443535) -- (1.2917298645142747,-4.070606662015118);
\draw[line width=2pt,color=bcduew] (5.310463194310983,1.989948106861429) -- (-4.3455,1.80544)(-4.2993772823371055,-2.3460119162021873) -- (5.141330001921974,-1.6387276571208862)(2.404447434172601,0.34474341725927676) -- (-4.3455,1.80544)(-1.2396040745723542,-0.28566211800883773) -- (5.141330001921974,-1.6387276571208862)(-4.3455,1.80544) -- (-4.2993772823371055,-2.3460119162021873)(-4.3455,1.80544) -- (-1.2396040745723542,-0.28566211800883773)(2.404447434172601,0.34474341725927676) -- (-1.2396040745723542,-0.28566211800883773)(-1.2396040745723542,-0.28566211800883773) -- (-4.2993772823371055,-2.3460119162021873)(5.141330001921974,-1.6387276571208862) -- (5.310463194310983,1.989948106861429)(5.141330001921974,-1.6387276571208862) -- (2.404447434172601,0.34474341725927676)(5.310463194310983,1.989948106861429) -- (2.404447434172601,0.34474341725927676);
\begin{scriptsize}
\draw [fill=ududff] (-4.3455,1.80544) circle (2.5pt);
\draw [fill=ududff] (-4.2993772823371055,-2.3460119162021873) circle (2.5pt);
\draw [fill=ududff] (-1.2396040745723542,-0.28566211800883773) circle (2.5pt);
\draw [fill=ududff] (2.404447434172601,0.34474341725927676) circle (2.5pt);
\draw [fill=ududff] (5.310463194310983,1.989948106861429) circle (2.5pt);
\draw [fill=ududff] (5.141330001921974,-1.6387276571208862) circle (2.5pt);
\end{scriptsize}
\end{tikzpicture}}}
	\hfill
	\subfloat[\v{C}ech\label{fig:Cech}]{\scalebox{0.15}{\begin{tikzpicture}[line cap=round,line join=round,>=triangle 45,x=1cm,y=1cm]
\clip(-12.616640864077574,-9.17579630531208) rectangle (12.698871070317466,8.967939674756556);
\draw [line width=2pt,fill=black,fill opacity=0.2] (-4.3455,1.80544) circle (2.1cm);
\draw [line width=2pt,fill=black,fill opacity=0.2] (-4.2993772823371055,-2.3460119162021873) circle (2.1cm);
\draw [line width=2pt,fill=black,fill opacity=0.2] (-1.2396040745723542,-0.28566211800883773) circle (2.1cm);
\draw [line width=2pt,fill=black,fill opacity=0.2] (2.404447434172601,0.34474341725927676) circle (2.1cm);
\draw [line width=2pt,fill=black,fill opacity=0.2] (5.310463194310983,1.989948106861429) circle (2.1cm);
\draw [line width=2pt,fill=black,fill opacity=0.2] (5.141330001921974,-1.6387276571208862) circle (2.1cm);
\fill[line width=2pt,color=zzttqq,fill=zzttqq,fill opacity=0.10000000149011612] (5.310463194310983,1.989948106861429) -- (5.141330001921974,-1.6387276571208862) -- (2.404447434172601,0.34474341725927676) -- cycle;
\draw [line width=2pt] (-4.3455,1.80544)-- (-4.2993772823371055,-2.3460119162021873);
\draw [line width=2pt] (-4.2993772823371055,-2.3460119162021873)-- (-1.2396040745723542,-0.28566211800883773);
\draw [line width=2pt] (-4.3455,1.80544)-- (-1.2396040745723542,-0.28566211800883773);
\draw [line width=2pt] (-1.2396040745723542,-0.28566211800883773)-- (2.404447434172601,0.34474341725927676);
\draw [line width=2pt] (2.404447434172601,0.34474341725927676)-- (5.310463194310983,1.989948106861429);
\draw [line width=2pt] (5.310463194310983,1.989948106861429)-- (5.141330001921974,-1.6387276571208862);
\draw [line width=2pt] (5.141330001921974,-1.6387276571208862)-- (2.404447434172601,0.34474341725927676);
\draw [line width=2pt,color=zzttqq] (5.310463194310983,1.989948106861429)-- (5.141330001921974,-1.6387276571208862);
\draw [line width=2pt,color=zzttqq] (5.141330001921974,-1.6387276571208862)-- (2.404447434172601,0.34474341725927676);
\draw [line width=2pt,color=zzttqq] (2.404447434172601,0.34474341725927676)-- (5.310463194310983,1.989948106861429);
\begin{scriptsize}
\draw [fill=ududff] (-4.3455,1.80544) circle (2.5pt);
\draw [fill=ududff] (-4.2993772823371055,-2.3460119162021873) circle (2.5pt);
\draw [fill=ududff] (-1.2396040745723542,-0.28566211800883773) circle (2.5pt);
\draw [fill=ududff] (2.404447434172601,0.34474341725927676) circle (2.5pt);
\draw [fill=ududff] (5.310463194310983,1.989948106861429) circle (2.5pt);
\draw [fill=ududff] (5.141330001921974,-1.6387276571208862) circle (2.5pt);
\end{scriptsize}
\end{tikzpicture}}}
	\hfill
	\subfloat[Vietoris-Rips\label{fig:VR}]{\scalebox{0.15}{\begin{tikzpicture}[line cap=round,line join=round,>=triangle 45,x=1cm,y=1cm]
\clip(-12.616640864077574,-9.17579630531208) rectangle (12.698871070317466,8.967939674756556);
\draw [line width=2pt,fill=black,fill opacity=0.2] (-4.3455,1.80544) circle (2.1cm);
\draw [line width=2pt,fill=black,fill opacity=0.2] (-4.2993772823371055,-2.3460119162021873) circle (2.1cm);
\draw [line width=2pt,fill=black,fill opacity=0.2] (-1.2396040745723542,-0.28566211800883773) circle (2.1cm);
\draw [line width=2pt,fill=black,fill opacity=0.2] (2.404447434172601,0.34474341725927676) circle (2.1cm);
\draw [line width=2pt,fill=black,fill opacity=0.2] (5.310463194310983,1.989948106861429) circle (2.1cm);
\draw [line width=2pt,fill=black,fill opacity=0.2] (5.141330001921974,-1.6387276571208862) circle (2.1cm);
\fill[line width=2pt,color=zzttqq,fill=zzttqq,fill opacity=0.10000000149011612] (5.310463194310983,1.989948106861429) -- (5.141330001921974,-1.6387276571208862) -- (2.404447434172601,0.34474341725927676) -- cycle;
\fill[line width=2pt,color=zzttqq,fill=zzttqq,fill opacity=0.10000000149011612] (-4.3455,1.80544) -- (-1.2396040745723542,-0.28566211800883773) -- (-4.2993772823371055,-2.3460119162021873) -- cycle;
\draw [line width=2pt] (-4.3455,1.80544)-- (-4.2993772823371055,-2.3460119162021873);
\draw [line width=2pt] (-4.2993772823371055,-2.3460119162021873)-- (-1.2396040745723542,-0.28566211800883773);
\draw [line width=2pt] (-4.3455,1.80544)-- (-1.2396040745723542,-0.28566211800883773);
\draw [line width=2pt] (-1.2396040745723542,-0.28566211800883773)-- (2.404447434172601,0.34474341725927676);
\draw [line width=2pt] (2.404447434172601,0.34474341725927676)-- (5.310463194310983,1.989948106861429);
\draw [line width=2pt] (5.310463194310983,1.989948106861429)-- (5.141330001921974,-1.6387276571208862);
\draw [line width=2pt] (5.141330001921974,-1.6387276571208862)-- (2.404447434172601,0.34474341725927676);
\draw [line width=2pt,color=zzttqq] (5.310463194310983,1.989948106861429)-- (5.141330001921974,-1.6387276571208862);
\draw [line width=2pt,color=zzttqq] (5.141330001921974,-1.6387276571208862)-- (2.404447434172601,0.34474341725927676);
\draw [line width=2pt,color=zzttqq] (2.404447434172601,0.34474341725927676)-- (5.310463194310983,1.989948106861429);
\draw [line width=2pt,color=zzttqq] (-4.3455,1.80544)-- (-1.2396040745723542,-0.28566211800883773);
\draw [line width=2pt,color=zzttqq] (-1.2396040745723542,-0.28566211800883773)-- (-4.2993772823371055,-2.3460119162021873);
\draw [line width=2pt,color=zzttqq] (-4.2993772823371055,-2.3460119162021873)-- (-4.3455,1.80544);
\begin{scriptsize}
\draw [fill=ududff] (-4.3455,1.80544) circle (2.5pt);
\draw [fill=ududff] (-4.2993772823371055,-2.3460119162021873) circle (2.5pt);
\draw [fill=ududff] (-1.2396040745723542,-0.28566211800883773) circle (2.5pt);
\draw [fill=ududff] (2.404447434172601,0.34474341725927676) circle (2.5pt);
\draw [fill=ududff] (5.310463194310983,1.989948106861429) circle (2.5pt);
\draw [fill=ududff] (5.141330001921974,-1.6387276571208862) circle (2.5pt);
\end{scriptsize}
\end{tikzpicture}}}
	
	\caption{
	We consider three different constructions of filtered simplicial complexes with a fixed sample as vertex set.}
\end{figure*}

The main advantage of the Alpha complex is its small size in low dimensions~\cite{debergDelaunayTriangulations2008}; namely the Alpha complex on a random sample scales exponentially with the dimension of the sample and linearly with the sample size, see~\cite{edelsbrunnerExpectedSizesPoisson2017} for a further discussion. This is acceptable for low dimension, but impractical for higher ones. The Vietoris-Rips complex does not scale with the dimension but it scales exponentially with the sample size. For small samples in high dimensions, this construction should be preferred.

Counting the simplices with a sign yields the Euler characteristic, a fundamental topological invariant.
\begin{definition}
Let $K$ be a finite abstract simplicial complex.
Its \textit{Euler characteristic} is
\[
    \chi(K) = \sum\limits_{\sigma \in K}(-1)^{\dim(\sigma)}.
\]
\end{definition}

In the following we use the \v{C}ech construction in the theoretical part. Due to its sparse nature the Alpha construction is used in the implementation. They are topologically equivalent by the nerve lemma \cite[III.2]{edelsbrunnerComputationalTopologyIntroduction2010}, hence they give the same ECC.

It should be noted that, for a given sample $X$, the Euler characteristics of its Vietoris-Rips complex, $\chi(\mathcal{R}_r(X))$, may be different from $\chi(\mathcal{A}_r(X))$ and $\chi(\mathcal{C}_r(X))$. An example can be found in the sample presented in Figure~\ref{fig:VR} in which the 2-simplex (triangle) on the left is filled in \redsout{in }the Vietoris-Rips complex, but empty for the \v{C}ech and Alpha complex. 

Keeping track of how the Euler characteristic changes with the scale parameter yields the main tool of our interest:
\begin{definition}
    Given a finite subset $X\subseteq \R^d$, define its \textit{Euler characteristic curve (ECC)} as
    \[
        \chi(X)\colon [0,\infty) \to \Z, \; r \mapsto \chi(\mathcal{A}_r(X)).
        \]
\end{definition}
The ECC of the sample from Figure \ref{fig:PointCloud} is displayed in Figure \ref{fig:ECC_example_pointcloud}.
\begin{figure}
    \includegraphics[width = \linewidth]{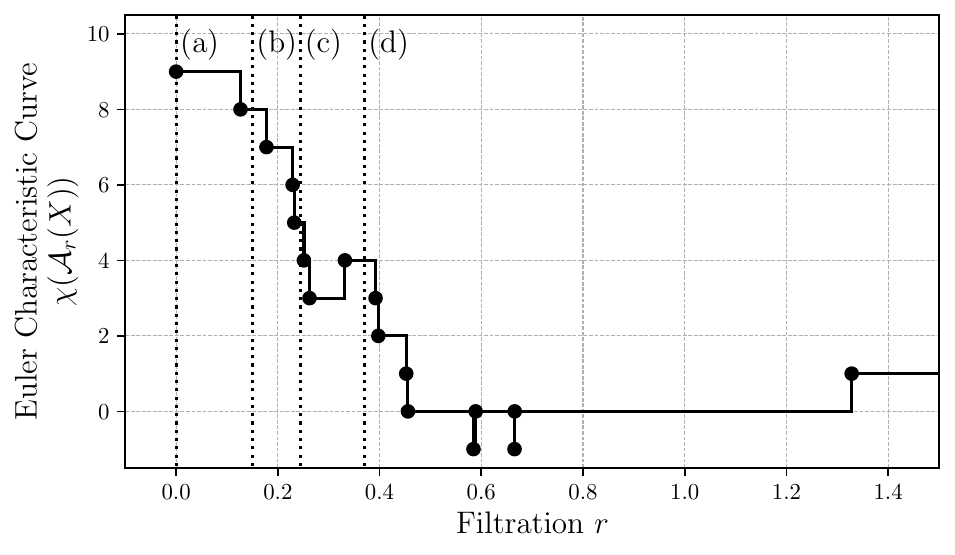}
    \caption{The ECC of the sample from Figure \ref{fig:PointCloud}. The filtration values (a)-(d) correspond to the complexes in Figure \ref{fig:PointCloud}-\ref{fig:VR3}.\label{fig:ECC_example_pointcloud}}
\end{figure}
\red{First applications of the ECC date to back to work of Worsley on astrophysics and medical imaging \cite{Worsley1996Geometry}.}

\subsubsection*{Topology of Random Geometric Complexes}
In the considered setting, the vertex set from which we build simplicial complexes \red{is} sampled from some unknown distribution.
The literature distinguishes two approaches, \textit{Poisson} and \textit{Bernoulli} sampling; see \cite{bobrowskiTopologyRandomGeometric2017} for a survey.
In the first setting, the samples are assumed to be generated by a spatial Poisson process.
We focus on the \textit{Bernoulli} sampling scheme in this paper. 
This means that we consider samples of $n$ points sampled i.i.d. from some $d$ dimensional distribution.
Furthermore, there are three regimes to be considered when the sample size goes to infinity \red{\cite[Section 1.4]{penroseRandomGeometricGraphs2003}}. We consider the geometric complex at scale $r\red{_n}$ \red{for a sequence $r_n\to0$ whose topology} is determined by the behaviour whether
\[
	n\cdot r\red{_n}^d \to \begin{cases}
		\infty,\\
		\lambda\in (0,\infty) \text{ constant},\\
		0.
	\end{cases}
\]
In the \textit{supercritical regime}, $n \cdot r\red{_n}^d \to \infty$, so that the domain gets densely sampled the geometric complex is highly connected. Intuitively, this regime maintains only global topological information and forgets about local density.
In the \textit{subcritical regime}, $n \cdot r\red{_n}^d \to 0$, so that the domain gets sparsely sampled and the geometric complex is, informally speaking, disconnected (consult~\cite{bobrowskiTopologyRandomGeometric2017} for details).
In this paper, we focus on the \textit{thermodynamic regime}, i.e. we keep the quantity $n \cdot r\red{_n}^d\red{ = \lambda}$ constant.
\red{
Up to a constant factor, the quantity $n\cdot r_n^d$ is the average number of points in a ball of radius $r_n$ \cite[Section 1]{bobrowskiTopologyRandomGeometric2017}.
This value neither goes to zero nor to infinity as $n\to \infty$ in the thermodynamic regime, leading to complex topology; see for instance \cite[Chapter 9]{penroseRandomGeometricGraphs2003}.
Now it is straightforward to observe that a subset of our sample $\sigma\subseteq X$ forms a simplex in the \v{C}ech complex at scale $r_n$ iff
\[
    \bigcap_{x\in\sigma} B_{r_n}(x) \neq \emptyset \Leftrightarrow \bigcap_{x\in n^{1/d}\sigma} B_\lambda(x) \neq \emptyset.
\]
This is because for any $x\in X, x'\in\R^d$, we have
\begin{align*}
    \|x'-x\| \leq r_n 
    &\Leftrightarrow n^{1/d}\|x'-x\|\leq n^{1/d}r_n \\
    &\Leftrightarrow \|n^{1/d}x'-n^{1/d}x\|\leq \lambda^{1/d}
\end{align*}
This observation motivates us to scale a sample of size $n$ by $n^{1/d}$.
In fact, this setup aligns with the approach of \cite{krebsApproximationTheoremsEuler2021}.
}
Due to this scaling, the average number of points in a ball of radius $r\red{=\lambda^{1/d}}$ stays the same \red{as we increase $n\to \infty$}.
Therefore, it makes sense to compare ECCs at fixed radius $r\red{=\lambda^{1/d}}$ for samples of different sizes\red{.
Visually speaking, we can compare (expected) ECCs from samples of different sizes in a common coordinate system using the $r$-axis scaled in this way.
In particular, one can study the point-wise limit of the expected ECC; that is, when the sample size approaches infinity for a fixed $r$.}
Moreover, this rescaling allows us to conduct two sample tests with samples of different sizes, cf. Section~\ref{sec:two_sample_test}.

\subsection{Previous Work}\label{subsec:PreviousWork}
Let us briefly review some related work on the intersection of topology and statistics.
 The most popular tool of TDA is \textit{persistent homology}.
Its key property is \textit{stability} \cite{cohen-steinerStabilityPersistenceDiagrams2007}; informally speaking, a small perturbation of the input yields a small change in the output.
However, persistent homology is a complicated setting for statistics; for example, there are no unique means \cite{turnerFrechetMeansDistributions2014}.

For a survey on the topology of random geometric complexes see \cite{bobrowskiTopologyRandomGeometric2017}\red{. A} text book for the case of one-dimensional complexes, i.e. graphs, is \cite{penroseRandomGeometricGraphs2003}.
The Euler characteristic of random geometric complexes has been studied in \cite{bobrowskiDistanceFunctionsCritical2014, bobrowskiTopologyProbabilityDistributions2013}.
Notably, in \cite{bobrowskiTopologyProbabilityDistributions2013}, the limiting ECC in the thermodynamic regime is computed for the uniform distribution on $[0,1]^3$.
More recently, \cite{thomasFunctionalLimitTheorems2020b} provided a functional central limit theorem for ECCs, which was subsequently generalized by \cite{krebsApproximationTheoremsEuler2021}.
The Euler characteristic has been studied in the context of random fields \cite{adler_random_2007} \red{ by Adler and Taylor.} Adler suggested to use it for model selection purposes and normality testing \cite[Section 7]{adler_new_2008}.
Building on this work, such a normality test has been extensively studied in \cite{bernardino_test_2017}. 
Using topological summaries for statistical testing has moreover been suggested by \cite{ciprianiTopologybasedGoodnessoffitTests2022a} for persistence vineyards, \cite{biscioTestingGoodnessFit2019} for persistent Betti numbers and \cite{botnanConsistencyAsymptoticNormality2021} for multiparameter persistent Betti numbers.
Mukherjee and Vejdemo-Johansson \cite{vejdemo-johanssonMultipleTestingPersistent2022} describe a framework for multiple hypothesis testing for persistent homology.
Very recently, Vishwanath et al. \cite{vishwanath_limits_2022} provided criteria to check the injectivity of topological summary statistics including ECCs.

\subsection{Our Contributions}
In this paper, to the best of our knowledge, we present the first mathematically rigorous approach 
using the Euler characteristic curves to perform general goodness-of-fit testing. Our procedure is theoretically justified by Theorem 
\ref{thm:TypeIIErrorAsymptotic}.
The concentration inequality for Gaussian processes \red{(Lemma \ref{lemma:lemma_bound})} might be of independent interest.

Simulations conducted in Section~\ref{sec:NumericalExperiments_one_sample} and~\ref{sec:NumericalExperiments2sample} indicate that TopoTest outperforms 
the Kol\-mo\-go\-rov\--Smirnov \red{test }we used as a baseline in arbitrary dimension both in terms of the test power
but also in terms of computational time for moderate sample sizes and dimensions. 

The implementation of TopoTest is publicly available at
\url{https://github.com/dioscuri-tda/topotests}.

\section{Method}
\label{sec:method}
\subsection{One-sample test}
While topological descriptors are computable and have a strong theory underlying them, they are not complete invariants of the underlying distributions, as recently pointed out in \cite{vishwanath_limits_2022}.
Hence the statement of the null hypothesis and the alternative require some care.
\begin{definition}\label{def:EulerEquivalent}
    We say two distributions $F,G$ are \textit{Euler equivalent}, denoted $F \chieq G$, if $\chi_F(t) \overset{D}{=} \chi_G(t)$ for all $t>0$.
\end{definition}
For instance, if $G$ arises from $F$ via translations, rotations or reflections, $F \chieq G$.
For a more interesting instance of Euler equivalent distributions, see example~\ref{ex:counterexample} below.

We aim to solve the following:
Given a fixed null distribution $F$ and a sample $X$ following an unknown distribution $G$, we test
\begin{equation}\label{eqn:EulerOneSampleHypothesis}
    H_0 \colon G \chieq F \qquad\text{vs.}\qquad H_1 \colon G \not\chieq F.
\end{equation}
Compare this formulation to the problem stated in~(\ref{eqn:OneSampleHypothesis}).
As the ECC of the Alpha and \v{C}ech complexes are equal, we will use them interchangeably.
We write
\[
    \chi(n,r) = \chi ( \mathcal{C}_r(X)),
\]
where $n$ is the cardinality of $X$. Given some distribution $F$ on $\R^d$ against which we want to test, we are interested in the expected ECC  of the \v{C}ech complex of scale $r$ of $n$ i.i.d. points drawn according to $F$, denoted as $\E_F(\chi(n,r))$. 
The TopoTest employs the supremum distance between the ECC computed based on sample points, $\chi(\mathcal{C}_r(X))$, and the expected ECC, $\E_F(\chi(n,r))$, under $H_0$, i.e. the test statistic is  

\begin{equation}
\label{eqn:Delta}
    \Delta_n := n^{-1/2}\supnorm{r} | \chi(\mathcal{C}_r(X)) - \E_F(\chi(n,r))|,
\end{equation}
where $T \in \mathbb{R}^+$. Therefore, by using ECC as topological summary of 
the dataset we reduce the initial $d$-dimensional problem to a one-dimensional setting.
If $\Delta_n$ defined in (\ref{eqn:Delta}) is large enough the null hypothesis is rejected, while for small values of $\Delta_n$ the test fails to reject the $H_0$. More precisely: given the significance level $\alpha$ we consider a rejection region
$R_\alpha = [t_\alpha, \infty)$ such that

\begin{equation}
\label{eqn:threshold_definition}
\begin{split}
    &\P(\Delta_n \in R_\alpha \vert H_0) \\&= \P\left(\left.n^{-1/2}\supnorm{r} | \chi(\mathcal{C}_r(X)) - \E_F(\chi(n,r))| >t_\alpha \right\vert H_0 \right)\\&=\alpha.  
\end{split}
\end{equation}

\red{
The threshold value $t_\alpha$ depends on the significance level $\alpha$
and $F$ (and hence also on dimension $d$),
however the dependence on $F$ is dropped in the notation.
}
We prove that this test is consistent below in Section \ref{subsec:power}.
\begin{remark} 
The test statistic (\ref{eqn:Delta}) is based on the difference between sample ECC and ECC expected under $H_0$.
A natural, yet still open question, arises; how likely it is that two isometry-nonequivalent distributions will be Euler-equivalent and hence indistinguishable for test statistics (\ref{eqn:threshold_definition}).
In a naive search where we considered over 1000 different univariate probability distributions defined in $\mathbb{R}_+$ we could not find any such example. 
Therefore we believe that, the Euler-equivalence is not a practical limitation of our method. 
\end{remark}

\subsection{Two-sample test}
\label{sec:two_sample_test}
A test statistic based on the Euler characteristic curve can also be adapted to the two-sample problem.
Given two samples $X, Y \subset \R^d$ of possibly different sizes, following unknown distributions 
$X\sim F$ and $Y\sim G$, we are testing the null hypothesis $H_0\colon G \chieq F$.
The test statistic in this setting is the supremum distance between the normalized ECCs
\[
    \Delta(\chi(X), \chi(Y)) = \supnorm{r} \left\vert \frac{1}{|X|} \chi(\mathcal{A}_r(X)) - \frac{1}{|Y|} \chi(\mathcal{A}_r(Y))  \right\vert.
\]
Moreover, recall that we rescale the samples to have a fixed average number of points in a ball of radius $r$, independently of the sample size.
Since the null distribution is unknown, we fall back on a permutation test \cite[Section 16.3]{arias2022principles} to compute the $p$-value, see Algorithm \ref{algo:2SampleTesting} for the details. 

As for any permutation test, the procedure is computationally expensive as it requires computing ECCs 
for a variety of point sets resampled from the union of the two input datasets. The application of this approach 
is therefore limited to rather small sizes of input data set\red{s}. 
See Section~\ref{sec:NumericalExperiments2sample} for results of a simulation study in which the performance of this approach is compared with the two-sample Kolmogorov-Smirnov test. 

\subsection{Power of the one-sample test}\label{subsec:power}
\subsubsection*{Overview}
The TopoTest relies on the Functional Central Limit Theorem of Krebs et al. \cite[Theorem 3.4]{krebsApproximationTheoremsEuler2021}, hence it works under the following, rather technical, assumption
\begin{assumption}\label{assumption}
The null distribution has compact convex support inside $[0,1]^d$.
It admits a bounded density $\kappa$ that can be uniformly approximated by blocked functions $\kappa_n$.
\end{assumption}
Recall from \cite[equation 3.8]{krebsApproximationTheoremsEuler2021}, that the approximation by blocked functions means $\lim\limits_{n\to \infty} \|\kappa -\kappa_n\| = 0$, where each $\kappa_n$ is constant on grid elements of a partition of the unit hypercube $[0,1]^d$ into an equidistant grid of $m^d$ subcubes.
In particular, bounded measurable functions satisfy this assumption. 

We will show, for a fixed significance level $\alpha$, that the mean of the test statistic $\Delta_n$ does not grow with $n$ under the null hypothesis, while it grows at least like $\sqrt{n}$ under the alternative hypothesis. Moreover, in both cases $\Delta_n$ is concentrated around its mean allowing to control the type II error of the TopoTests.
\subsubsection*{Case $H_0$ true}
By \cite{thomasFunctionalLimitTheorems2020b} and \cite[Theorem 3.4]{krebsApproximationTheoremsEuler2021}), 
we have convergence of 
$\Delta_n$ in distribution in the Skorokhod $J_1$-topology to a centered Gaussian process $f_r$,
\begin{equation}
\label{eqn:GaussianProcess}
    n^{-1/2} \left( \chi(\mathcal{C}_r(X)) - \E_F(\chi(n,r)) \right) \xrightarrow[n \rightarrow \infty]{D} f_r.
\end{equation}
Here it is assumed that the sample is drawn from a distribution satisfying Assumption~\ref{assumption} and scaled by $n^{1/d}$. 
Let us denote
\[
Z_T = \supnorm{r} |f_r|.
\]
In the following we will approximate the finite-sample distribution of 
$n^{-1/2} (\chi(\mathcal{C}_r(X)) - \E_F(\chi(n,r)))$ 
by the limiting Gaussian process $f_r$. 
Therefore, \red{for} sufficiently large $n$ we assume that
\begin{equation}
\label{eqn:finite_sample_appoximation}
\Delta_n \stackrel{D}{=} Z_T.
\end{equation}
The quality of this approximation was studied numerically -- please refer to Figure \ref{fig:finite_sample_appoximation}.

For $Z_T$ we have the Borell-TIS inequality{\footnote{\red{The abbreviation stands for Tsirelson, Ibragimov, and Sudakov, who discovered the inequality independently of Borell.}}\cite[Section 2.1]{AdlerGaussianInequalities2007},
\begin{equation}
\label{eqn:borel_tis}
\begin{split}
    \P\left(Z_T > t \right) &= \P\left( \supnorm{r} |f_r| > t\right) \\
    &\le \exp\left(-\left[t-\E\left(\supnorm{r} \vert f_r\vert\right)\right]^2 / 2\sigma_T^2\right),
\end{split}
\end{equation}
where $\sigma_T^2 = \supnorm{r} \E(f_r^2)$.
\begin{figure*}
\label{fig:finite_sample_appoximation}
    \includegraphics[width = 0.32\linewidth]{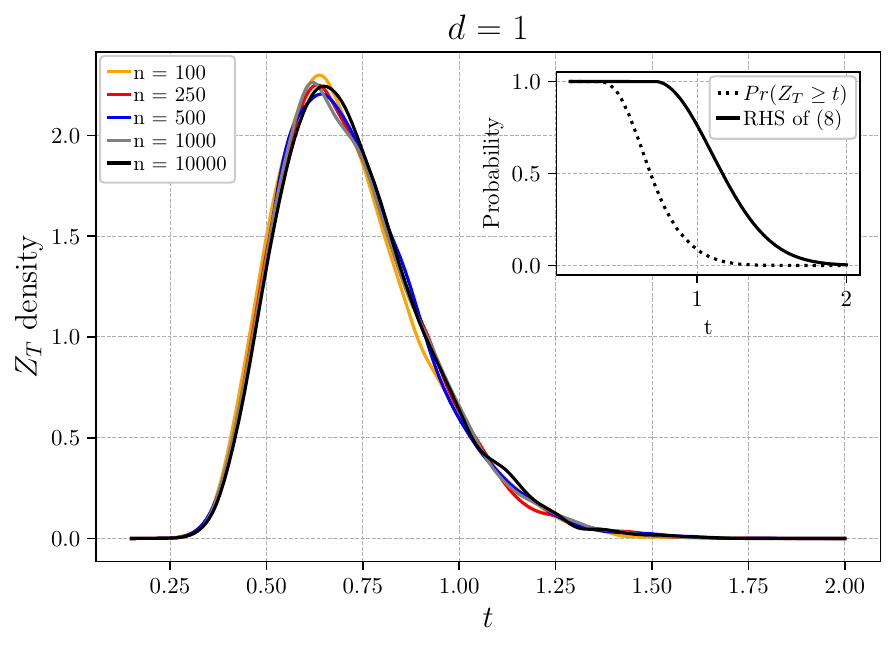}
    \includegraphics[width = 0.32\linewidth]{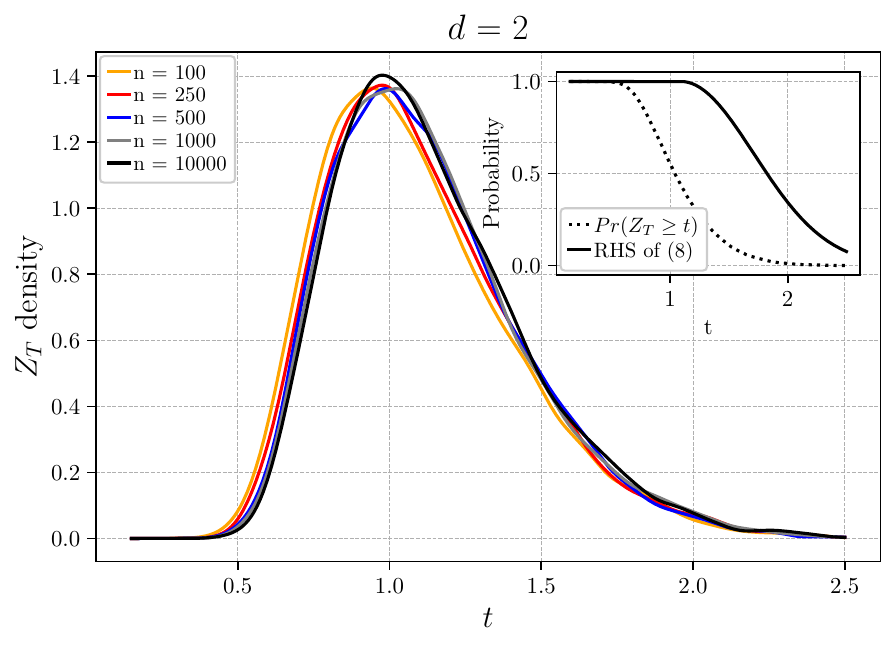}
    \includegraphics[width = 0.32\linewidth]{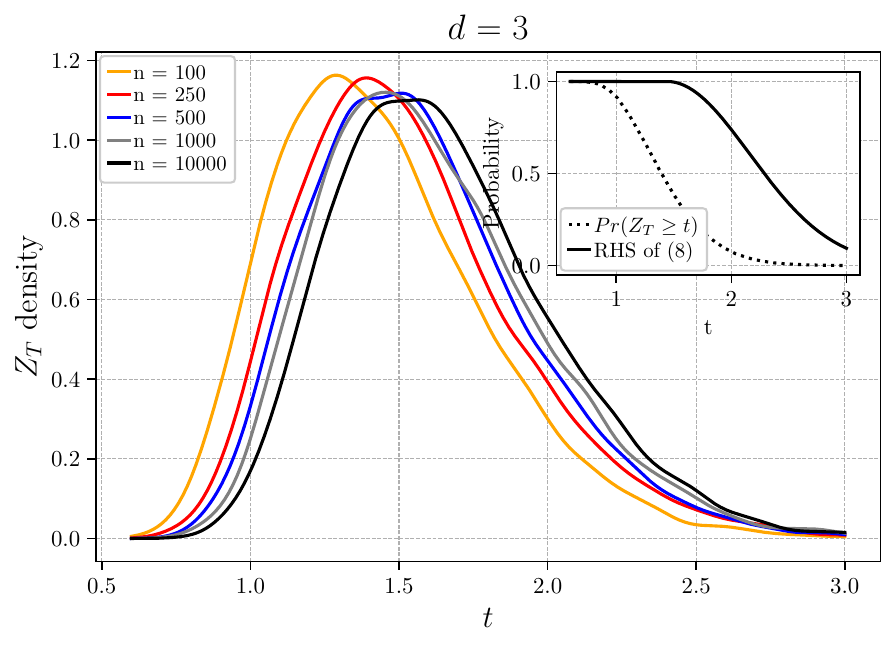}
    \caption{Numerical inspection of the quality of finite sample approximation (\ref{eqn:finite_sample_appoximation}). \red{The empirical} distribution of $Z_T$ converges with the increasing sample size. Even for three-dimensional case the distribution obtained for $n=100$ 
    is a reasonable approximation for large-sample empirical distribution.
    An inset in each plot shows \red{left- and right-hand side} of the inequality (\ref{eqn:borel_tis}) \red{--this provides another justification for approximation (\ref{eqn:finite_sample_appoximation}}).
    }
\end{figure*}

Therefore, for $n$ large enough,
\begin{align*}
& \P\left( \Delta_n > t |H_0\right) \\
 & =
        \P\left(\supnorm{r} \left|\frac{ \chi(\mathcal{C}_r(X)) - \E_F(\chi(n,r))}{\sqrt{n}}\right| >t\right)\\ \nonumber
        & \le \exp\left(-\left[t-\E\left(\supnorm{r} \frac{1}{\sqrt{n}} | \chi(\mathcal{C}_r(X)) \right.\right.\right.\\
        & \left.\left.\left.\phantom{\exp \frac{1}{\sqrt{n}}t-E-}- \E_F(\chi(n,r))|\right)\right]^2 / 2\sigma_T^2\right).
\end{align*}
Plugging in (\ref{eqn:threshold_definition}) yields 
\begin{equation*}
\begin{split}
    \alpha \le \exp\left(-\left[t_\alpha-\E\left(\supnorm{r} \frac{1}{\sqrt{n}}| \chi(\mathcal{C}_r(X)) \right.\right.\right.\\
    \left.\left.\left.\phantom{\frac{1}{\sqrt{N}}-t-E}- \E_F(\chi(n,r))|\right)\right]^2 / 2\sigma_T^2\right)
    \end{split}
\end{equation*}
which leads to
\begin{equation}
\label{eqn:threshold}
\begin{split}
        t_\alpha \le &\sqrt{-2 \sigma_T^2 \ln \alpha)} \\
        &+ \E\left(\supnorm{r} \frac{1}{\sqrt{n}} | \chi(\mathcal{C}_r(X)) - \E_F(\chi(n,r))|\right),
\end{split}
\end{equation}
i.e. $t_\alpha=O(1)$. 
\subsubsection*{Case $H_0$ false}
Now let us study the asymptotic size of 
\[
    \supnorm{r} | \chi(\mathcal{C}_r(Y)) - \E_F(\chi(n,r))|
\]
as $n \rightarrow \infty$ when $Y\sim G, G \not\chieq F$.

We have
\begin{align*}
 \E&\left(\supnorm{r} | \chi(\mathcal{C}_r(Y)) - \E_F(\chi(n,r))|\right)\\
    &\ge \supnorm{r} \E \left\vert \chi(\mathcal{C}_r(Y)) - \E_F(\chi(n,r))\right\vert \\
    &\ge  \supnorm{r} | \E_G(\chi(n,r)) - \E_F(\chi(n,r))|.
\end{align*}

Because the limiting distributions of the ECCs are different under the alternative hypothesis, this last expression diverges.
Due to \cite{bobrowskiTopologyProbabilityDistributions2013}, Corollary 4.5, $E_F(\chi(n,r)) \sim n$ with constant depending on $F$ and $d$. 
In our setting, we obtain
\begin{equation}
\label{eqn:E_omega_n}
    \E\left(\supnorm{r} n^{-1/2} | \chi(\mathcal{C}_r(Y)) - \E_F(\chi(n,r))|\right) = \Omega(\sqrt{n}).
\end{equation}

To complete the discussion, it is required to show that in the case of $H_0$ false, one also has a 
concentration around the mean, i.e. one needs to control

\begin{equation}
    \label{eqn:mean_concentration}
\begin{split}
    C_{F,G}(t) = \P\left(n^{-1/2}\left|\supnorm{r} | \chi(\mathcal{C}_r(Y)) - \E_F(\chi(n,r))|\right.\right.\\
    \left.\left. - \E\left(\supnorm{r} | \chi(\mathcal{C}_r(Y)) - \E_F(\chi(n,r))|\right)\right| > t\right).
\end{split}
\end{equation}

The lemma below provides a generalization of the Borell-TIS inequality to the case of non-centred Gaussian process.
\begin{lemma}
\label{lemma:lemma_bound}
Let $f_r$ be a centred Gaussian process and $g(r)$ some deterministic function. We have
\begin{equation}
\label{eqn:lemma_bound}
\begin{split}
    \P\left(\left| \sup_{r \in [0,T]} |f_r + g(r)| - \E\left(\sup_{r\in[0,T]} |f_r+g(r)|\right) \right| > t \right) \\ \le 2e^{-t^2/2\sigma^2},
\end{split}
\end{equation}
where $\sigma = \sup_{r\in[0,T]} \left( \E[f_r^2] \right)^{1/2}$.
\end{lemma}
\begin{proof}
We follow the strategy of Ledoux \cite[Section 7.1]{ledouxConcentrationMeasurePhenomenon2005}.
Argument (2.35) in Ledoux \cite{ledouxConcentrationMeasurePhenomenon2005} yields that if $\gamma$ is a standard Gaussian measure on $\R^n$ then for every 
1-Lipschitz function $F$ on $\R^n$ and $t\ge0$ we have
\begin{equation}
\gamma\left(\left\{F \ge \int F d\gamma + t \right\}\right) \le e^{-t^2/2}.
\label{eqn:1lipschitz}    
\end{equation}
Let $r_1,\ldots,r_n$ be fixed in $[0,T]$ and consider \red{centered} Gaussian random vector $(f_{r_1}, \ldots, f_{r_n})$ in $\R^n$
with covariance matrix $\Gamma = B^TB$. Consequently, the law of $(f_{r_1}, \ldots, f_{r_n})$ is the same
as the law of $B\mathcal{N}$ where $\mathcal{N}=(N_1, \ldots, N_n)^T$ is distributed according to the standard Gaussian measure $\gamma$ on $\R^n$. 
Let $F: \R^n \rightarrow \R$ be defined as 
\[
F(x) = \max_{1 \le i \le n}\left| (Bx)_i + g(r_i) \right|, x\in\R^n.
\]
Although we have a different $F$ in our setting than \cite{ledouxConcentrationMeasurePhenomenon2005}, we can still bound the Lipschitz norm of $F$ to be at most the operator norm of $B\colon (\R^n, \|\cdot\|_2) \to (\R^n,\|\cdot\|_\infty) $.
Indeed, consider any $c>0$ such that $\|Bx\|_{\infty}\leq c \|x\|_2$ for all $x\neq 0$.
Using the triangle inequality, we estimate that for any $x\neq y \in \R^n$,
\begin{align*}
    \vert F(x) - F(y)\vert 
    & = \left\vert \max\limits_{1\leq i\leq n} |(Bx)_i+g(r_i)|\right. \\ &\phantom{\leq |}\left.- \max\limits_{1\leq i\leq n} |(By)_i+g(r_i)|\right\vert\\
    & \leq \max\limits_{1\leq i\leq n} \left\vert  (Bx)_i+g(r_i) - (By)_i-g(r_i)\right\vert\\
    & = \max\limits_{1\leq i\leq n} \left\vert  (B(x-y)_i\right\vert\\
    & \leq c \|x-y\|_2.
\end{align*}
Notice that $f_{r_i}=\sum_{j=1}^n B_{ij}N_j$ and by independence of $\{N_j\}_{1\leq j\leq n}$ we have $\E f_{(r_i)}^2 = \sum_{j=1}^n B_{ij}^2$. 
This allows us to bound the operator norm of $B$ as follows:
\begin{align*}
\|B\|_{op} &= \max_{1 \le i \le n} \left(\sum_{j=1}^n B_{ij}^2 \right)^{1/2}
    = \max_{1 \le i \le n}\left(\E(f_{(r_i)}^2)\right)^{1/2}\\
    &\le \sup_{r\in[0,T]} \left( (\E(f_{(r_i)}^2) \right)^{1/2}
    = \sigma.
\end{align*}

Consequently, $F/\sigma$ is 1-Lipschitz and by (\ref{eqn:1lipschitz}) we have
\[
\P\left(\frac{1}{\sigma}F(\mathcal{N}) - \E[\frac{1}{\sigma}F(\mathcal{N})] \ge \tilde{t} \right) \le e^{-\tilde{t}^2/2}
\]
Letting $t=\sigma\tilde{t}$ and by symmetry argument we obtain
\[
\P\left(|F(\mathcal{N}) - \E(F(\mathcal{N}))| \ge t \right) \le 2e^{-t^2/2\sigma^2}
\]
and
\begin{equation*}
\begin{split}
\P\left(\left|\sup_{1\le i \le n}|f_{r_i} + g(r_i)| - \E\left(\sup_{1 \le i \le n} |f_{r_i}+g(r_i)|\right) \right| \ge t \right) \\
\le 2e^{-t^2/2\sigma^2}.
\end{split}
\end{equation*}
The right hand side does not depend on $f(r_i)$, hence letting $n \rightarrow \infty$, inequality (\ref{eqn:lemma_bound}) is obtained. 
\end{proof}

Using the Lemma \ref{lemma:lemma_bound} we obtain following theorem
\begin{theorem}
\label{the:concentration}
Concentration around the mean $C_{F,G}(t)$, defined in (\ref{eqn:mean_concentration}), is exponentially bounded
\begin{equation}
    C_{F,G}(t) \le 2e^{-t^2/2\sigma_{G}^{2}}.
    \label{eqn:mean_concentraton_bound}
\end{equation}
\end{theorem}
\begin{proof}
Subtracting and adding $\E_G(\chi(n,r))$ in (\ref{eqn:mean_concentration}) yields

\begin{align*}
& C_{F,G}(t) \\
&= \P\left(n^{-1/2}\left|\supnorm{r} | \chi(\mathcal{C}_r(Y)) - \E_F(\chi(n,r))| \right.\right.\\
&\left.\left.\phantom{=P(}- \E\left(\supnorm{r} | \chi(\mathcal{C}_r(Y)) - \E_F(\chi(n,r))|\right)\right| > t\right) \\
      & =  \P\left(n^{-1/2}\left|\supnorm{r} | \chi(\mathcal{C}_r(Y)) - \E_G(\chi(n,r))\right.\right. \\
      &  \left.\left.\phantom{=P(n^{-1/2}|\sup\limits_{T} }+ \E_G(\chi(n,r)) - \E_F(\chi(n,r))| \right.\right.\\
      &  \left.\left. \phantom{=P(}- \E\left(\supnorm{r} | \chi(\mathcal{C}_r(Y)) - \E_G(\chi(n,r)) \right.\right.\right.\\
      & \left.\left.\left.\phantom{=P(-E\sup\limits_{T}}+ \E_G(\chi(n,r))- \E_F(\chi(n,r))|\right)\right| > t\right) \\
      & =  \P\left(\left|\supnorm{r} |g_r + h(r)| \right.\right.\\
      &\left.\left. \phantom{=P(|}- \E\left(\supnorm{r}|g_r+h(r)|) \right)\right| > t\right),
\end{align*}
where the notation
\begin{align*}
 g_r & = \left(\chi(\mathcal{C}_r(Y)) - \E_G(\chi(n,r))\right)/\sqrt{n},\\
 h(r) & = \left(\E_G(\chi(n,r)) - \E_F(\chi(n,r))\right)/\sqrt{n}
\end{align*}
was introduced.
Note that by (\ref{eqn:GaussianProcess}) applied for distribution $G$ the $g_r$ converges to 
a centred Gaussian process, whereas $h(r)$ is a deterministic function. 
Let $\sigma_G^2 = \supnorm{r} \E(g_r^2)$.
Therefore using the same argument as in (\ref{eqn:finite_sample_appoximation})
by Lemma \ref{lemma:lemma_bound} bound (\ref{eqn:mean_concentraton_bound}) is obtained.
\end{proof}

The rate of type I error is controlled by the significance level $\alpha$. 
An asymptotic upper bound for type II error is given by the following theorem.

\begin{figure}
    \centering
    \resizebox{\linewidth}{!}{
    \begin{tikzpicture}[domain=0:16]
        \draw[->] (0,0) -- (16.3,0);
        \draw[->] (0,0) -- (0,4);
        \node[above, color=blue] at (13,3) {pdf of $\Delta_n$ under $H_1$};
        \fill[fill=blue!20, domain=8:11] (8,0) -- plot (\x,{3*exp((-(\x-13)^2)/4 )}) -- (11,0) -- cycle;
        \draw[color=blue, samples=100]   plot (\x,{3*exp((-(\x-13)^2)/4 )});    
        \draw[color=red] (0,0) .. controls (0.5,8) and (0.5,0) .. (16,0.05);
        \node[above, color=red] at (4,3) {pdf of $\Delta_n$ under $H_0$};        
        \draw[dashed] (13,3) -- (13,-0.1) node[below] {$\E = \Omega(\sqrt{n})$};
        \draw (11,0.35) -- (11,-0.1);
        \draw[dotted] (11,1.1) -- (11,-0.1) node[below] {$t_\alpha = O(1)$};
        \draw [<->] (11,1) -- (13,1);
        \node[above] at (12,1) {$ t_{\alpha,n}^{*}$};        
    \end{tikzpicture}
    }
    \caption{The area of shaded blue region is the probability of a type II error occuring.
    As $n \to \infty$, it goes to zero.}
    \label{fig:ProbabilityType2Error}
\end{figure}
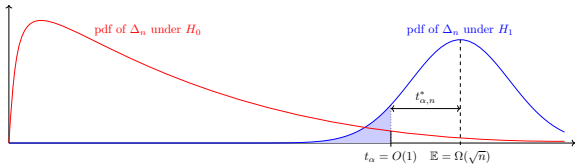
\begin{theorem}\label{thm:TypeIIErrorAsymptotic}
    For fixed $\alpha$, the probability of a type II error goes to 0 exponentially as $n \rightarrow \infty$.
\end{theorem}
\begin{proof}
We will use the threshold $t_\alpha$ defined in (\ref{eqn:threshold_definition}) 
and the concentration inequality of Theorem \ref{the:concentration}.
The idea is illustrated in Figure~\ref{fig:ProbabilityType2Error}.
Introduce 
\[
t_{\alpha,n}^{*} = \E\left(\supnorm{r}n^{-1/2}\left|\chi(\mathcal{C}_r(Y)) - \E_F(\chi(n,r))\right|\right) - t_\alpha.
\]

Due to equation (\ref{eqn:E_omega_n}) first term above is $\Omega(\sqrt{n})$ while second term is $O(1)$, therefore $t_{\alpha,n}^{*} = \Omega(\sqrt{n})$ and is positive for sufficiently large $n$. Hence we can estimate
\begin{align*}
    \P&(\text{type II error}) \\
    &\le  \P \left(\supnorm{r}n^{-1/2}\left|\chi(\mathcal{C}_r(Y)) - \E_F(\chi(n,r))\right| < t_\alpha \right) \\
& = \frac{1}{2}\P\left[n^{-1/2}\left|\E\left(\supnorm{r}\left|\chi(\mathcal{C}_r(Y)) - \E_F(\chi(n,r))\right|\right) \right.\right. \\
& \left.\left. \phantom{=\frac{1}{2} P ( }- \supnorm{r}\left|\chi(\mathcal{C}_r(Y)) - \E_F(\chi(n,r))\right|\right| > t_{\alpha,n}^{*} \right] \\
& \le  \exp\left (\frac{-{t_{\alpha,n}^{*}}^2}{2\sigma^2}\right) \sim e^{-n} \rightarrow 0.
\end{align*}
\end{proof}

\subsection{Properties of the TopoTests}
\label{sec:properties}
TopoTests rely on the Euler characteristics curve which is computed based on the Alpha complex of the input sample. The Alpha complex captures distance pattern between all data points in the samples.
\red{Therefore, TopoTest is not capable to discriminate distributions that are isometry equivalent, e.g. differ only by  translation, reflection or rotation.}
As a consequence TopoTest, contrary to Kolmogorov-Smirnov,
is not able to distinguish between 
e.g. $\mathcal{N}\left((0, 0),  \begin{bmatrix} 1 & 0 \\ 0 & 1 \end{bmatrix} )\right)$ from 
$\mathcal{N}\left((\mu_1, \mu_2),  \begin{bmatrix} 1 & \alpha \\ \alpha & 1 \end{bmatrix} )\right), \alpha \in [-1, 0) \cup (0, 1]$ as those distributions are equivalent up to translation and rotation. 
As a consequence, the alternative hypotheses in Kolmogorov-Smirnov and TopoTest 
are in fact slightly different: in the former we have $H_1: G \ne F$ while in later the inequality is understood
only up to Euler equivalence, cf. Equation~(\ref{eqn:EulerOneSampleHypothesis}). 
The same discussion also applies to the null hypothesis.
Hence, such pairs of distributions were excluded from the forthcoming numerical study.

\subsection{Non-compactly supported distributions}
\label{sec:what_to_do_in_non_compact_case}

The results on the asymptotic convergence presented in Section~\ref{subsec:power} work for compactly supported distributions. However, most of the distributions considered in practice, starting from normal distributions, are defined on non--compact support and the presented results \red{do} not apply to them directly. There are a number of ways we can adjust such a distribution so that the presented methodology applies. In what follows we discuss three possible strategies, starting from the one we consider the most practical one

\begin{enumerate}
\item \emph{Restricting a distribution to a compact subset;}\newline
In this case, the given distribution is restricted to a compact rectangle. In our case we choose a symmetric rectangle $[-a,a]^d$ for $a$ being the maximal representable double precision number. This ensures that every sample that can be analyzed in a computer is automatically coming from such a restricted distribution. We note that, formally, such a restricted distribution need to be rescaled to become a probability distribution. However, in all practically relevant cases we are aware of, such a restricted distribution will be infinitesimally close, on its domain, to the original one, defined on an unbounded domain. Therefore, we argue that in practice, the presented methods can be applied even to distributions with no compact support. Additionally, the simulations performed provide strong evidence for this claim.
\item \emph{Rescaling a distribution to a compact subset;}\newline
Here a transformation, $\arctan(\gamma x) : \mathbb{R} \rightarrow [-\frac{\pi}{2},\frac{\pi}{2}]$ is applied separately to each coordinate to map the unbounded domain to a compact region\red{.}

We observe that for $x \in [-2,2]$, or for any similar interval centered around zero, $\arctan(x)$ is close to a linear function, hence the distance between points before and after applying the map, should be proportional to each other regardless of the points. To keep such a distortion of distances between points before and after rescaling, the scaling parameter $\gamma$ is used. 
For instance, we may choose it in the way that $10$ standard deviations in our data, after divided by $\gamma$, have values in the interval $[-2,2]$. For multivariate distributions the scaling can be applied separately in each dimension.
Such a rescaling does not have any major impact on the powers of the tests as discussed in Sections~\ref{sec:NumericalExperiments_one_sample} and~\ref{sec:NumericalExperiments2sample}. At the same time, it allows to map any unbounded distribution to a compact domain. One should note, however, that a bounded distribution, transformed by $\arctan$ may be, in some pathological cases, unbounded. Hence, before using this transformation, the \red{boundedness} of the output distribution needs to be verified.

\emph{Transforming into copula;} \newline
The marginals $F_1, \ldots, F_d$ of the distribution $F$ are continuous, hence one can apply the probability integral transform~\cite{castella2002} to each component of the random vector $X$ sampled form a distribution $F$. Then the random vector 
\begin{equation}
\label{eqn:copula}
(U_1, \ldots, U_d) = (F_1(X_1), \ldots, F_d(X_d))
\end{equation}
is supported on a unit cube $[0, 1]^d$ and has uniformly distributed marginals. The joint distribution function of $(U_1, \ldots, U_d)$ forms a copula. 
Since the null distribution $F$ is given, the marginal distributions $F_1, \ldots, F_d$ can be derived. The transformation (\ref{eqn:copula}) must be applied to both the sample and null distribution $F$.
Transformation (\ref{eqn:copula}) preserves the correlation structure and transforms the initial distribution $F$ onto a compact support fulfilling the Assumption~\ref{assumption}.  Although such transformation is easy to compute and quite general, simulation studies showed that the power of resulting test is significantly reduced.
\end{enumerate}

\section{Algorithms}
\label{sec:algorithms}

\subsection{One-sample test}
The test statistic for one-sample TopoTest, $\Delta$ defined in (\ref{eqn:Delta}), involves $\E_F(\chi(n,r))$ being 
the ECC expected under $H_0$. There is no compact formula that can be applied to compute $\E_F(\chi(n,r))$ for an arbitrary
distribution function $F$ in arbitrary dimension $d$ although some formulas are available in case of the multivariate 
uniform distribution~\cite{bobrowskiTopologyProbabilityDistributions2013}. 
However one can use the approximation of $\E_F(\chi(n,r))$ based on average ECC computed on \red{a} collection of randomly generated ECCs. 
Notice that $\chi(\mathcal{C}_r(X))$ can only take on finitely many values because the underlying sample 
is finite. Therefore, $\E_F(\chi(n,r))$ is finite. 
The strong law of large numbers applies and we can approximate this expectation empirically, \red{i.e.} let $Y_1,\ldots Y_M$ be i.i.d. samples each consisting of $n$ points drawn i.i.d. from $F$, then
\begin{equation}
\label{eqn:estimated_E}
    \widehat{\E}_F(\chi(n ,r)) := \sum\limits_{i=1}^{M} \frac{\chi(\mathcal{C}_r(Y_i))}{M} \xrightarrow[M\to \infty]{a.s.} \E_F(\chi(n,r)).
\end{equation}
Due to the continuous mapping theorem, the above point-wise convergence result allows us to use an empirical estimate $\widehat{\E}_F(\chi(n ,r))$ 
instead of $\E_F(\chi(n,r))$ in practice when computing the statistic $\Delta_n$  leading to statistic
\begin{equation}
\label{eqn:Delta_hat}
\begin{split}
    \widehat{\Delta}_n &:= \widehat{\Delta}(\chi(\mathcal{C}(X), \widehat{E}_F(\chi(n,r))))\\
    &:=  \supnorm{r} \frac{1}{\sqrt{n}} | \chi(\mathcal{C}_r(X)) - \widehat{\E}_F(\chi(n,r))|,
\end{split}
\end{equation}
that was actually used in simulations. 
It should be mentioned that the estimator $\widehat{\E}_F(\chi(n ,r))$ does not depend on the sample being tested
and by increasing $M$ can be arbitrary close to $\E_F(\chi(n,r))$.

The algorithm for computing the TopoTest for one sample can be divided into two steps.
Firstly, in the \textit{preparation step} an average ECC for given null distribution $F$ is computed.
Then the critical value of the test statistic is estimated empirically by drawing a set of random samples from $F$ 
and computing the distance between ECCs corresponding to those samples and the average ECC computed previously.
Secondly, in the \textit{testing step}, the distance of the ECC of the given sample to the averaged 
ECC for the considered distribution is computed and compared to the critical values obtained in the first step.
This procedure is provided in details by Algorithm \ref{algo:1SampleTesting}.
\begin{algorithm}
    \caption{Algorithm for one-sample testing}\label{algo:1SampleTesting}
    \SetAlgoLined
    \KwInput{point sample $X\in \R^d$, null distribution $F$, significance level $\alpha$, \red{$M$:} number of samples draw from $F$ to estimate average ECC, \red{$m$:} number of samples draw \red{from} $F$ to
    estimate the threshold value.}
    \KwOutput{\red{Rejecting or failure to reject} of null hypothesis, $p$-value}
    Let $n = |X|$\\
        \tcc{"Preparation", i.e. determine the threshold $t_\alpha$ for \red{rejecting} the null hypothesis}
    \For{$i \gets 1,\ldots, M$}
    {
        $Y_i \gets$ i.i.d. sample of $n$ points from $F$\\
        Compute the ECC $\chi(\mathcal{A}(Y_i))$\\
    }
        Compute the average ECC $\overline{\chi}(t) \gets \frac{1}{M}\sum\limits_{i=1}^{M} \chi(\mathcal{A}_t(Y_i))$\\
    \For{$i \gets 1,\ldots, m$}
    {
        $Y'_i \gets$ i.i.d. sample of $n$ points from $F$\\
        Compute the ECC $\chi(\mathcal{A}(Y'_i))$\\
        Compute the deviation from average $ \Delta_i \gets \sup\limits_{t} \frac{1}{\sqrt{n}}\vert \chi(\mathcal{A}_t(Y'_i)) - \overline{\chi}(t)\vert$\\
    }   
    Let $t_{\alpha}\in\R$ such that $\#\{\Delta_i > t_{\alpha}\} <\alpha m$\\

    \tcc{"Testing", i.e. compare the threshold value with sample distance}
    Compute the ECC $\chi(\mathcal{A}(X))$\\
    $\Delta(\chi(\mathcal{A}(X)), \overline{\chi}) \gets \sup\limits_{t} \frac{1}{\sqrt{n}}\vert \chi(\mathcal{A}_t(X)) - \overline{\chi}(t)\vert$\\
    $pv\gets \frac{1}{M} \#\{\Delta_ i > \Delta(\chi(\mathcal{A}(X)), \overline{\chi})\}$\\
    \Return{$\Delta(\chi(\mathcal{A}(X)), \overline{\chi}) < t_\alpha$, pv}
\end{algorithm}

\begin{remark}
The \textit{preparation step} in Algorithm \ref{algo:1SampleTesting} depends only on sample size $n$ and
null distribution $F$ but is independent of actual sample $X$. Hence needs to \red{be} performed only once if several data samples \red{of size} $n$ are considered.
\end{remark}

\begin{remark}
The threshold value $t_{\alpha}$ used in the TopoTest is obtained from a numerical Monte Carlo simulation performed for a family of finite samples of a size $n$ and does not explicitly employ asymptotic bounds from Section~\ref{sec:method}.
\end{remark}

\begin{remark}
The Monte Carlo parameters $M$ and $m$ should be sufficiently large to obtain an accurate resulting test. For the distributions considered in this paper, values $M=m=1000$ were selected.
\end{remark}

\red{
\begin{remark}
The need to utilize the Monte Carlo approach to determine threshold value $t_\alpha$ stems from the fact that the distribution of the test statistic (\ref{eqn:Delta}) depends on the distribution of $F$ and the size of the samples for which TopoTest was built. In general, this distribution is unknown. 
The simulations showed that employing an asymptotic distribution, approximated numerically by using a large sample size $n$ in the preparation step, provided incorrect empirical significance levels
in case of samples much smaller than $n$.
\end{remark}
}

\begin{example}
 \begin{figure}
 	\includegraphics[width=\linewidth]{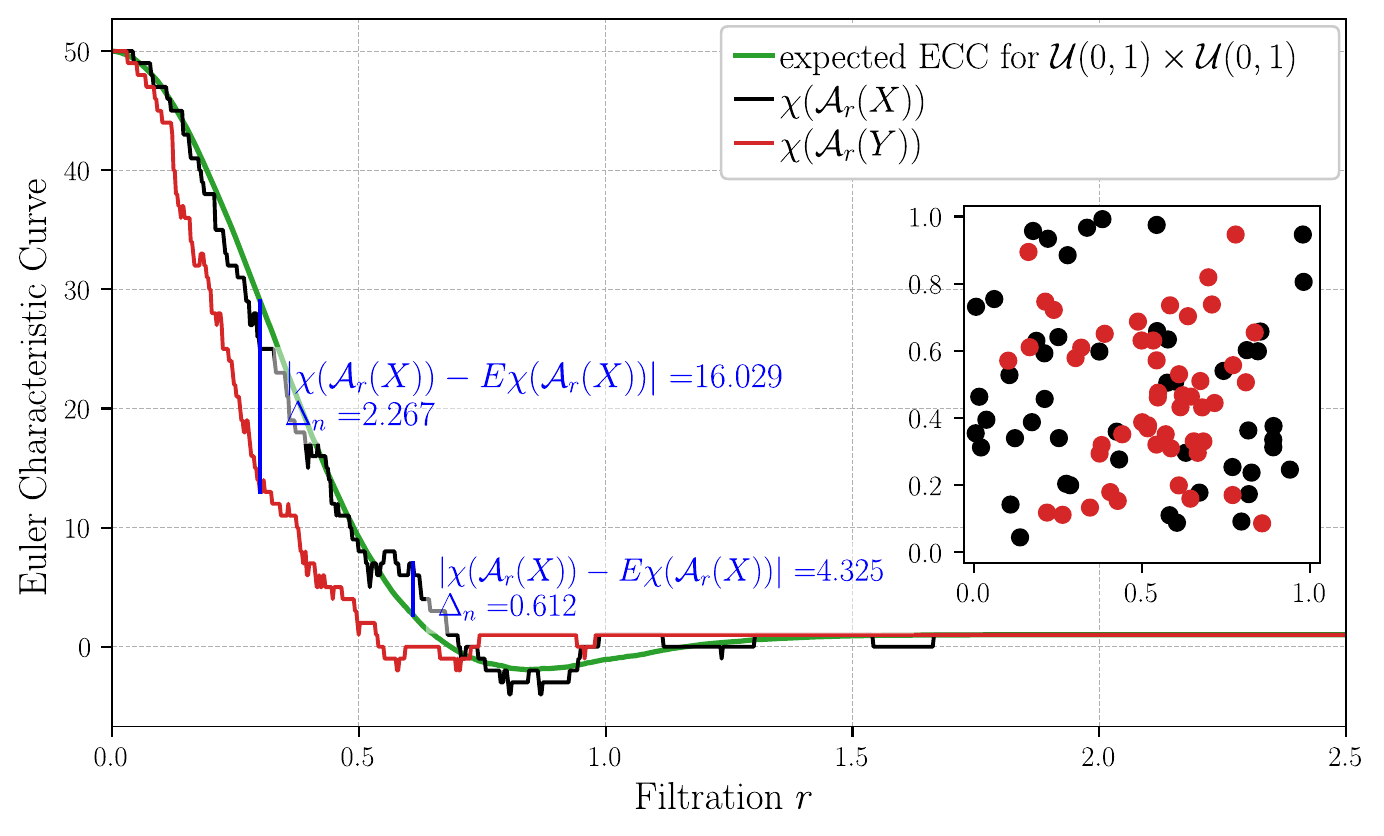}%
 	\caption{
  Euler characteristic curves of two samples of a size $50$; $X \sim \mathcal{U}(0,1) \times \mathcal{U}(0,1)$ (in black) and $Y \sim \beta(3,3) \times \beta(3,3)$ (in red). The green curve represents the expected ECC for $\mathcal{U}(0,1) \times \mathcal{U}(0,1)$.
  Samples are shown in the inset.
  }
  \label{fig:UniformSquareECCDistance}
 \end{figure}
Consider the samples $X, Y\subseteq [0,1]^2$ consisting of the $50$ black and $50$ red points
as shown in the inset in Figure \ref{fig:UniformSquareECCDistance}.
Let us look at the two samples separately, for each of them we perform the one-sample test against the uniform distribution.
We want to test, at significance level $\alpha=0.05$, whether they follow (up to an isometry of $\R^2$) the uniform distribution.
The ECC of $X$ is shown in black and the one of $Y$ in red in Figure \ref{fig:UniformSquareECCDistance}.
The green curve represents the expected ECC under the null hypothesis, estimated via $M=1000$ Monte Carlo iterations using (\ref{eqn:estimated_E}).
We find the test statistic (\ref{eqn:Delta_hat}) computed between 
the $\chi(\mathcal{A}_r(X))$ and the average curve is $\widehat{\Delta}_n = 0.612$.
Comparing this with the computed threshold of $t_\alpha = 1.318$, we conclude that we do not have evidence to reject the null hypothesis. The $p$-value is $0.916$.
In contrast, test statistics computed for $\chi(\mathcal{A}_r(Y))$ is much larger and equals $\widehat{\Delta}_n = 2.267$. Again using $\alpha = 0.05$, the test provides evidence to reject the null hypothesis with $p$-value computed to be $0.00$. And indeed, we generated $X$ from the bivariate uniform distribution (i.e. null distribution) whereas $Y$ was sampled from $\beta(3,3) \times \beta(3,3)$
\red{, i.e.  Cartesian product of two independent univariate $\beta(3,3)$ distributions. }.
\end{example}

\begin{example}\label{ex:counterexample}
\begin{figure}
 	\includegraphics[width=\linewidth]{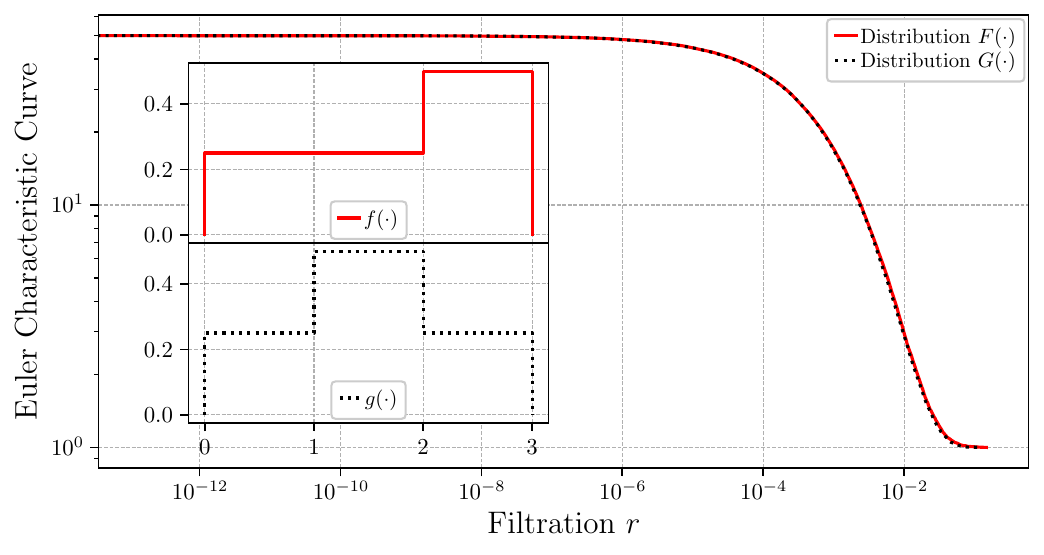}%
    \caption{Expected ECCs of distributions $F$ and $G$ for $n=50$. The inset shows the corresponding densities $f$ and $g$}
    \label{fig:sameECC}
\end{figure}
Consider the distributions $F$ and $G$ with densities
\begin{align*}
        f(x) &= \frac{1}{2}\mathbb{I}_{(0, 2)}(x) + \frac{1}{2}\mathbb{I}_{(2, 3)}(x), \\
         g(x) &= \frac{1}{4}\mathbb{I}_{(0,1)}(x) + \frac{1}{2}\mathbb{I}_{(1,2)}(x) 
        + \frac{1}{4}\mathbb{I}_{(2,3)}(x)
\end{align*}
    Observe that for each $t>0$,
    \begin{equation}
    \label{eqn:beta_condition}
    \int\limits_{f\geq t} f(x) \text{d}x = \int\limits_{g\geq t} g(x) \text{d}x = \begin{cases}
        1  & \text{if }  t\leq 1/4,\\
        1/2  & \text{if }  1/4 < t \leq 1/2,\\
        0  & \text{if }  t > 1/2.
    \end{cases}
    \end{equation}
    Hence by Lemma 5.1 of \cite{vishwanath_limits_2022}, the ECCs of $F$ and $G$ in the thermodynamic limit follow the same distribution. The limiting ECCs for $F$ and $G$ are shown in Figure \ref{fig:sameECC}. Note that distributions $F$ and $G$ are not isometric-equivalent and yet the corresponding ECCs are the same as the distributions are $\beta$-equivalent, hence also Euler equivalent.
    $F$ and $G$ therefore form an example of distributions that are indistinguishable by TopoTest.
    Indeed, the power of one-sample Kolmogorov-Smirnov test, when $F$ is used as a null distribution and 50 elements samples are drawn from $G$, is $0.91$ and only $0.05$, i.e. $\alpha$, for TopoTest.  

\end{example}

\subsection{Two-sample test}
In Section~\ref{sec:two_sample_test} a related approach to the two-sample problem was presented.
This idea is formally provided by the Algorithm \ref{algo:2SampleTesting} while a particular realization is presented in the example below.
\begin{algorithm}
    \caption{Two-sample testing}\label{algo:2SampleTesting}
    \SetAlgoLined
    \KwInput{two sample points $X=\{x_1, \ldots, x_m\}, Y = \{y_1,\ldots,y_n\}$  both in $\R^d$, number $K$ of Monte Carlo iterations, significance level $\alpha$.}
    \KwOutput{\red{Rejecting or failure to reject} of null hypothesis, $p$-value}
    
    Compute the distance $D$ between normalized ECCs build on top of $X$ and $Y$
    $$D \gets \sup_r \left\vert \frac{1}{m}\chi(\mathcal{A}_r(X)) - \frac{1}{n}\chi(\mathcal{A}_r(Y)) \right\vert $$
    Pool the data points $Z \gets X \cup Y$
    
    \For{$p \gets 1,\ldots, K$}
    {
        $Z_{(p)}^{\#} \gets permute(Z)$
        
        Split $Z_{(p)}$ into two samples of size $m$ and $n$
        
        $X_{(p)} \gets \{Z_{(p),1}, Z_{(p),2}, \ldots, Z_{(p),m}\}$

        $Y_{(p)} \gets \{ Z_{(p),m+1}, Z_{(p),m+2}, \ldots, Z_{(p),m+n}\}$
 
        Compute the distance between ECCs build on top of  $X_{(p)}$ and  $Y_{(p)}$
        
        $d_{(p)} \gets \sup_r \left\vert \frac{1}{m}\chi\left(\mathcal{A}_r\left(X_{(p)}\right)\right) - \frac{1}{n}\chi\left(\mathcal{A}_r\left(Y_{(p)}\right)\right) \right\vert$
    }
    $pv \gets \frac{1}{K} \#\{  d_{(p)} > D\}$
    
    \Return{\red{$pv < \alpha$, $pv$}}
\end{algorithm}
Let us begin with the situation in which the null hypothesis is not rejected.
\begin{example}
Consider both $X$ and $Y$ sampled from $\mathcal{U}(0,1)^2$ with $\vert X\vert = 30$, $\vert Y\vert=50$, shown in the inset of Figure \ref{fig:2sample-H0}.
\begin{figure}
 		\includegraphics[width=\linewidth]{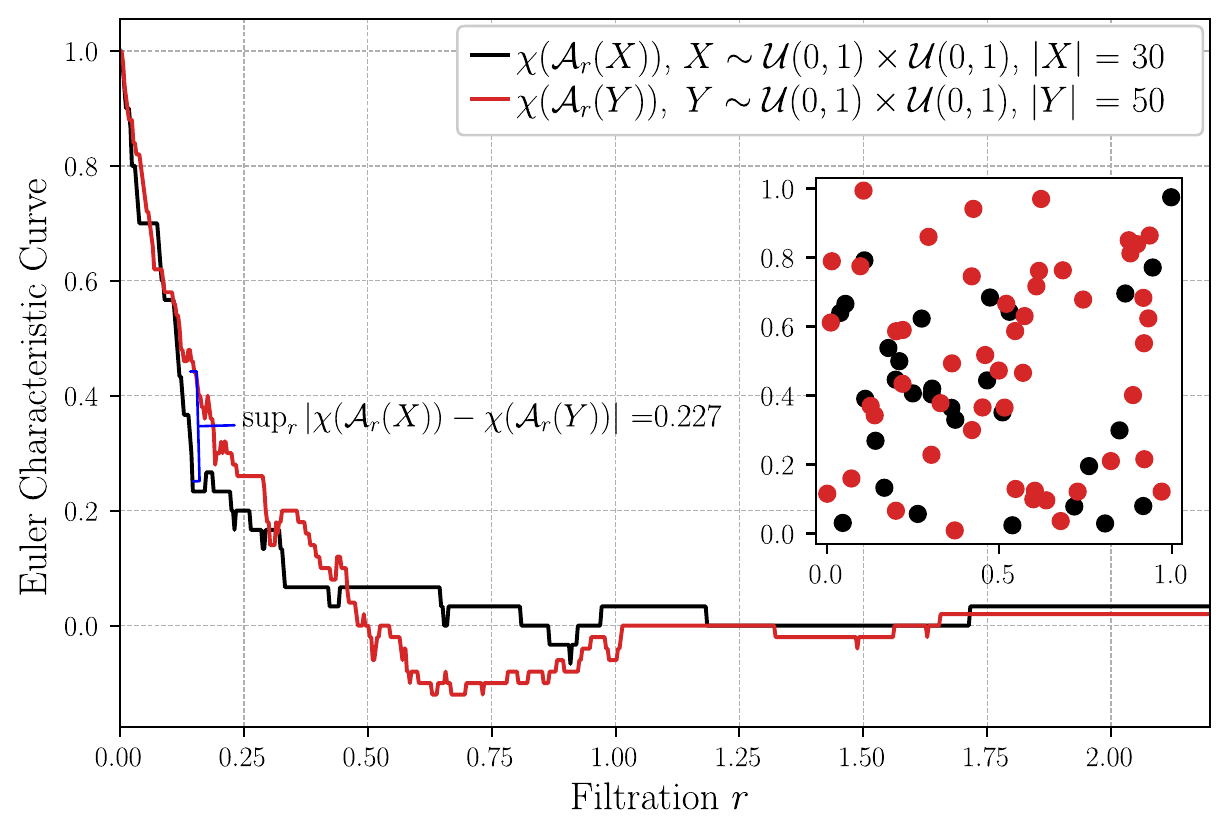}%
    \caption{Normalized Euler Characteristic Curves of two samples of size $30$ and $50$ drawn from bivariate uniform distribution, $\mathcal{U}(0,1) \times \mathcal{U}(0,1)$. 
    Samples are shown in the inset.} 
    \label{fig:2sample-H0}
\end{figure}
We compute the supremum distance between the normalized ECCs to be $D=0.227$, as illustrated in Figure \ref{fig:2sample-H0}. Using $K=1000$ Monte Carlo iterations we find that a distance between ECCs at least as extreme as $D$ happens roughly $73\%$ of the time.
We conclude that we do not have evidence to reject the null hypothesis at significance level $\alpha=0.05$.
\end{example}
Now let us turn to an example in which the null hypothesis is rejected.
\begin{example}
In the Figure \ref{fig:2sample-H1}, we have sampled $X$ as $30$ points from the bivariate uniform distribution on the unit square $\mathcal{U}(0,1)^2$, whereas $Y$ consists of $50$ points sampled from $\beta(3,3)\times \mathcal{U}(0,1)$.
We compute the distance between corresponding normalized ECCs to be $D=0.453$.
In $K=1000$ Monte Carlo iterations, we find that an ECC distance at least as extreme as $D$ never happens, hence using $\alpha = 0.05$ this establishes evidence to reject the null hypothesis.
\begin{figure}
    \centering
 		\includegraphics[width=\linewidth]{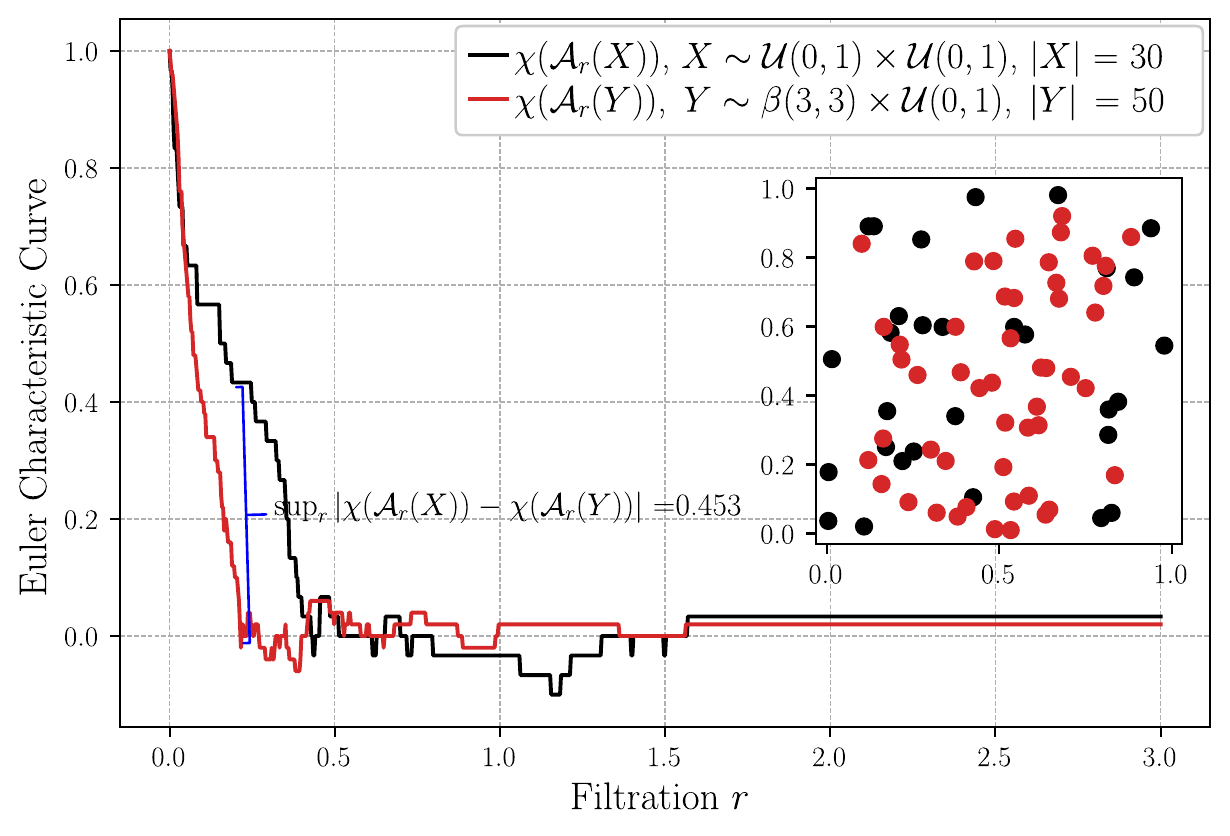}%
    \caption{Normalized Euler Characteristic Curves of two samples of size 30 and 50 drawn from
    different distributions: $X \sim \mathcal{U}(0,1) \times \mathcal{U}(0,1)$ and $Y \sim \beta(3, 3) \times \mathcal{U}(0,1)$.}
    \label{fig:2sample-H1}
\end{figure}

\end{example}

\section{Numerical Experiments, one-sample problem}
\label{sec:NumericalExperiments_one_sample}

In this study, Monte Carlo simulations were used to evaluate \red{the} power of TopoTests and compare it with the power of corresponding Kolmogorov-Smirnov tests.
In case of univariate distributions, Cram\'er-von Mises was considered as well for completeness.
To obtain more detailed insight into performance of TopoTests, samples of various sizes ranging from $n=30$ up to $n=1000$, were examined.
In the following subsections three types of experiments are presented:

\begin{enumerate}
\item Fixing the null distribution to be standard normal and test samples drawn from a vast variety of alternative distributions with different parameters; Laplace, uniform, t-distribution, as well as Cauchy, logistic distributions and mixture of Gaussians. This set of experiments allowed to assess how well TopoTests performs to recognize standard normal distributions.
\item Fixing a family of distributions, and treat each of them as null distribution while all others are considered as alternative distribution. For each such a pair of distributions, the empirical power of the test, i.e. 1 minus probability of type II error, was computed using Monte Carlo methods.
The result was visualized in a form of heat-maps. 
\item In addition, for various dimension\red{s}, a relation between power of the test and number $n$ of data points in the sample was examined. As expected, the power of the test increases monotonically with the sample size. 
\end{enumerate}

In this section both simulations satisfying Assumption~\ref{assumption} and those that do not satisfy it (for instance multivariate normal) were considered.
To theoretically underpin this approach, several ideas were suggested in Section~\ref{sec:what_to_do_in_non_compact_case}.
In practice, the fact that the Assumption~\ref{assumption} was not satisfied in some cases did not affect the test powers.

\remark{
In this section we benchmark TopoTest by comparing its power with the power of Kolmogorov-Smirnov test, i.e. the probability that the test correctly rejects null hypothesis when the alternative distribution is different than null distribution. Since TopoTests is not able to distinguish different but Euler-equivalent distributions, which Kolmogorov-Smirnov can distinguish, the setting under which it operates (\ref{eqn:EulerOneSampleHypothesis}) is different from the Kolmogorov-Smirnov setting  (\ref{eqn:OneSampleHypothesis}), and hence the reported power of TopoTest might be overestimated. To mediate this effect a vast collection of distributions was considered.
}

\subsection{Compactly supported distributions}
As a first example a collection of distributions supported on three-dimensional unit cube $[0, 1]^3$
was considered. The collection consisted of a number of three-fold Cartesian products of independent beta, cosine (rescaled to fit unit interval) and uniform univarite distributions.  
In such setup the Assumption~\ref{assumption} is fulfilled and developed theory can be applied straightforwardly. 
In Figure \ref{fig:power_matrix_compact} the power of TopoTest was compared with power of 
Kolmogorov-Smirnov test for a collection of trivariate
distributions on compact domain. Several sample sizes were considered but here only results obtained for $n=100$ are reported as similar conclusions can be drawn for different values of $n$.
\begin{figure*}
    \caption{
        Average power of TopoTest (left panel) and Kolmogorov-Smirnov test for selected trivariate on compact
        support on $[0, 1]^3$. Average power, at significance level $\alpha=0.05$, is estimated based on 
        $K=1000$ Monte Carlo realizations for sample size $n=100$.
    }
    \label{fig:power_matrix_compact}
    \includegraphics[width=0.49\linewidth]{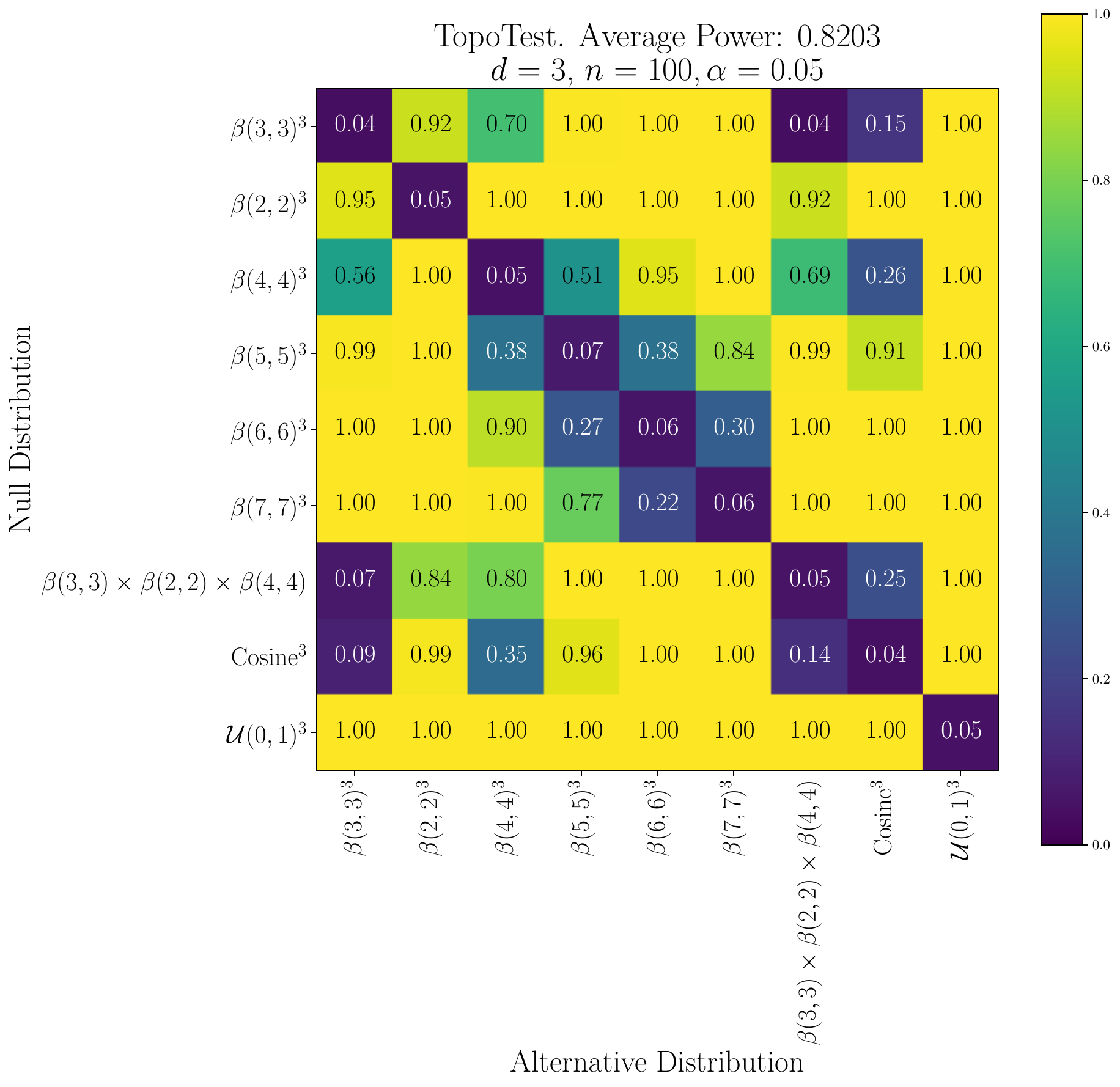}
    \includegraphics[width=0.49\linewidth]{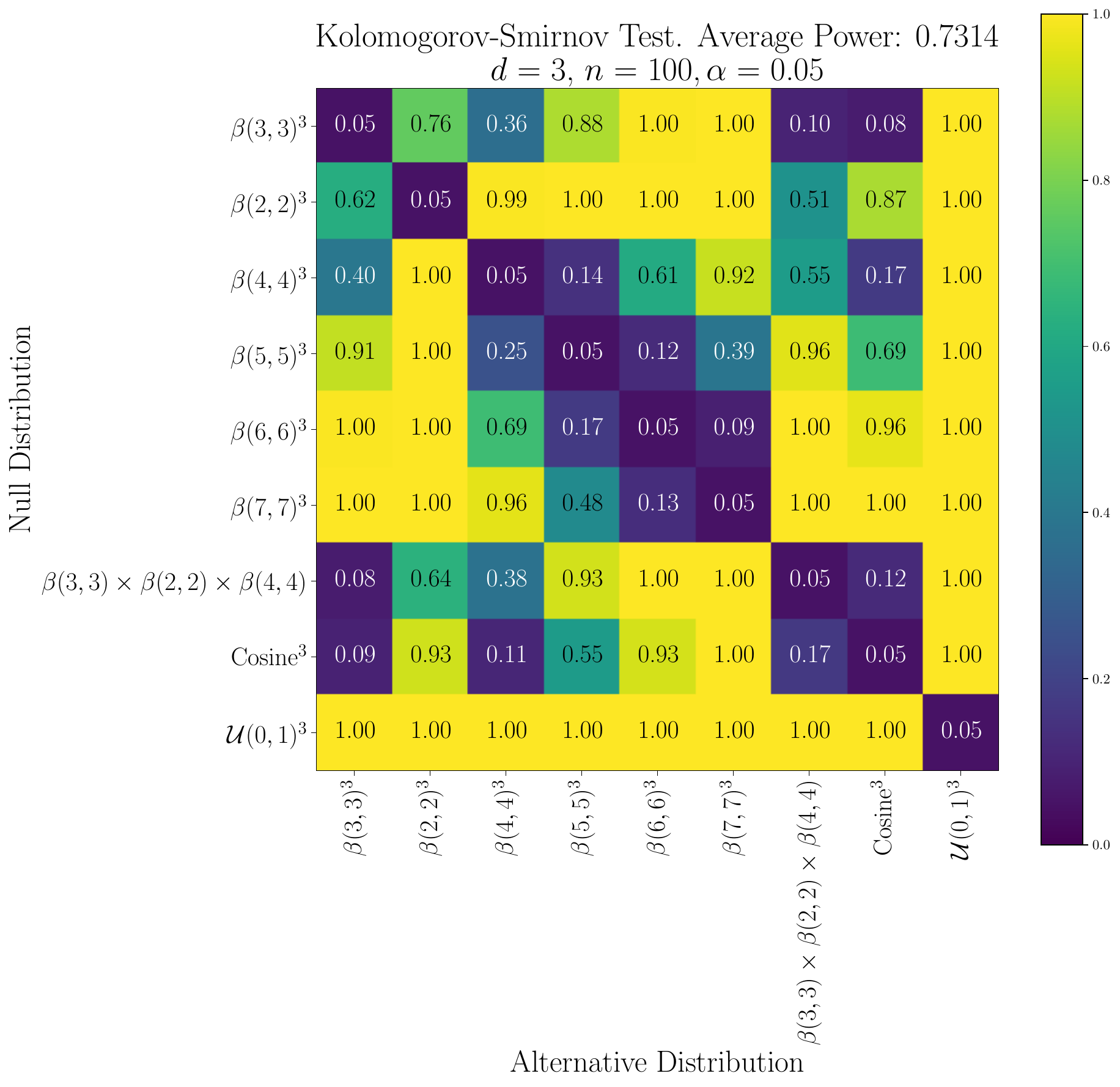}
\end{figure*} 
The TopoTest provided higher power for vast majority of considered pairs of null and alternative distributions resulting in average power, at significance level $\alpha=0.05$, for this collection of distributions to be $0.82$ for TopoTest and $0.73$ for Kolmogorov-Smirnov. 
In fact, for collection of distributions considered in Figure \ref{fig:power_matrix_compact} in only one, out of 72, comparisons the power of Kolmogorov-Smirnov test was higher than the one for TopoTest, and the difference was slim ($0.07$ vs. $0.08$).

\subsection{Univariate unbounded distributions}
In this section we consider a vast collection of univariate unbounded distribution represented on a computer (hence, restricted to a representable range of double precision numbers). The collection include normal distributions $\mathcal{N}(0, \sigma^2)$ with different values of $\sigma$, Cauchy, Laplace, Logistic distributions, Student's t-distributions with increasing number of degrees of freedom $\nu$ 
as well as
Gaussian mixtures defined as $GM(p, \mu, \sigma) = p\mathcal{N}(0,1) + (1-p)\mathcal{N}(\mu, \sigma)$, for $p\in\{0.1, 0.3, 0.5, 0.7, 0.9\}$,
$\mu\in\{0, 1\}$ and $\sigma\in\{\frac{1}{2}, 1, 2\}$. 
For completeness some distributions defined on compact support are considered as well.

Table~\ref{tab:onesample_1d} provides the empirical power of TopoTests, assessed based on $K=5000$ Monte Carlo simulations, in distinguishing a standard normal $\mathcal{N}(0,1)$ from a number of alternative distributions at significance level $\alpha=0.05$.

\begin{table*}
    \caption{Empirical powers of the one-sample TopoTest for different alternative distributions and sample sizes $n$ -- 
    the null distribution was standard normal $\mathcal{N}(0,1)$. 
    Corresponding powers of Kolmogorov-Smirnov tests are given in parenthesis for comparison -- higher result is given in bold for easier comparison. 
    Results for the significance level $\alpha=0.05$. Empirical powers estimated based on $K=5000$ Monte Carlo simulations.}
    \label{tab:onesample_1d}
    \begin{tabular}{lrrrrr}
    \hline
& \multicolumn{5}{c}{Sample size $n$} \\
\cmidrule{2-6}
Alternative Distribution & \red{30}&                     50  &                     100 &                     250 &                     500 \\ \hline
$\mathcal{N}(0,  0.50)$                           &  \textbf{0.953} (0.417) &  \textbf{0.997} (0.820) &           1.000 (1.000) &           1.000 (1.000) &           1.000 (1.000) \\
$\mathcal{N}(0, 0.75)$                            &  \textbf{0.278} (0.061) &  \textbf{0.369} (0.097) &  \textbf{0.705} (0.247) &  \textbf{0.995} (0.734) &  \textbf{1.000} (0.998) \\
$\mathcal{N}(0, 1.25)$                            &  \textbf{0.222} (0.096) &  \textbf{0.291} (0.123) &  \textbf{0.477} (0.211) &  \textbf{0.879} (0.459) &  \textbf{0.998} (0.899) \\
$\mathcal{N}(0, 1.5)$                             &  \textbf{0.519} (0.228) &  \textbf{0.670} (0.327) &  \textbf{0.956} (0.688) &  \textbf{1.000} (0.990) &           1.000 (1.000) \\
Laplace$(0, 1)$                                   &  \textbf{0.224} (0.055) &  \textbf{0.309} (0.058) &  \textbf{0.544} (0.084) &  \textbf{0.918} (0.145) &  \textbf{1.000} (0.534) \\
$\mathcal{U}(-\sqrt{3}, \sqrt{3})$                &  0.037 (\textbf{0.110}) &  0.041 (\textbf{0.141}) &  0.099 (\textbf{0.249}) &  \textbf{0.840} (0.558) &  \textbf{1.000} (0.930) \\
$\mathcal{U}(0, 1)$                               &           1.000 (1.000) &           1.000 (1.000) &           1.000 (1.000) &           1.000 (1.000) &           1.000 (1.000) \\
$t(3)$                                            &  \textbf{0.280} (0.070) &  \textbf{0.400} (0.066) &  \textbf{0.674} (0.122) &  \textbf{0.966} (0.267) &  \textbf{1.000} (0.700) \\
$t(5)$                                            &  \textbf{0.151} (0.054) &  \textbf{0.169} (0.054) &  \textbf{0.306} (0.068) &  \textbf{0.636} (0.080) &  \textbf{0.918} (0.176) \\
$t(10)$                                           &  \textbf{0.084} (0.049) &  \textbf{0.080} (0.043) &  \textbf{0.111} (0.051) &  \textbf{0.246} (0.053) &  \textbf{0.346} (0.074) \\
$t(25)$                                           &  \textbf{0.059} (0.052) &  \textbf{0.054} (0.041) &  \textbf{0.066} (0.060) &  \textbf{0.072} (0.045) &  \textbf{0.081} (0.053) \\
Cauchy$(0, 1)$                                    &  \textbf{0.907} (0.281) &  \textbf{0.971} (0.456) &  \textbf{1.000} (0.850) &           1.000 (1.000) &           1.000 (1.000) \\
Logistic$(0, 1)$                                  &  \textbf{0.760} (0.322) &  \textbf{0.903} (0.511) &  \textbf{0.996} (0.885) &           1.000 (1.000) &           1.000 (1.000) \\
0.9$\mathcal{N}(0, 1)$ + 0.1$\mathcal{N}(0, 0.5)$ &  \textbf{0.065} (0.042) &  \textbf{0.048} (0.038) &  \textbf{0.073} (0.072) &  \textbf{0.090} (0.059) &  \textbf{0.137} (0.093) \\
0.7$\mathcal{N}(0, 1)$ + 0.3$\mathcal{N}(0, 0.5)$ &  \textbf{0.124} (0.052) &  \textbf{0.136} (0.078) &  \textbf{0.248} (0.136) &  \textbf{0.542} (0.337) &  \textbf{0.816} (0.784) \\
0.5$\mathcal{N}(0, 1)$ + 0.5$\mathcal{N}(0, 0.5)$ &  \textbf{0.292} (0.088) &  \textbf{0.375} (0.152) &  \textbf{0.637} (0.404) &  \textbf{0.978} (0.912) &  0.999 (\textbf{1.000}) \\
0.3$\mathcal{N}(0, 1)$ + 0.7$\mathcal{N}(0, 0.5)$ &  \textbf{0.544} (0.159) &  \textbf{0.746} (0.329) &  \textbf{0.961} (0.855) &           1.000 (1.000) &           1.000 (1.000) \\
0.1$\mathcal{N}(0, 1)$ + 0.9$\mathcal{N}(0, 0.5)$ &  \textbf{0.852} (0.304) &  \textbf{0.977} (0.672) &  \textbf{1.000} (0.995) &           1.000 (1.000) &           1.000 (1.000) \\
0.9$\mathcal{N}(0, 1)$ + 0.1$\mathcal{N}(0, 2)$   &  \textbf{0.092} (0.052) &  \textbf{0.077} (0.050) &  \textbf{0.143} (0.064) &  \textbf{0.229} (0.056) &  \textbf{0.413} (0.087) \\
0.7$\mathcal{N}(0, 1)$ + 0.3$\mathcal{N}(0, 2)$   &  \textbf{0.256} (0.085) &  \textbf{0.350} (0.098) &  \textbf{0.627} (0.140) &  \textbf{0.943} (0.315) &  \textbf{1.000} (0.778) \\
0.5$\mathcal{N}(0, 1)$ + 0.5$\mathcal{N}(0, 2)$   &  \textbf{0.514} (0.152) &  \textbf{0.683} (0.212) &  \textbf{0.952} (0.449) &  \textbf{1.000} (0.933) &           1.000 (1.000) \\
0.3$\mathcal{N}(0, 1)$ + 0.7$\mathcal{N}(0, 2)$   &  \textbf{0.733} (0.291) &  \textbf{0.898} (0.450) &  \textbf{0.997} (0.858) &  \textbf{1.000} (0.999) &           1.000 (1.000) \\
0.1$\mathcal{N}(0, 1)$ + 0.9$\mathcal{N}(0, 2)$   &  \textbf{0.875} (0.491) &  \textbf{0.968} (0.750) &  \textbf{1.000} (0.984) &           1.000 (1.000) &           1.000 (1.000) \\
0.9$\mathcal{N}(0, 1)$ + 0.1$\mathcal{N}(1, 2)$   &  \textbf{0.096} (0.068) &  \textbf{0.111} (0.063) &  \textbf{0.171} (0.092) &  \textbf{0.319} (0.135) &  \textbf{0.548} (0.280) \\
0.7$\mathcal{N}(0, 1)$ + 0.3$\mathcal{N}(1, 2)$   &  \textbf{0.318} (0.182) &  \textbf{0.464} (0.249) &  \textbf{0.747} (0.508) &  \textbf{0.985} (0.932) &           1.000 (1.000) \\
0.5$\mathcal{N}(0, 1)$ + 0.5$\mathcal{N}(1, 2)$   &  \textbf{0.588} (0.453) &  \textbf{0.760} (0.665) &  \textbf{0.971} (0.948) &           1.000 (1.000) &           1.000 (1.000) \\
0.3$\mathcal{N}(0, 1)$ + 0.7$\mathcal{N}(1, 2)$   &  \textbf{0.778} (0.747) &  0.927 (\textbf{0.930}) &           0.999 (0.999) &           1.000 (1.000) &           1.000 (1.000) \\
0.1$\mathcal{N}(0, 1)$ + 0.9$\mathcal{N}(1, 2)$   &  0.889 (\textbf{0.921}) &  0.987 (\textbf{0.990}) &           1.000 (1.000) &           1.000 (1.000) &           1.000 (1.000) \\
\hline
Average Power                                     &  \textbf{0.446} (0.246) &  \textbf{0.527} (0.338) &  \textbf{0.659} (0.501) &  \textbf{0.808} (0.643) &  \textbf{0.866} (0.764) \\
\hline
    \end{tabular}    
\end{table*}

As we can observe in Table~\ref{tab:onesample_1d}, TopoTest outperformed the Kolmogorov-Smirnov test when distinguishing between the standard normal distribution from the normal distribution 
with variance different from $1$, regardless of the sample size. The power of the TopoTest is also greater when the alternative distribution is Student's t-distribution: the difference compared to the Kolmogorov-Smirnov test was particularly pronounced when the number of degrees of freedom $\nu$ was small. When $\nu$ was 10 or more, the power of both tests is much lower, as expected, but still TopoTest outperformed the Kolmogorov-Smirnov test.
Similar conclusion can be drawn for heavier tail alternative distributions such as Cauchy, Laplace or Logistic distribution: 
the empirical probability of type II error was always lower for TopoTest than for Kolmogorov-Smirnov counterpart.
On the other hand, when Gaussian mixtures were considered, it was the Kolmogorov-Smirnov test that performs better, regardless of the value of mixing coefficient $p$.  
\subsection{Two and three dimensional unbounded distributions}
In Table \ref{tab:onesample_2d} result for collection of bivariate distributions are shown. 
The $MG(a)$ denotes a multivariate normal distribution with non-diagonal covariance matrix, i.e.
\begin{equation}
    \label{eqn:MG}
    MG(a) = \mathcal{N}\left(0, \begin{bmatrix}
        1 & a & a & \dots  & a \\
        a & 1 & a & \dots  & a \\
        \vdots & \vdots & \vdots & \ddots & \vdots \\
        a & a & a & \dots  & 1
    \end{bmatrix} \right),
\end{equation}
where \red{the} value of the parameter $a$ varies from $0$ to $1$ to reflect increasing correlation of components.

\begin{table*}
    \caption{The same as Table \ref{tab:onesample_1d} but for two dimensional distributions. 
    Null distribution is $N_0 \sim \mathcal{N}(0, I_{2})$, 
    where $I_2$ is a $2\times2$ identity matrix.  Empirical powers, based on $K=1000$ Monte Carlo simulations. 
    Alternative distributions include Gaussian mixtures of $N_0$, 
    $N_1 \sim \mathcal{N}( (1, 1), 3I_2), N_2 \sim \mathcal{N}( (0, 0), 3I_2)$ and $N_3 \sim \mathcal{N}( (-1, -1), 3I_2)$.
    }
    \label{tab:onesample_2d}
        \begin{tabular}{lrrrrr}
    \hline
& \multicolumn{5}{c}{Sample size $n$} \\
\cmidrule{2-6}
Alternative Distribution & \red{30} &                    50  &                     100 &                     250 &                     500 \\ \hline
$MG(0.05)$                                         &  0.036 (\textbf{0.052}) &           0.050 (0.050) &  \textbf{0.049} (0.038) &  0.061 (\textbf{0.070}) &  \textbf{0.059} (0.048) \\
$MG(0.1)$                                          &  0.042 (\textbf{0.044}) &  0.041 (\textbf{0.056}) &  \textbf{0.048} (0.042) &  0.052 (\textbf{0.074}) &  0.065 (\textbf{0.096}) \\
$MG(0.2)$                                          &  0.040 (\textbf{0.073}) &  0.064 (\textbf{0.114}) &  0.060 (\textbf{0.106}) &  0.062 (\textbf{0.170}) &  0.062 (\textbf{0.298}) \\
$MG(0.3)$                                          &  0.046 (\textbf{0.072}) &  0.064 (\textbf{0.130}) &  0.071 (\textbf{0.134}) &  0.090 (\textbf{0.368}) &  0.121 (\textbf{0.702}) \\
$MG(0.5)$                                          &  0.093 (\textbf{0.124}) &  0.115 (\textbf{0.258}) &  0.200 (\textbf{0.478}) &  0.369 (\textbf{0.952}) &  0.652 (\textbf{1.000}) \\
$MG(0.7)$                                          &  \textbf{0.232} (0.229) &  0.381 (\textbf{0.578}) &  0.688 (\textbf{0.902}) &  0.966 (\textbf{1.000}) &           1.000 (1.000) \\
$\mathcal{U}(-\sqrt{3},\sqrt{3}) \times \mathcal{U}(-\sqrt{3},\sqrt{3}$ &  0.044 (\textbf{0.157}) &  0.082 (\textbf{0.292}) &  \textbf{0.487} (0.468) &  \textbf{1.000} (0.942) &           1.000 (1.000) \\
$\mathcal{U}(0,1) \times \mathcal{U}(0,1)$         &           1.000 (1.000) &           1.000 (1.000) &           1.000 (1.000) &           1.000 (1.000) &           1.000 (1.000) \\
$t(3) \times t(3)$                                 &  \textbf{0.399} (0.121) &  \textbf{0.673} (0.176) &  \textbf{0.956} (0.308) &  \textbf{1.000} (0.838) &  \textbf{1.000} (0.996) \\
$t(5) \times t(5)$                                 &  \textbf{0.152} (0.073) &  \textbf{0.305} (0.104) &  \textbf{0.609} (0.124) &  \textbf{0.960} (0.304) &  \textbf{0.999} (0.660) \\
$t(10) \times t(10)$                               &  0.045 (\textbf{0.064}) &  \textbf{0.094} (0.088) &  \textbf{0.191} (0.078) &  \textbf{0.470} (0.094) &  \textbf{0.782} (0.100) \\
$t(25) \times t(25)$                               &  0.039 (\textbf{0.047}) &  0.066 (\textbf{0.068}) &  \textbf{0.067} (0.044) &  \textbf{0.096} (0.052) &  \textbf{0.165} (0.058) \\
$\mathcal{N}(0, 1) \times t(3)$                    &  \textbf{0.096} (0.064) &  \textbf{0.235} (0.086) &  \textbf{0.466} (0.102) &  \textbf{0.882} (0.244) &  \textbf{0.993} (0.422) \\
$\mathcal{N}(0, 1) \times t(5)$                    &  0.059 (\textbf{0.062}) &  \textbf{0.086} (0.068) &  \textbf{0.196} (0.086) &  \textbf{0.472} (0.122) &  \textbf{0.787} (0.116) \\
$\mathcal{N}(0, 1) \times t(10)$                   &  0.041 (\textbf{0.043}) &  0.052 (\textbf{0.060}) &  \textbf{0.068} (0.060) &  \textbf{0.141} (0.066) &  \textbf{0.270} (0.072) \\
$0.9N_0 + 0.1N_1$                                  &  0.051 (\textbf{0.074}) &  0.092 (\textbf{0.096}) &  \textbf{0.184} (0.102) &  \textbf{0.448} (0.238) &  \textbf{0.701} (0.406) \\
$0.7N_0 + 0.3N_1$                                  &  \textbf{0.284} (0.257) &  \textbf{0.519} (0.452) &  \textbf{0.842} (0.782) &  \textbf{0.998} (0.996) &           1.000 (1.000) \\
$0.5N_0 + 0.5N_1$                                  &  0.600 (\textbf{0.637}) &  \textbf{0.908} (0.902) &           0.998 (0.998) &           1.000 (1.000) &           1.000 (1.000) \\
$0.3N_0 + 0.7N_1$                                  &  0.843 (\textbf{0.917}) &  0.982 (\textbf{0.988}) &           1.000 (1.000) &           1.000 (1.000) &           1.000 (1.000) \\
$0.1N_0 + 0.9N_1$                                  &  0.943 (\textbf{0.995}) &  0.998 (\textbf{1.000}) &           1.000 (1.000) &           1.000 (1.000) &           1.000 (1.000) \\
$0.9N_0 + 0.1N_2$                                  &  0.050 (\textbf{0.064}) &  0.064 (\textbf{0.074}) &  \textbf{0.128} (0.052) &  \textbf{0.281} (0.080) &  \textbf{0.511} (0.114) \\
$0.7N_0 + 0.3N_2$                                  &  \textbf{0.185} (0.110) &  \textbf{0.369} (0.170) &  \textbf{0.679} (0.236) &  \textbf{0.982} (0.596) &  \textbf{1.000} (0.900) \\
$0.5N_0 + 0.5N_2$                                  &  \textbf{0.487} (0.237) &  \textbf{0.777} (0.422) &  \textbf{0.982} (0.678) &  \textbf{1.000} (0.984) &           1.000 (1.000) \\
$0.3N_0 + 0.7N_2$                                  &  \textbf{0.746} (0.433) &  \textbf{0.956} (0.702) &  \textbf{0.999} (0.956) &           1.000 (1.000) &           1.000 (1.000) \\
$0.1N_0 + 0.9N_2$                                  &  \textbf{0.902} (0.665) &  \textbf{0.996} (0.930) &           1.000 (1.000) &           1.000 (1.000) &           1.000 (1.000) \\
$0.9N_0 + 0.05N_1 + 0.05N_3$                       &  0.055 (\textbf{0.059}) &  \textbf{0.080} (0.078) &  \textbf{0.207} (0.080) &  \textbf{0.453} (0.128) &  \textbf{0.750} (0.178) \\
$0.7N_0 + 0.15N_1 + 0.15N_3$                       &  \textbf{0.308} (0.137) &  \textbf{0.566} (0.232) &  \textbf{0.879} (0.384) &  \textbf{0.998} (0.878) &  \textbf{1.000} (0.996) \\
$0.5N_0 + 0.25N_1 + 0.25N_3$ &  \textbf{0.679} (0.371) &  \textbf{0.918} (0.634) &  \textbf{0.998} (0.858) &           1.000 (1.000) &           1.000 (1.000) \\ \hline
Average Power                                      &  \textbf{0.303} (0.256) &  \textbf{0.412} (0.350) &  \textbf{0.538} (0.432) &  \textbf{0.671} (0.578) &  \textbf{0.747} (0.649) \\ \hline
\end{tabular}
\end{table*}   
Similarly to the univariate case, TopoTests provided lower type II errors in case of alternative distributions being 
products involving a Student's t-distribution. This conclusion holds also when one of \red{the} marginal distribution was a $\mathcal{N}(0,1)$ 
and second being Student's t-distribution.
A similar result is true for bivariate distributions being a Cartesian product involving  
Logistic or Laplace distribution.
We notice that TopoTest usually provided higher 
efficiency in case of Gaussian mixtures. 
On the other hand, TopoTest is significantly weaker than Kolmogorov-Smirnov when considering correlated multivariate normal distributions MG. 
All of these conclusions can be generalized to three dimensional distributions as initiated by results in Table \ref{tab:onesample_3d}.

\begin{table*}
    \caption{The same as Table \ref{tab:onesample_1d} but for three dimensional distributions. Null distribution is $N_0 \sim \mathcal{N}(0, I_{3})$, 
    where $I_{3}$ is a $3\times3$ identity matrix.  Empirical powers, based on $K=250$ Monte Carlo simulations. 
    Alternative distributions include Gaussian mixtures of $N_0, N_1 \sim \mathcal{N}( (1, 1, 1), 3I_3)$.
    }
    \label{tab:onesample_3d}
    \begin{tabular}{lrrrrr}
    \hline
& \multicolumn{5}{c}{Sample size $n$} \\
\cmidrule{2-6}
Alternative Distribution &    \red{30}&                 50  &                     100 &                     250 &                     500 \\ \hline 
$MG(0.05)$                               &  \textbf{0.052} (0.028) &  0.048 (\textbf{0.052}) &  0.064 (\textbf{0.068}) &  \textbf{0.062} (0.056) &  \textbf{0.056} (0.044) \\
$MG(0.1)$                                &  \textbf{0.056} (0.052) &  0.062 (\textbf{0.112}) &  \textbf{0.076} (0.068) &  0.038 (\textbf{0.104}) &  0.054 (\textbf{0.104}) \\
$MG(0.2)$                                &  \textbf{0.084} (0.076) &  0.062 (\textbf{0.120}) &  0.086 (\textbf{0.128}) &  0.074 (\textbf{0.328}) &  0.084 (\textbf{0.592}) \\
$MG(0.3)$                                &  0.084 (\textbf{0.104}) &  0.080 (\textbf{0.216}) &  0.134 (\textbf{0.252}) &  0.168 (\textbf{0.776}) &  0.276 (\textbf{0.992}) \\
$MG(0.5)$                                &  0.204 (\textbf{0.212}) &  0.252 (\textbf{0.576}) &  0.524 (\textbf{0.852}) &  0.854 (\textbf{1.000}) &  0.994 (\textbf{1.000}) \\
$\mathcal{U}(-{\sqrt{3}}, {\sqrt{3}})^3$ &  0.048 (\textbf{0.176}) &  0.156 (\textbf{0.408}) &  \textbf{0.632} (0.568) &  \textbf{1.000} (0.968) &           1.000 (1.000) \\
$\mathcal{U}(0,1)^3$                     &           1.000 (1.000) &           1.000 (1.000) &           1.000 (1.000) &           1.000 (1.000) &           1.000 (1.000) \\
$t(3)^3$                                 &  \textbf{0.624} (0.168) &  \textbf{0.836} (0.388) &  \textbf{0.998} (0.524) &  \textbf{1.000} (0.996) &           1.000 (1.000) \\
$t(5)^3$                                 &  \textbf{0.268} (0.056) &  \textbf{0.402} (0.196) &  \textbf{0.806} (0.240) &  \textbf{0.992} (0.560) &  \textbf{1.000} (0.936) \\
$t(10)^3$                                &  0.048 (\textbf{0.064}) &           0.108 (0.108) &  \textbf{0.266} (0.080) &  \textbf{0.624} (0.176) &  \textbf{0.906} (0.276) \\
$Logistic(0, 1)^3$                       &  \textbf{0.988} (0.904) &  \textbf{1.000} (0.996) &           1.000 (1.000) &           1.000 (1.000) &           1.000 (1.000) \\
$Laplace(0, 1)^3$                        &  \textbf{0.496} (0.116) &  \textbf{0.774} (0.220) &  \textbf{0.990} (0.332) &  \textbf{1.000} (0.924) &           1.000 (1.000) \\
$N_0 \times t(5) \times t(5)$            &  \textbf{0.140} (0.052) &  \textbf{0.238} (0.128) &  \textbf{0.520} (0.120) &  \textbf{0.824} (0.224) &  \textbf{0.996} (0.480) \\
$N_0 \times N_0 \times t(5)$             &  \textbf{0.056} (0.028) &  \textbf{0.082} (0.076) &  \textbf{0.154} (0.064) &  \textbf{0.304} (0.080) &  \textbf{0.586} (0.116) \\
$0.9N_0 + 0.1N_1$                        &  \textbf{0.100} (0.052) &  0.110 (\textbf{0.132}) &  \textbf{0.228} (0.116) &  \textbf{0.502} (0.304) &  \textbf{0.772} (0.500) \\
$0.5N_0 + 0.5N_1$                        &  \textbf{0.792} (0.748) &  \textbf{0.954} (0.944) &           1.000 (1.000) &           1.000 (1.000) &           1.000 (1.000) \\
$0.1N_0 + 0.9N_1$                        &  0.996 (\textbf{1.000}) &           1.000 (1.000) &           1.000 (1.000) &           1.000 (1.000) &           1.000 (1.000) \\ \hline
Average Power                            &  \textbf{0.355} (0.284) &  \textbf{0.421} (0.392) &  \textbf{0.558} (0.436) &  \textbf{0.673} (0.617) &  \textbf{0.748} (0.708) \\\hline
    \end{tabular}

\end{table*}
The last row of Tables \ref{tab:onesample_1d}, \ref{tab:onesample_2d} and \ref{tab:onesample_3d} 
show the average powers of TopoTest and Kolmogorov-Smirnov test for the considered set of alternative distributions. The average power of TopoTest is greater than that of Kolmogorov-Smirnov test for all studied sample sizes.
\subsection{All-to-all tests}

Results presented in Tables \ref{tab:onesample_1d},\ref{tab:onesample_2d},\ref{tab:onesample_3d} focused on the ability 
to discriminate the standard normal distribution
from a set of different distributions. However in TopoTest one can choose arbitrary continuous distributions as null and alternative. 
Hence below we present power matrices where all possible pairs of null and alternative distributions formed 
from the previous set were considered -- results are presented in Figures \ref{fig:power_matrix_1d},\ref{fig:power_matrix_2d},\ref{fig:power_matrix_3d}. 
For easier evaluation of the effectiveness of the TopoTest in comparison to Kolmogorov-Smirnov, the difference in 
power was shown in the figures. Hence, the blue region corresponds to combinations of
null and alternative distribution for which the TopoTest yielded higher power while red regions reflect the combinations for which TopoTest 
was outperformed by Kolmogorov-Smirnov.
red{White} color stands for combinations for which both tests performed similar.

\begin{figure*}
    \caption{Comparison of the power of TopoTest and Kolmogorov-Smirnov one-sample tests 
    in case of univariate probability distributions.
    In each matrix element a difference between power of TopoTest and Kolmogorov-Smirnov test was given. 
    The difference in power was estimated based on $K=1000$ Monte Carlo realizations. Left and right panels shows tests powers for sample sizes $n=100$ and $n=250$, respectively.
    The average power (excluding diagonal elements) of TopoTest is $0.722$ $(0.832)$ and $0.634$ $(0.794)$ for Kolmogorov-Smirnov for $n=100$ ($n=250$).
    }
    \label{fig:power_matrix_1d}
    \includegraphics[width=0.49\linewidth]{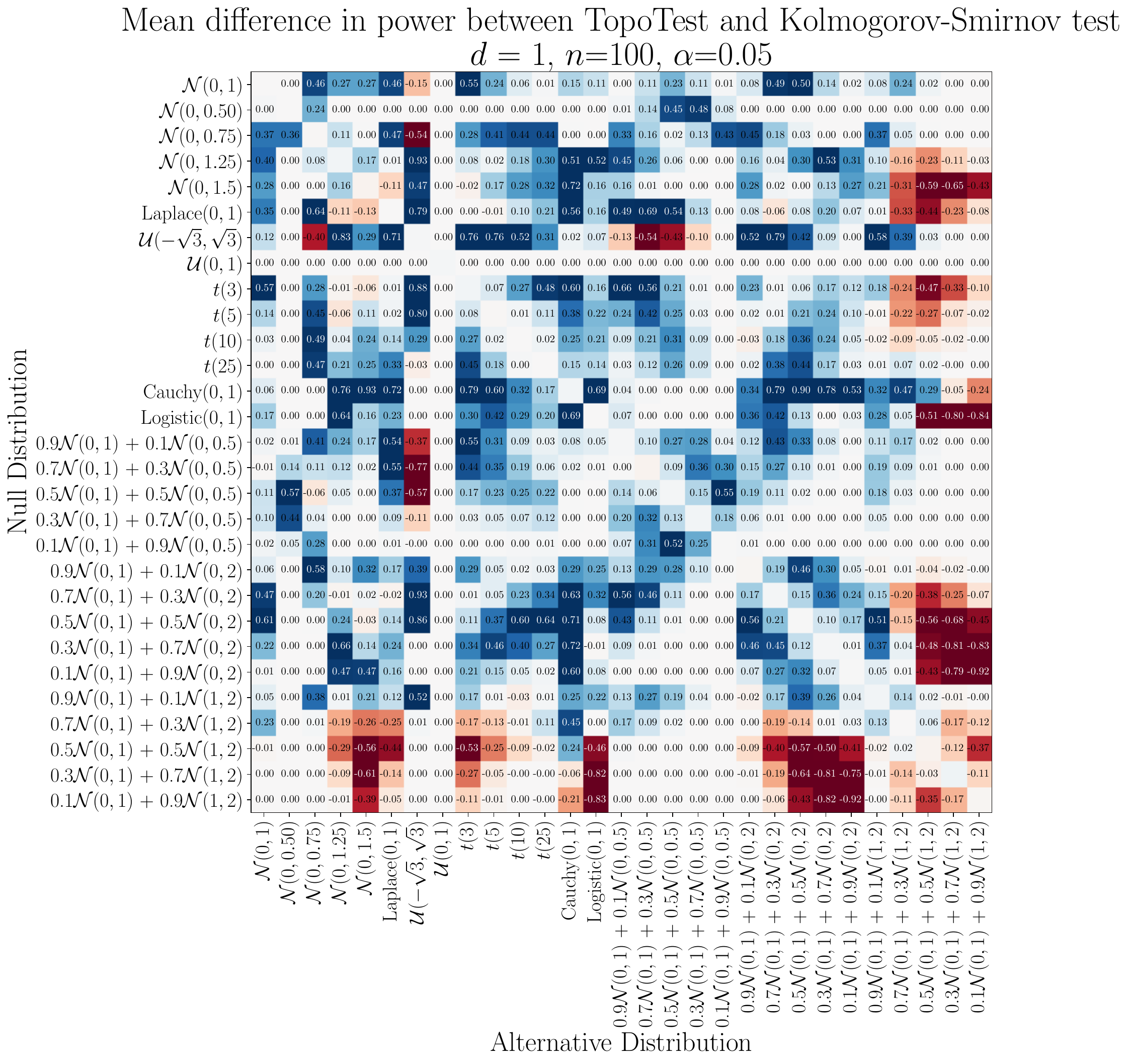}
    \includegraphics[width=0.49\linewidth]{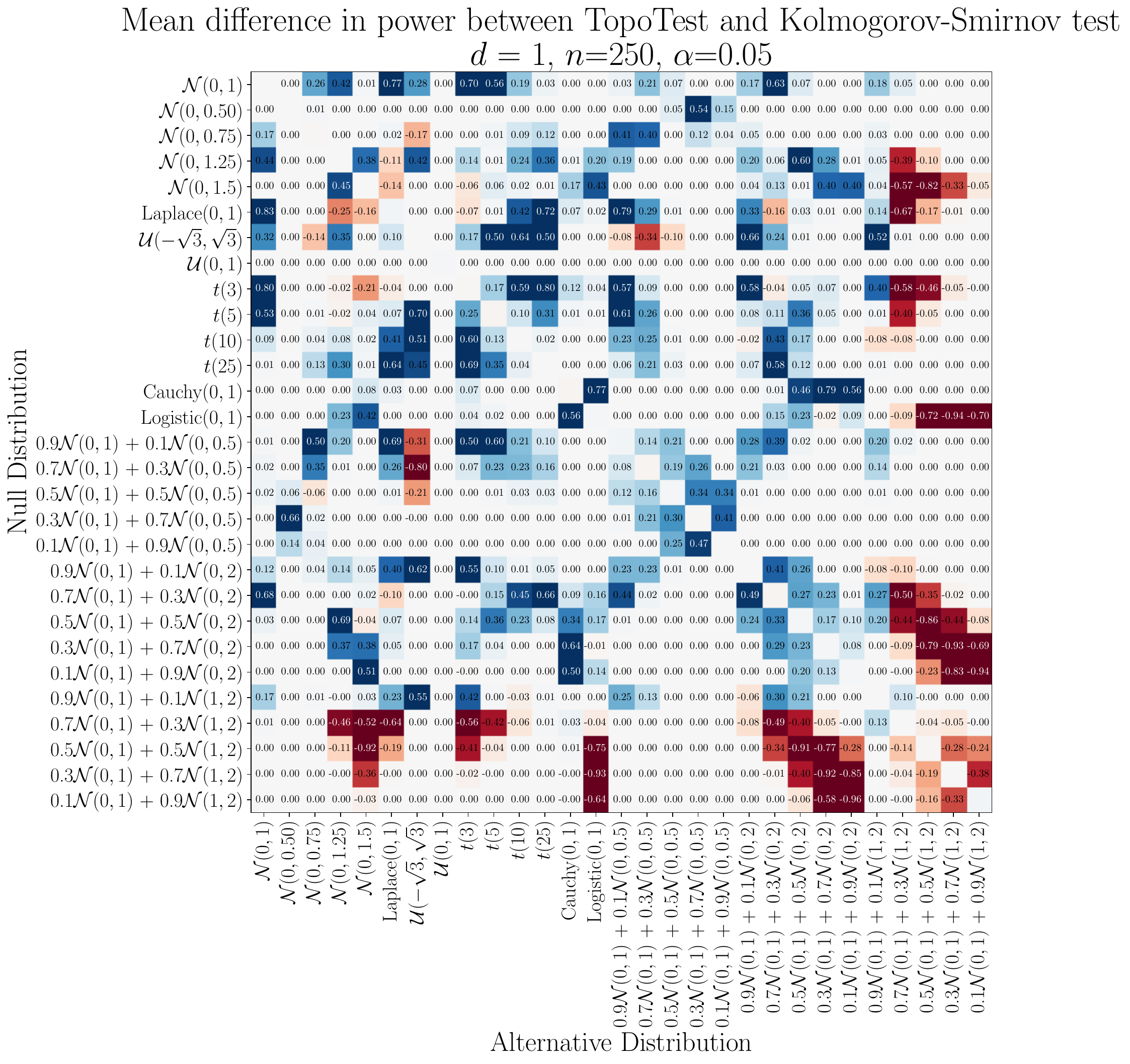}
\end{figure*}  

\begin{figure*}
    \caption{The same as Figure \ref{fig:power_matrix_1d} but for bivariate distributions. 
    Results based on $K=1000$ Monte Carlo realizations. 
    Average power is $0.642$ $(0.772)$ for TopoTest and $0.560$ $(0.720)$ for Kolmogorov-Smirnov
    for $n=100$ ($n=250$).
    }
    \label{fig:power_matrix_2d}
    \includegraphics[width=0.49\linewidth]{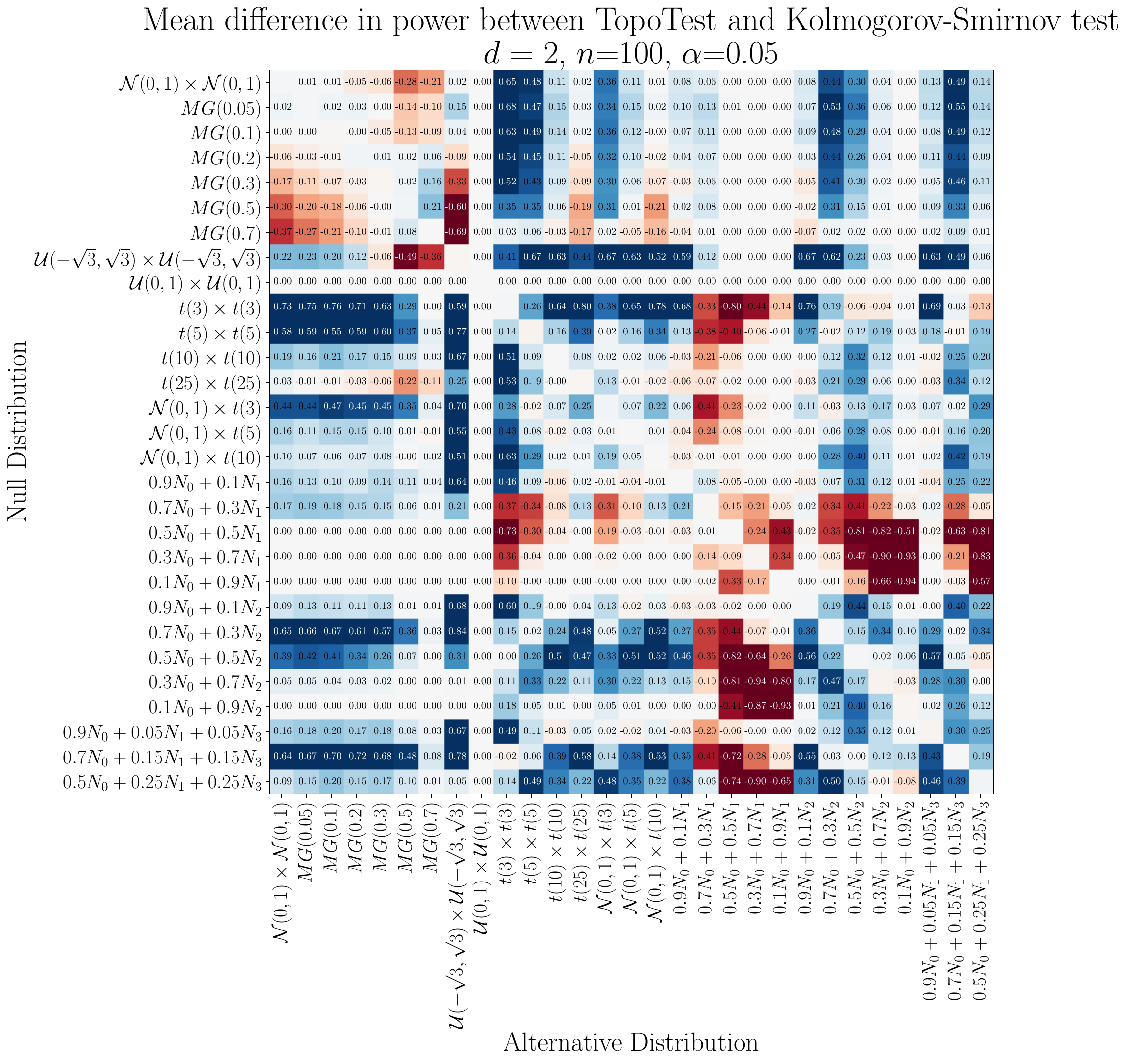}
    \includegraphics[width=0.49\linewidth]{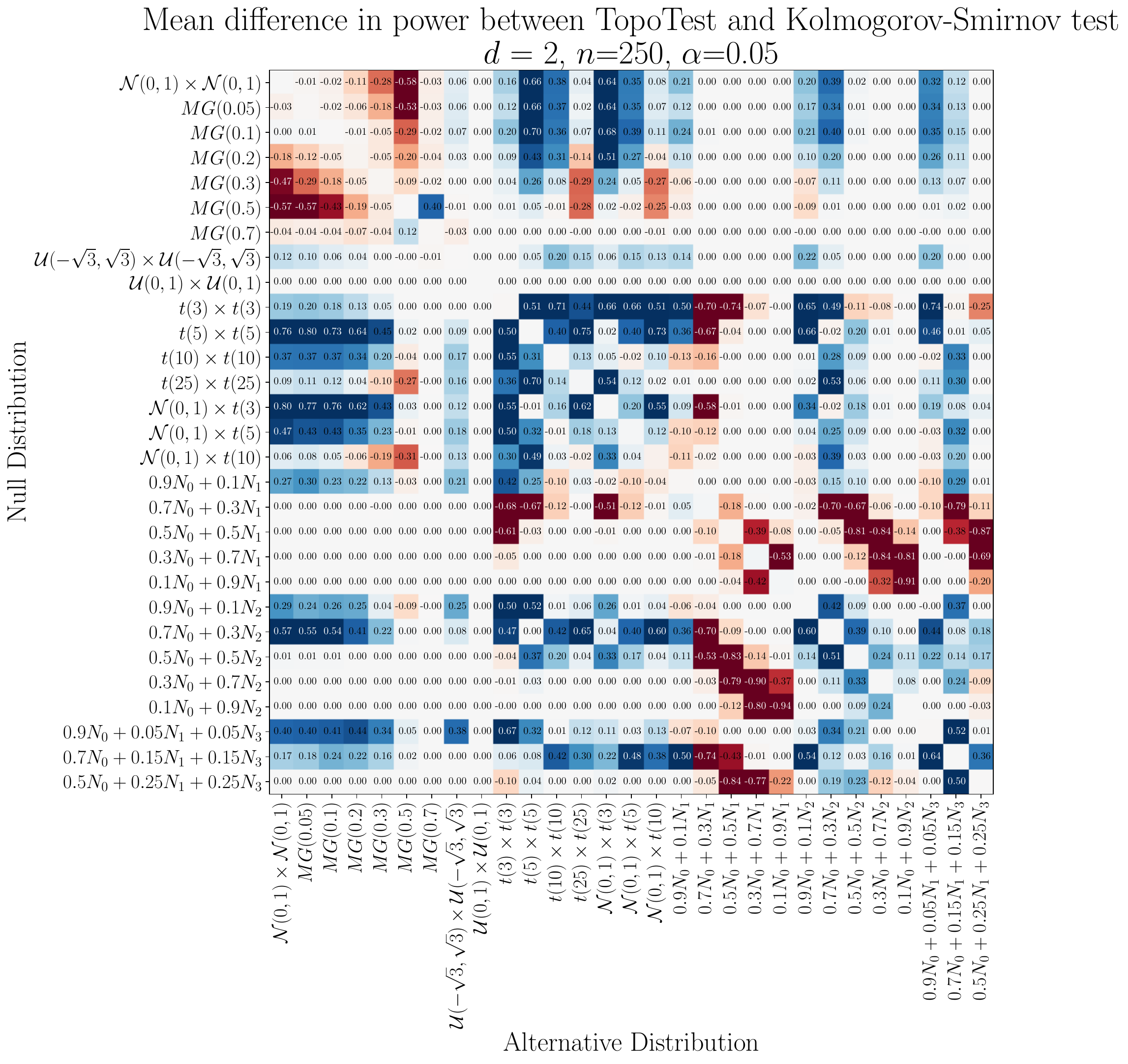}
\end{figure*}  

\begin{figure*}
    \caption{The same as Figure \ref{fig:power_matrix_1d} but for three-dimensional distributions. 
    Results based on $K=250$ Monte Carlo realizations. 
    Average power is $0.708$ $(0.824)$ for TopoTest and $0.602$ $(0.763)$ 
    for Kolmogorov-Smirnov for $n=100$ ($n=250$).
    }
    \label{fig:power_matrix_3d}
    \includegraphics[width=0.49\linewidth]{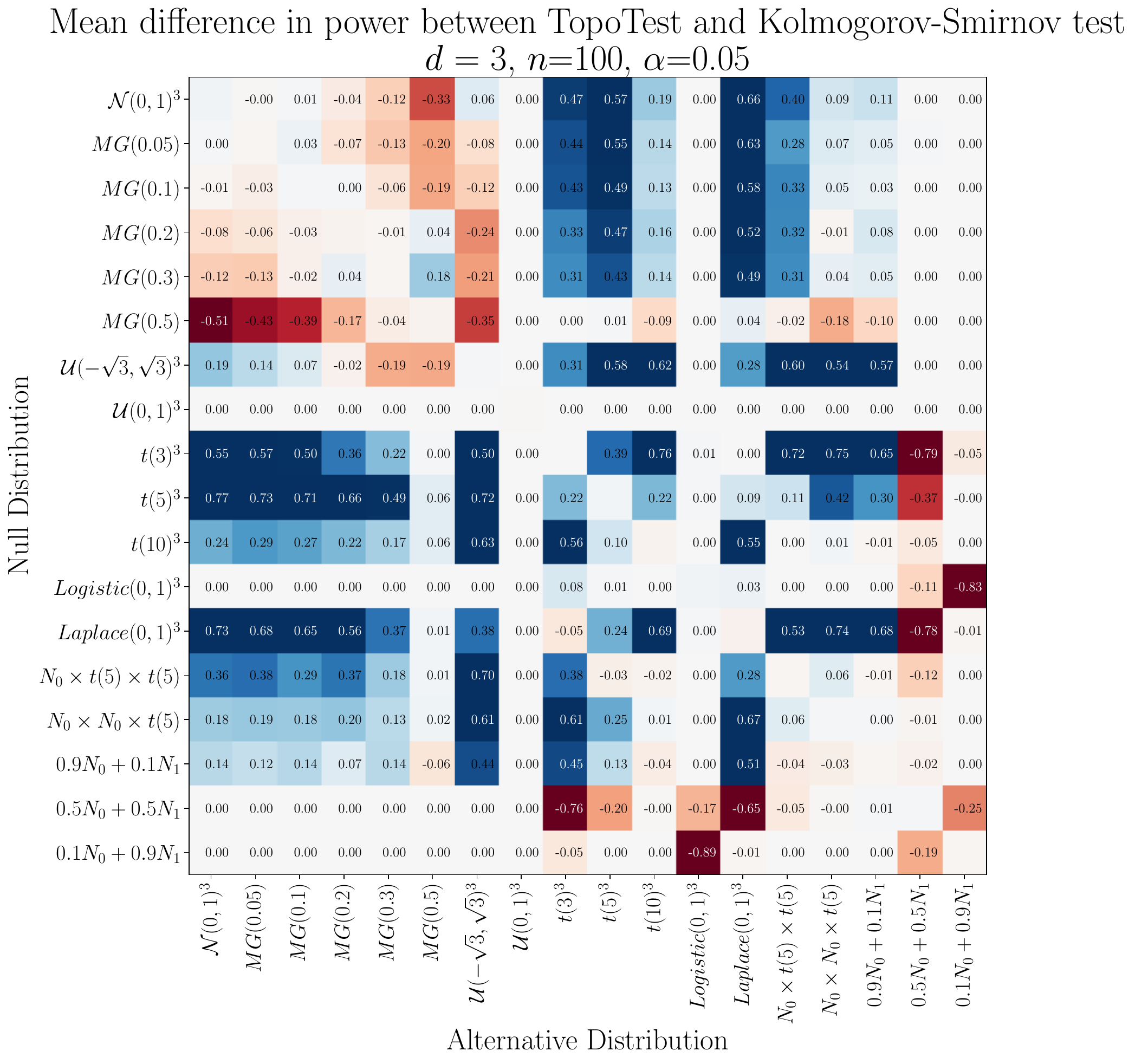}
    \includegraphics[width=0.49\linewidth]{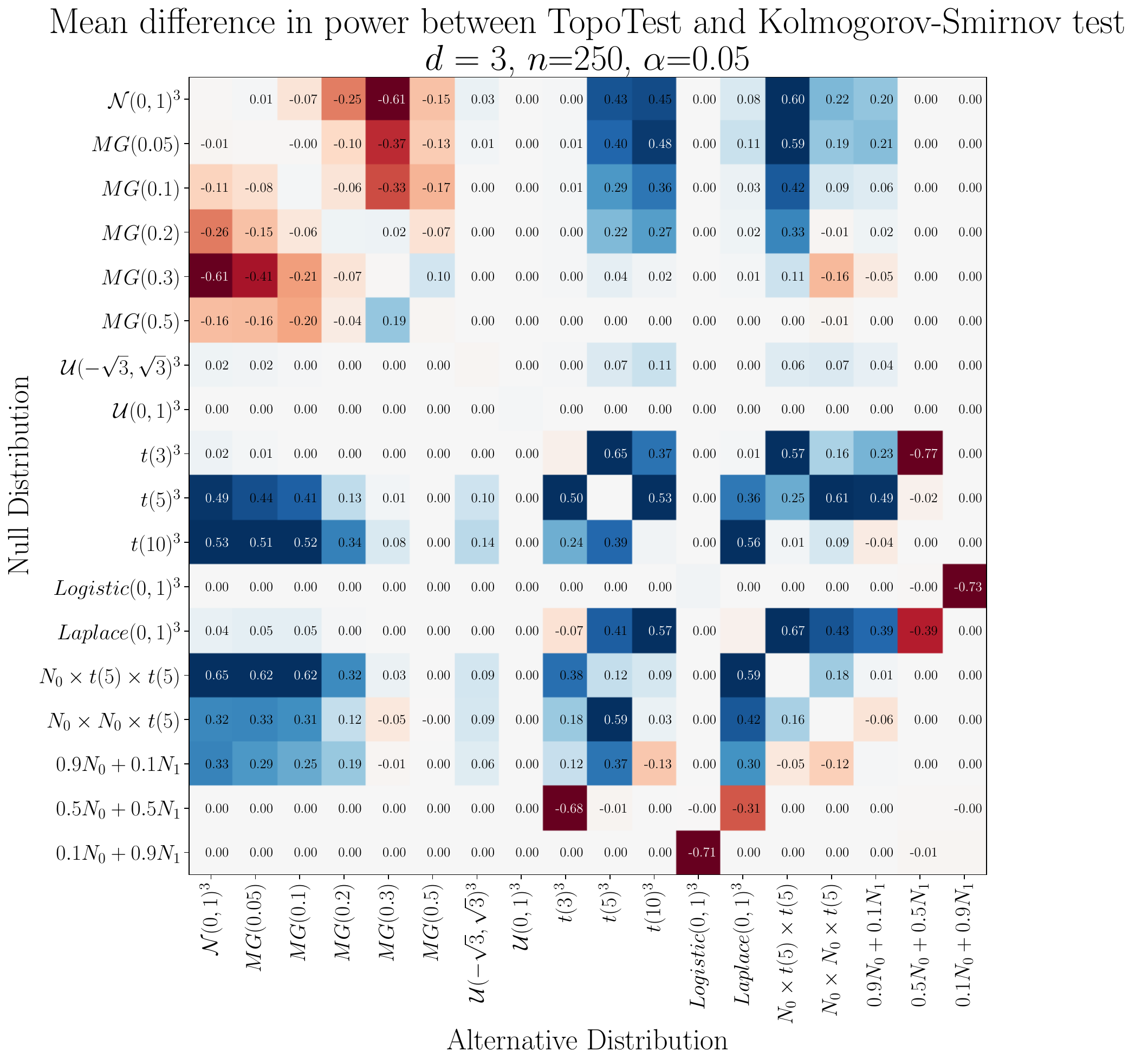}
\end{figure*}  

The analysis was conducted also dimension $d=5$ as can be seen in Figure \ref{fig:power_matrix_5d}.
For $d>3$ the Kolmogorov-Smirnov test was not preformed due to too long computation time, hence results for
TopoTest are presented only as this method provided feasible computational complexity.

As can be seen the TopoTest stayed sensitive enough to differentiate between multivariate normal distribution and Cartesian products of involving Student's t-distribution and standard normal 
as marginals, especially given
that considered samples sizes are low for such high dimensional spaces.

The heatmap presented in Figure \ref{fig:power_matrix_2d} reveals several prominent red-blocks, i.e. combinations of null and alternative distributions for which the power of the TopoTest is significantly lower than the power of KS test: e.g. the combination $G=p\mathcal{N}(0, 1) + (1-p)\mathcal{N}(0, 2)$ and $F=p\mathcal{N}(0, 1) + (1-p)\mathcal{N}(\mu, 2), \mu=1$. This observation is related to the Lemma 5.1 by Vishwanath \textit{et al.} \cite{vishwanath_limits_2022} (c.f. Example~\ref{ex:counterexample}) regarding equivalence in expected ECCs. Although the distributions $F$ and $G$ are not Euler equivalent and the condition (\ref{eqn:beta_condition}) is not met but only approximately, the expected ECCs are quite similar for small values of $\mu$ making \red{them} hard to distinguish by the TopoTest test statistic (\ref{eqn:threshold_definition}). Similar situations holds for trivariate distributions as shown in Figure \ref{fig:power_matrix_3d}. 

\begin{figure*}
    \caption{
        Average power of TopoTest for five dimension distributions,
        for sample sizes $n=250$ and $n=500$.
        Results based on $K=1000$ Monte Carlo realizations.
        }
    \label{fig:power_matrix_5d}
    \includegraphics[width=0.49\linewidth]{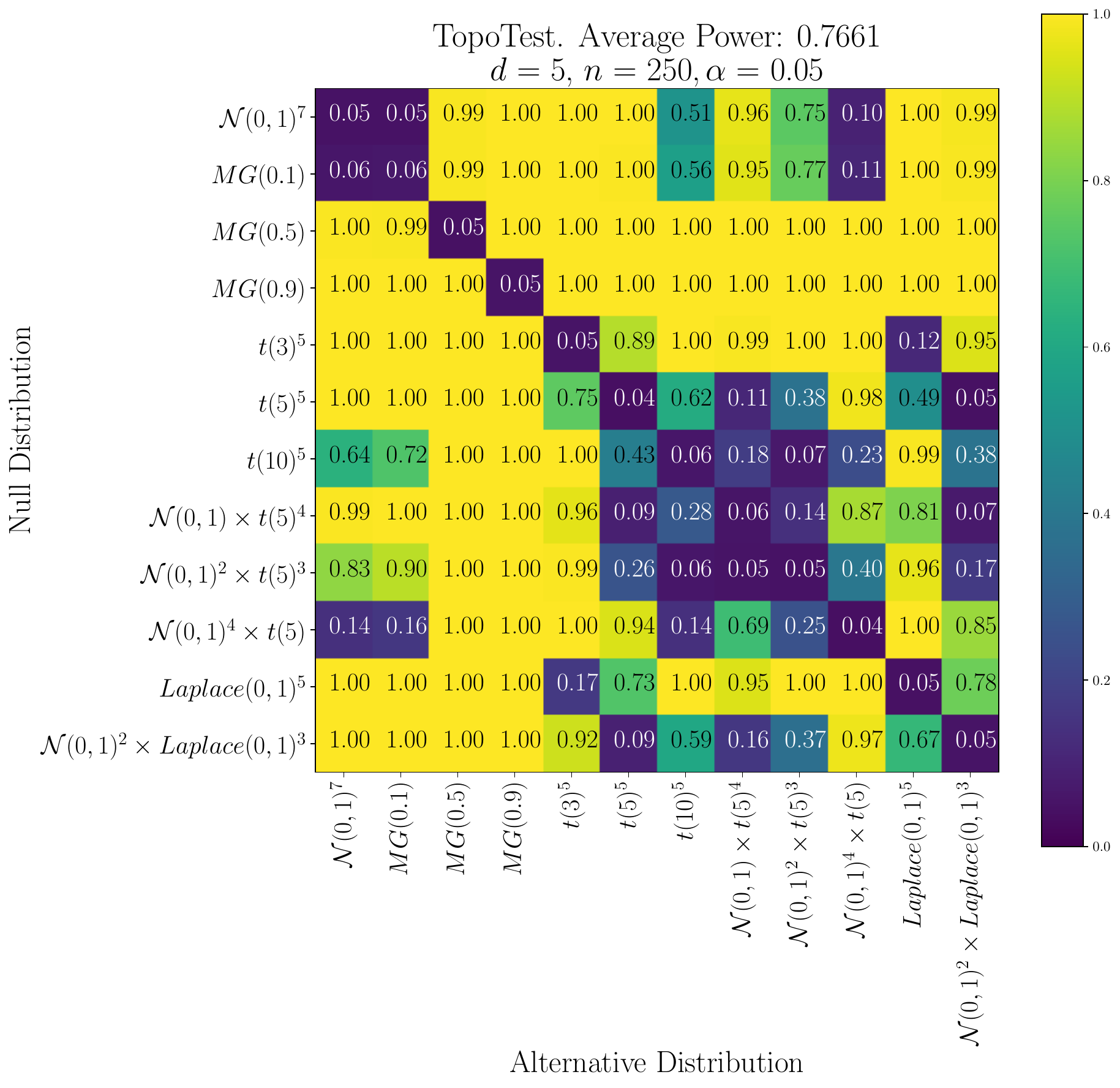}
    \includegraphics[width=0.49\linewidth]{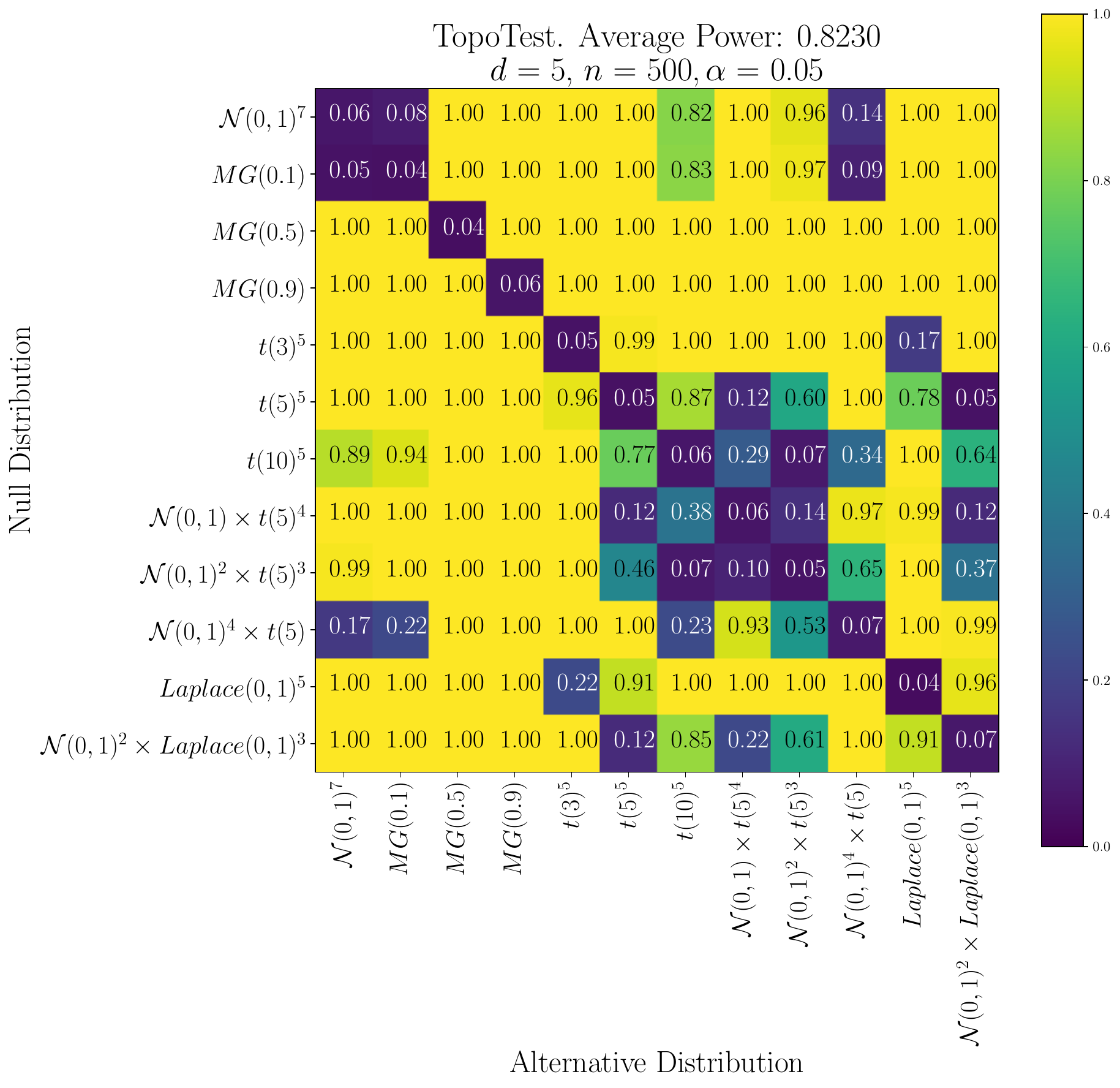}
\end{figure*} 
\subsection{Dependence of the test power on sample size}
The dependence of the power of TopoTest and Kolmogorov-Smirnov tests on the sample size $n$ 
is shown in Figure \ref{fig:power_sample_size} for random samples in dimensions $d=1, 2, 3$. To compute average power, all combinations of null and alternative distributions,
as considered in Figures \ref{fig:power_matrix_1d}, \ref{fig:power_matrix_2d} and \ref{fig:power_matrix_3d}, were taken into account,
except alternative being the same as null distribution. 
In all cases, the average power increased with sample size as expected. 
In case of univariate distribution (leftmost panel in Figure \ref{fig:power_sample_size}) the results obtained using Cram\'er-von Mises test
were added for completeness. The overall performance of this test is similar to Kolmogorov-Smirnov, hence detailed analysis was omitted. 
The TopoTest however provides higher average power for all sample sizes regardless of the data dimension. 
It should be noted that powers presented in Figure \ref{fig:power_sample_size} should not be directly compared across different dimensions as the actual value depends on the list of considered distributions
which is different for each dimension.  

\begin{figure*}
    \caption{Average power of the TopoTest (black curve) and Kolmogorov-Smirnov (red curve) as a function of sample size $n$ for dimensions $d=1, 2, 3$.
    In case of $d=1$ the average power of Cram\'er-von Mises (green curve) test was shown as well. To guide an eye the data points are connect by lines.}
    \label{fig:power_sample_size}
    \includegraphics[width=0.32\linewidth]{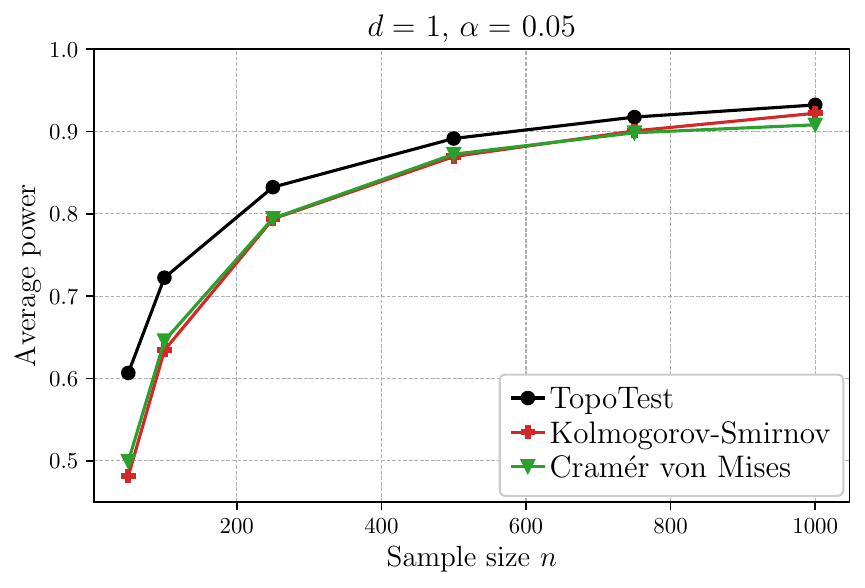}
    \includegraphics[width=0.32\linewidth]{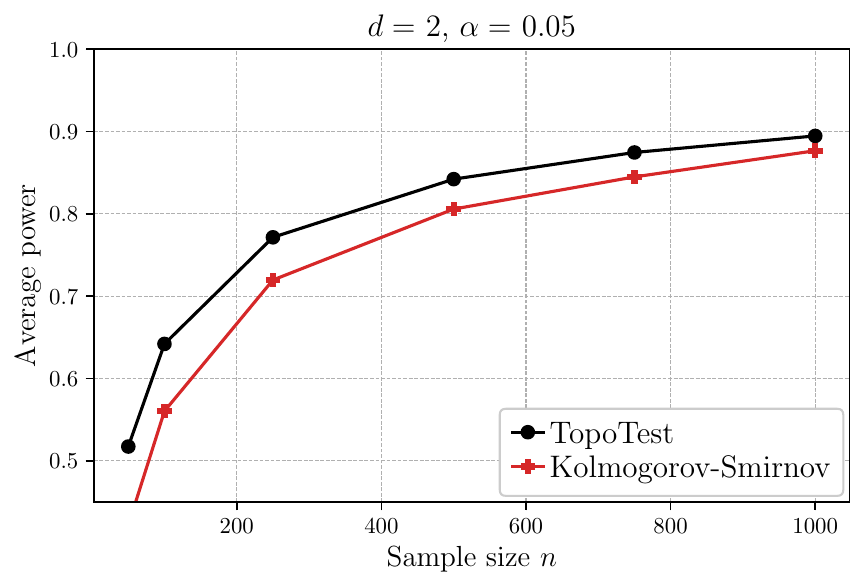}
    \includegraphics[width=0.32\linewidth]{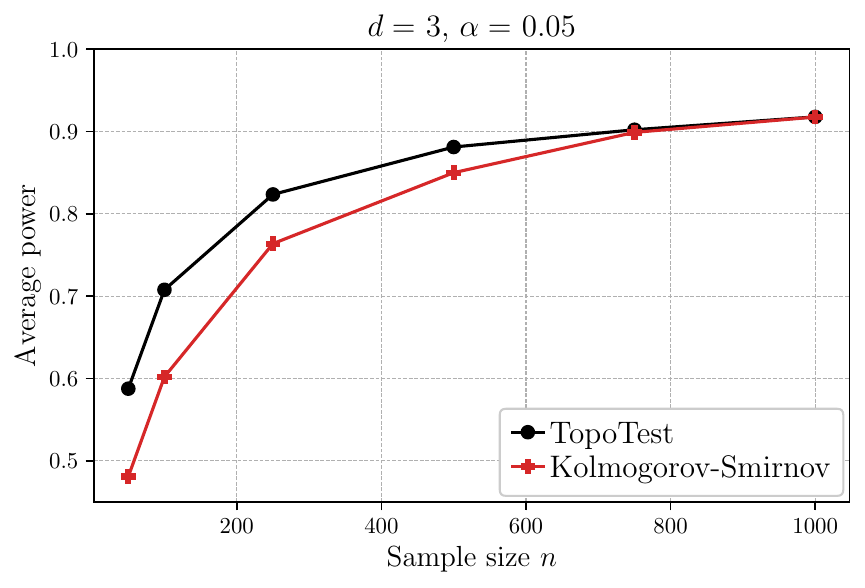}
\end{figure*} 
\section{Numerical experiments, two-sample problem}
\label{sec:NumericalExperiments2sample}
A numerical study was conducted also for two-sample problems, in which Algorithm \ref{algo:2SampleTesting} was applied.
The two-sample problem was considered for completeness purpose as practical application is limited by high computational costs,
therefore results presented here are restricted to comparison of empirical power of two-sample
TopoTest and Kolmogorov-Smirnov tests in $d=1$ (cf. Table \ref{tab:twosample_1d}) and $d=2$ (cf. Table \ref{tab:twosample_2d}).
Simulations showed that in both cases the TopoTest outperformed the Kolmogorov-Smirnov test: in the vast majority of examined cases
the power of the former is greater. Moreover, the average power for TopoTest is greater \red{than the} corresponding average power of Kolmogorov-Smirnov test for all sample sizes $n$. 
\begin{table*}
    \caption{Empirical powers of the two-sample TopoTest for different alternative distributions and sample sizes $n$ --
    the null distribution is standard normal $\mathcal{N}(0,1)$.
    Corresponding powers of Kolmogorov-Smirnov tests are given in parenthesis for comparison -- higher result is given in bold for easier comparison. 
    Results for the significance level $\alpha=0.05$ 
    Empirical powers estimated based on $K=500$ Monte Carlo realizations.}
    \label{tab:twosample_1d}
\begin{tabular}{lrrrrr}
\hline
& \multicolumn{5}{c}{Sample size $n$} \\
\cmidrule{2-6}
Second Sample Distribution &      \red{30}&               50  &                     100 &                     250 &                     500 \\ \hline 
$\mathcal{N}(0,  0.50)$                           &  \textbf{0.694} (0.218) &  \textbf{0.890} (0.358) &  \textbf{0.996} (0.816) &           1.000 (1.000) &           1.000 (1.000) \\
$\mathcal{N}(0, 0.75)$                            &  \textbf{0.202} (0.054) &  \textbf{0.290} (0.070) &  \textbf{0.462} (0.114) &  \textbf{0.858} (0.376) &  \textbf{0.938} (0.790) \\
$\mathcal{N}(0, 1.25)$                            &  \textbf{0.188} (0.056) &  \textbf{0.166} (0.040) &  \textbf{0.300} (0.110) &  \textbf{0.682} (0.228) &  \textbf{0.822} (0.474) \\
$\mathcal{N}(0, 1.5)$                             &  \textbf{0.366} (0.084) &  \textbf{0.468} (0.124) &  \textbf{0.792} (0.240) &  \textbf{0.984} (0.782) &  0.984 (\textbf{0.994}) \\
Laplace$(0, 1)$                                   &  \textbf{0.154} (0.036) &  \textbf{0.204} (0.046) &  \textbf{0.458} (0.068) &  \textbf{0.892} (0.076) &  \textbf{0.992} (0.154) \\
$\mathcal{U}(-\sqrt{3}, \sqrt{3})$                          &  \textbf{0.092} (0.042) &  \textbf{0.094} (0.054) &  \textbf{0.204} (0.082) &  \textbf{0.756} (0.274) &  \textbf{0.998} (0.592) \\
$\mathcal{U}(0, 1)$                                         &  \textbf{1.000} (0.970) &           1.000 (1.000) &           1.000 (1.000) &           1.000 (1.000) &           1.000 (1.000) \\
$t(3)$                                            &  \textbf{0.230} (0.024) &  \textbf{0.276} (0.058) &  \textbf{0.564} (0.046) &  \textbf{0.930} (0.084) &  \textbf{0.956} (0.220) \\
$t(5)$                                            &  \textbf{0.116} (0.038) &  \textbf{0.124} (0.030) &  \textbf{0.238} (0.036) &  \textbf{0.568} (0.036) &  \textbf{0.844} (0.072) \\
$t(10)$                                           &  \textbf{0.088} (0.048) &  \textbf{0.082} (0.030) &  \textbf{0.098} (0.028) &  \textbf{0.204} (0.062) &  \textbf{0.370} (0.052) \\
$t(25)$                                           &  \textbf{0.102} (0.036) &  \textbf{0.062} (0.028) &  \textbf{0.064} (0.040) &  \textbf{0.094} (0.040) &  \textbf{0.110} (0.046) \\
Cauchy$(0, 1)$                                    &  \textbf{0.784} (0.060) &  \textbf{0.894} (0.118) &  \textbf{0.914} (0.350) &  0.906 (\textbf{0.956}) &  0.916 (\textbf{1.000}) \\
Logistic$(0, 1)$                                  &  \textbf{0.494} (0.096) &  \textbf{0.712} (0.164) &  \textbf{0.948} (0.392) &  \textbf{0.994} (0.942) &  0.998 (\textbf{1.000}) \\
0.9$\mathcal{N}(0, 1)$ + 0.1$\mathcal{N}(0, 0.5)$ &  \textbf{0.072} (0.036) &  \textbf{0.092} (0.038) &  \textbf{0.076} (0.048) &  \textbf{0.104} (0.078) &  0.082 (\textbf{0.086}) \\
0.7$\mathcal{N}(0, 1)$ + 0.3$\mathcal{N}(0, 0.5)$ &  \textbf{0.124} (0.048) &  \textbf{0.122} (0.068) &  \textbf{0.188} (0.098) &  \textbf{0.278} (0.206) &  0.266 (\textbf{0.430}) \\
0.5$\mathcal{N}(0, 1)$ + 0.5$\mathcal{N}(0, 0.5)$ &  \textbf{0.190} (0.072) &  \textbf{0.242} (0.096) &  \textbf{0.456} (0.178) &  \textbf{0.638} (0.550) &  0.610 (\textbf{0.938}) \\
0.3$\mathcal{N}(0, 1)$ + 0.7$\mathcal{N}(0, 0.5)$ &  \textbf{0.334} (0.088) &  \textbf{0.490} (0.176) &  \textbf{0.810} (0.380) &  \textbf{0.950} (0.922) &  0.822 (\textbf{1.000}) \\
0.1$\mathcal{N}(0, 1)$ + 0.9$\mathcal{N}(0, 0.5)$ &  \textbf{0.568} (0.172) &  \textbf{0.782} (0.282) &  \textbf{0.954} (0.674) &  0.992 (\textbf{0.998}) &  0.958 (\textbf{1.000}) \\
0.9$\mathcal{N}(0, 1)$ + 0.1$\mathcal{N}(0, 2)$   &  \textbf{0.114} (0.040) &  \textbf{0.102} (0.038) &  \textbf{0.090} (0.048) &  \textbf{0.220} (0.044) &  \textbf{0.402} (0.076) \\
0.7$\mathcal{N}(0, 1)$ + 0.3$\mathcal{N}(0, 2)$   &  \textbf{0.184} (0.030) &  \textbf{0.272} (0.048) &  \textbf{0.424} (0.058) &  \textbf{0.814} (0.146) &  \textbf{0.980} (0.338) \\
0.5$\mathcal{N}(0, 1)$ + 0.5$\mathcal{N}(0, 2)$   &  \textbf{0.284} (0.038) &  \textbf{0.502} (0.084) &  \textbf{0.758} (0.152) &  \textbf{0.992} (0.476) &  \textbf{0.996} (0.934) \\
0.3$\mathcal{N}(0, 1)$ + 0.7$\mathcal{N}(0, 2)$   &  \textbf{0.458} (0.100) &  \textbf{0.722} (0.126) &  \textbf{0.944} (0.344) &  \textbf{1.000} (0.906) &           1.000 (1.000) \\
0.1$\mathcal{N}(0, 1)$ + 0.9$\mathcal{N}(0, 2)$   &  \textbf{0.604} (0.118) &  \textbf{0.822} (0.276) &  \textbf{0.988} (0.630) &  0.998 (\textbf{1.000}) &  0.996 (\textbf{1.000}) \\
0.9$\mathcal{N}(0, 1)$ + 0.1$\mathcal{N}(1, 2)$   &  \textbf{0.086} (0.050) &  \textbf{0.120} (0.042) &  \textbf{0.128} (0.042) &  \textbf{0.286} (0.074) &  \textbf{0.548} (0.134) \\
0.7$\mathcal{N}(0, 1)$ + 0.3$\mathcal{N}(1, 2)$   &  \textbf{0.210} (0.064) &  \textbf{0.280} (0.108) &  \textbf{0.540} (0.190) &  \textbf{0.906} (0.630) &  \textbf{0.974} (0.958) \\
0.5$\mathcal{N}(0, 1)$ + 0.5$\mathcal{N}(1, 2)$   &  \textbf{0.354} (0.174) &  \textbf{0.552} (0.330) &  \textbf{0.814} (0.692) &  \textbf{1.000} (0.990) &  0.996 (\textbf{1.000}) \\
0.3$\mathcal{N}(0, 1)$ + 0.7$\mathcal{N}(1, 2)$   &  \textbf{0.556} (0.380) &  \textbf{0.744} (0.684) &  \textbf{0.972} (0.952) &           1.000 (1.000) &  0.998 (\textbf{1.000}) \\
0.1$\mathcal{N}(0, 1)$ + 0.9$\mathcal{N}(1, 2)$   &  \textbf{0.688} (0.616) &  0.888 (\textbf{0.892}) &  0.990 (\textbf{1.000}) &  0.998 (\textbf{1.000}) &           1.000 (1.000) \\ \hline
Average Power                                     &  \textbf{0.333} (0.135) &  \textbf{0.428} (0.193) &  \textbf{0.577} (0.315) &  \textbf{0.752} (0.531) &  \textbf{0.806} (0.653) \\ \hline
\end{tabular}
\end{table*}

\begin{table*}
    \caption{The same as Table \ref{tab:twosample_1d} but for $d=2$. Standard bivariate normal is used as a null distribution. The MG distribution is defined in (\ref{eqn:MG}).}
    \label{tab:twosample_2d}
\begin{tabular}{lrrrrr}
\hline
& \multicolumn{5}{c}{Sample size $n$} \\
\cmidrule{2-6}
Second Sample Distribution &    \red{30}&                 50  &                     100 &                     250 &                     500 \\ \hline 
$MG(0.05)$                                         &  \textbf{0.084} (0.058) &  0.052 (\textbf{0.066}) &  \textbf{0.080} (0.066) &  \textbf{0.066} (0.058) &  0.058 (\textbf{0.086}) \\
$MG(0.1)$                                          &  0.060 (\textbf{0.072}) &  \textbf{0.074} (0.064) &  \textbf{0.078} (0.066) &  0.036 (\textbf{0.076}) &  0.060 (\textbf{0.092}) \\
$MG(0.2)$                                          &  \textbf{0.078} (0.062) &  \textbf{0.080} (0.074) &  0.052 (\textbf{0.074}) &  0.060 (\textbf{0.124}) &  0.074 (\textbf{0.196}) \\
$MG(0.3)$                                          &  \textbf{0.082} (0.062) &  0.054 (\textbf{0.066}) &  0.064 (\textbf{0.114}) &  0.080 (\textbf{0.236}) &  0.100 (\textbf{0.472}) \\
$MG(0.5)$                                          &  0.086 (\textbf{0.092}) &  0.100 (\textbf{0.136}) &  0.136 (\textbf{0.264}) &  0.236 (\textbf{0.666}) &  0.368 (\textbf{0.976}) \\
$MG(0.7)$                                          &  \textbf{0.142} (0.132) &  0.226 (\textbf{0.254}) &  0.374 (\textbf{0.582}) &  0.764 (\textbf{0.986}) &  0.958 (\textbf{1.000}) \\
$\mathcal{U}(-\sqrt{3},\sqrt{3}) \times \mathcal{U}(-\sqrt{3},\sqrt{3})$ &  \textbf{0.096} (0.090) &  0.144 (\textbf{0.156}) &  \textbf{0.346} (0.244) &  \textbf{0.944} (0.584) &  \textbf{1.000} (0.930) \\
$\mathcal{U}(0, 1) \times \mathcal{U}(0, 1)$                           &           1.000 (1.000) &           1.000 (1.000) &           1.000 (1.000) &           1.000 (1.000) &           1.000 (1.000) \\
$t(3) \times t(3)$                                 &  \textbf{0.328} (0.042) &  \textbf{0.494} (0.068) &  \textbf{0.792} (0.140) &  \textbf{0.990} (0.412) &  \textbf{0.980} (0.868) \\
$t(5) \times t(5)$                                 &  \textbf{0.196} (0.050) &  \textbf{0.224} (0.066) &  \textbf{0.412} (0.086) &  \textbf{0.806} (0.150) &  \textbf{0.982} (0.298) \\
$t(10) \times t(10)$                               &  \textbf{0.120} (0.054) &  \textbf{0.110} (0.050) &  \textbf{0.144} (0.064) &  \textbf{0.274} (0.066) &  \textbf{0.546} (0.108) \\
$t(25) \times t(25)$                               &  \textbf{0.068} (0.040) &  \textbf{0.076} (0.054) &  \textbf{0.064} (0.058) &  \textbf{0.080} (0.078) &  \textbf{0.130} (0.052) \\
$\mathcal{N}(0, 1) \times t(3)$                    &  \textbf{0.156} (0.052) &  \textbf{0.160} (0.064) &  \textbf{0.288} (0.076) &  \textbf{0.598} (0.128) &  \textbf{0.866} (0.190) \\
$\mathcal{N}(0, 1) \times t(5)$                    &  \textbf{0.070} (0.052) &  \textbf{0.112} (0.064) &  \textbf{0.180} (0.052) &  \textbf{0.304} (0.056) &  \textbf{0.518} (0.118) \\
$\mathcal{N}(0, 1) \times t(10)$                   &  \textbf{0.082} (0.034) &  0.054 (\textbf{0.062}) &  \textbf{0.090} (0.054) &  \textbf{0.088} (0.062) &  \textbf{0.152} (0.086) \\
$0.9N_0 + 0.1N_1$                                  &  \textbf{0.098} (0.052) &  \textbf{0.102} (0.080) &  \textbf{0.136} (0.068) &  \textbf{0.296} (0.154) &  \textbf{0.454} (0.204) \\
$0.7N_0 + 0.3N_1$                                  &  \textbf{0.280} (0.120) &  \textbf{0.376} (0.178) &  \textbf{0.628} (0.414) &  \textbf{0.956} (0.900) &           0.998 (0.998) \\
$0.5N_0 + 0.5N_1$                                  &  \textbf{0.508} (0.292) &  \textbf{0.694} (0.588) &  0.914 (\textbf{0.922}) &           1.000 (1.000) &           1.000 (1.000) \\
$0.3N_0 + 0.7N_1$                                  &  \textbf{0.712} (0.636) &  0.892 (\textbf{0.900}) &           1.000 (1.000) &           1.000 (1.000) &           1.000 (1.000) \\
$0.1N_0 + 0.9N_1$                                  &  0.820 (\textbf{0.888}) &  0.972 (\textbf{0.996}) &           1.000 (1.000) &           1.000 (1.000) &           1.000 (1.000) \\
$0.9N_0 + 0.1N_2$                                  &  \textbf{0.108} (0.058) &  \textbf{0.096} (0.076) &  \textbf{0.118} (0.064) &  \textbf{0.190} (0.066) &  \textbf{0.332} (0.086) \\
$0.7N_0 + 0.3N_2$                                  &  \textbf{0.224} (0.064) &  \textbf{0.250} (0.090) &  \textbf{0.496} (0.106) &  \textbf{0.840} (0.278) &  \textbf{0.994} (0.658) \\
$0.5N_0 + 0.5N_2$                                  &  \textbf{0.352} (0.080) &  \textbf{0.582} (0.140) &  \textbf{0.858} (0.282) &  \textbf{1.000} (0.806) &  \textbf{1.000} (0.996) \\
$0.3N_0 + 0.7N_2$                                  &  \textbf{0.636} (0.140) &  \textbf{0.810} (0.318) &  \textbf{0.970} (0.664) &  \textbf{1.000} (0.998) &           1.000 (1.000) \\
$0.1N_0 + 0.9N_2$                                  &  \textbf{0.798} (0.246) &  \textbf{0.918} (0.574) &  \textbf{1.000} (0.914) &           1.000 (1.000) &           1.000 (1.000) \\
$0.9N_0 + 0.05N_1 + 0.05N_3$                       &  \textbf{0.088} (0.050) &  \textbf{0.094} (0.076) &  \textbf{0.142} (0.076) &  \textbf{0.304} (0.102) &  \textbf{0.476} (0.130) \\
$0.7N_0 + 0.15N_1 + 0.15N_3$                       &  \textbf{0.234} (0.084) &  \textbf{0.400} (0.108) &  \textbf{0.662} (0.188) &  \textbf{0.968} (0.502) &  \textbf{1.000} (0.904) \\
$0.5N_0 + 0.25N_1 + 0.25N_3$                       &  \textbf{0.588} (0.128) &  \textbf{0.786} (0.242) &  \textbf{0.966} (0.554) &  \textbf{1.000} (0.990) &           1.000 (1.000) \\ \hline
Average Power                                      &  \textbf{0.289} (0.169) &  \textbf{0.355} (0.236) &  \textbf{0.464} (0.328) &  \textbf{0.603} (0.481) &  \textbf{0.680} (0.587) \\
\hline
\end{tabular}
\end{table*}

As in the Section~\ref{sec:NumericalExperiments_one_sample}, the above collection of distribution is examined also in all-to-all settings. The difference in average power between TopoTest and Kolmogorov-Smirnov tests are shown are shown Figure~\ref{fig:2sample_diff_heatmap}. 
 
\begin{figure*}
    \caption{
        Difference in average power of two-sample TopoTest and two-sample Kolmogorov-Smirnov tests
        for univariate (left panel) and bivariate (right panel) distributions. In both cases sample sizes were $n=100$ and $K=500$ Monte Carlo realizations were performed to estimate the average power. Average power of TopoTest is $0.643$ ($0.537$) while for Kolmogorov-Smirnov it is $0.453$ ($0.437$) in $d=1$ ($d=2$).
    }
    \label{fig:2sample_diff_heatmap}
    \includegraphics[width=0.49\linewidth]{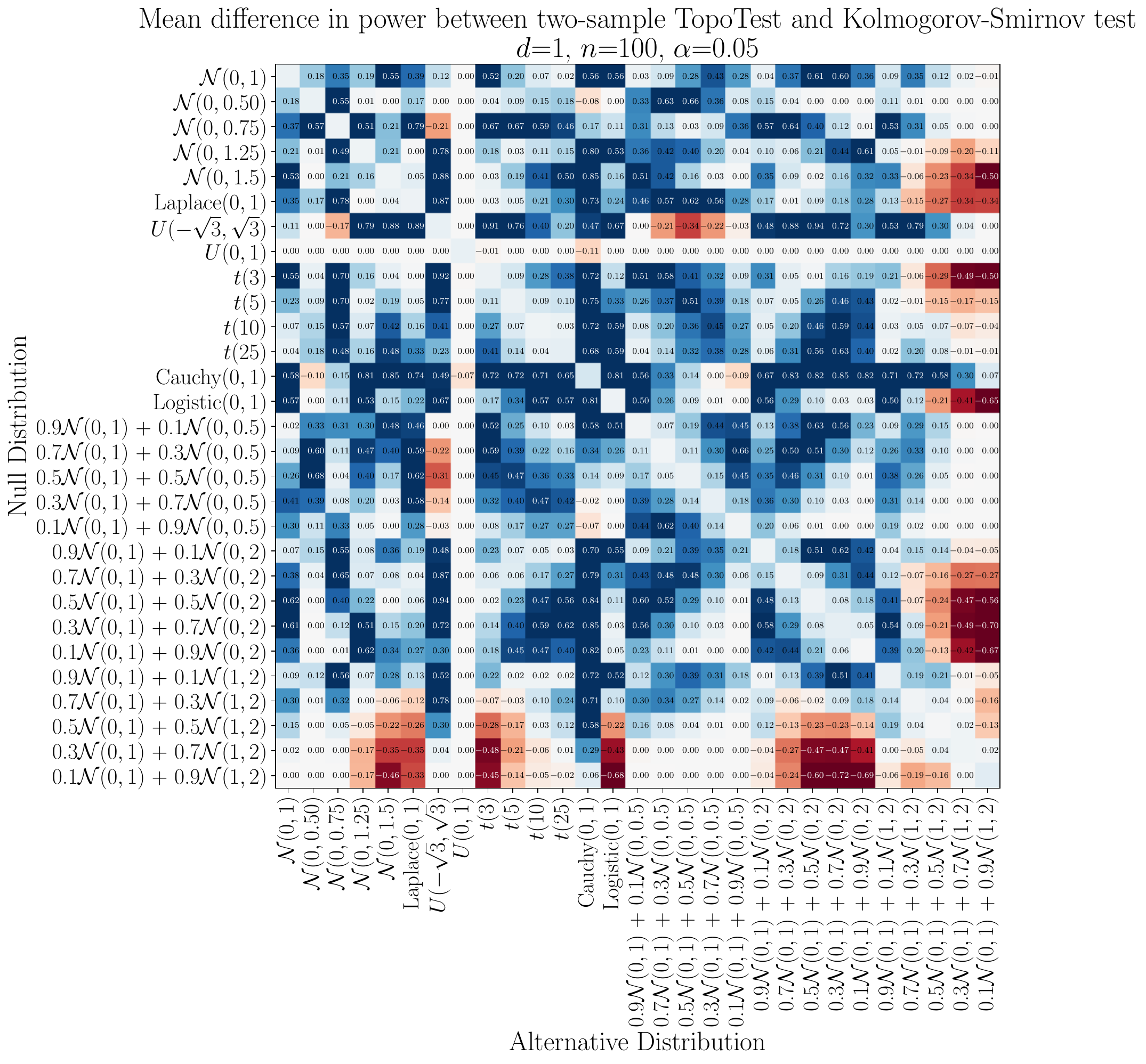}
    \includegraphics[width=0.49\linewidth]{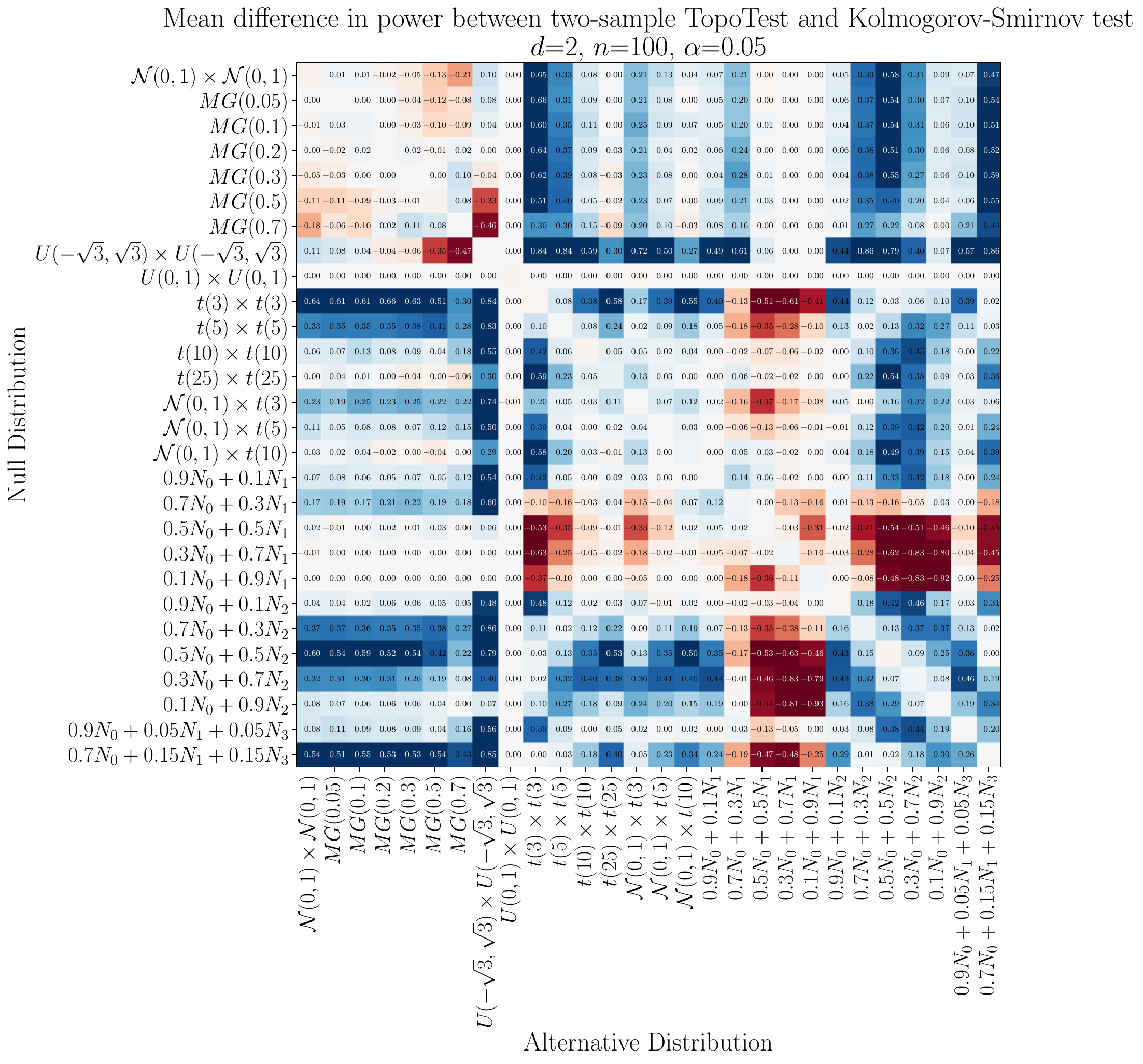}
\end{figure*} 

\red{
\section{Real data analysis}
In this section, we show two exemplary applications of the developed method to the analysis of real data.

First, we consider Fisher's \textit{Iris data} in the one-sample setting. 
This data includes three multivariate samples corresponding to three different species of Iris, i.e. \textit{Iris setosa}, \textit{Iris virginica}, and \textit{Iris versicolor}. 
There are 50 samples from each species, containing four measurements of the flower. 
We would like to determine if the distribution of each species follows a four-dimensional normal distribution. This can be formulated as a one-sample problem, where $G$ is the distribution of a sample, and $F$ is the specified four-dimensional normal distribution. $F$ involves an unknown mean vector $\mathbf{\mu}$ and unknown covariance matrix $\Sigma$. For each species, $\mathbf{\mu}$ and $\Sigma$ are estimated by sample mean and sample covariance matrix. Our one-sample test for testing $H_0: G = F$ against $H_1: G \ne F$ gave $p$-values of 0.057, 0.569 and 0.999 for \textit{Iris setosa}, \textit{Iris virginica} and \textit{Iris versicolor}, respectively. These $p$-values indicate that, at significance level 0.05, $H_0$ should not be rejected for each of the Iris species. 
However, when the same procedure is applied to the entire Iris dataset (i.e. without splitting into species), the $p$-value is $<10^{-4}$, hence the null hypothesis is to be rejected, which indicates that multivariate normal distribution does not fit whole Iris dataset. The conclusions are consistent with the literature \cite{Dhar2014}.

\begin{figure*}
    \caption{
    Spatial distribution of selected social services with the municipality of Rennes, France (left panel),
    corresponding Euler  curves (right panel).
    }
    \label{fig:rennes}
    \includegraphics[width=0.95\linewidth]{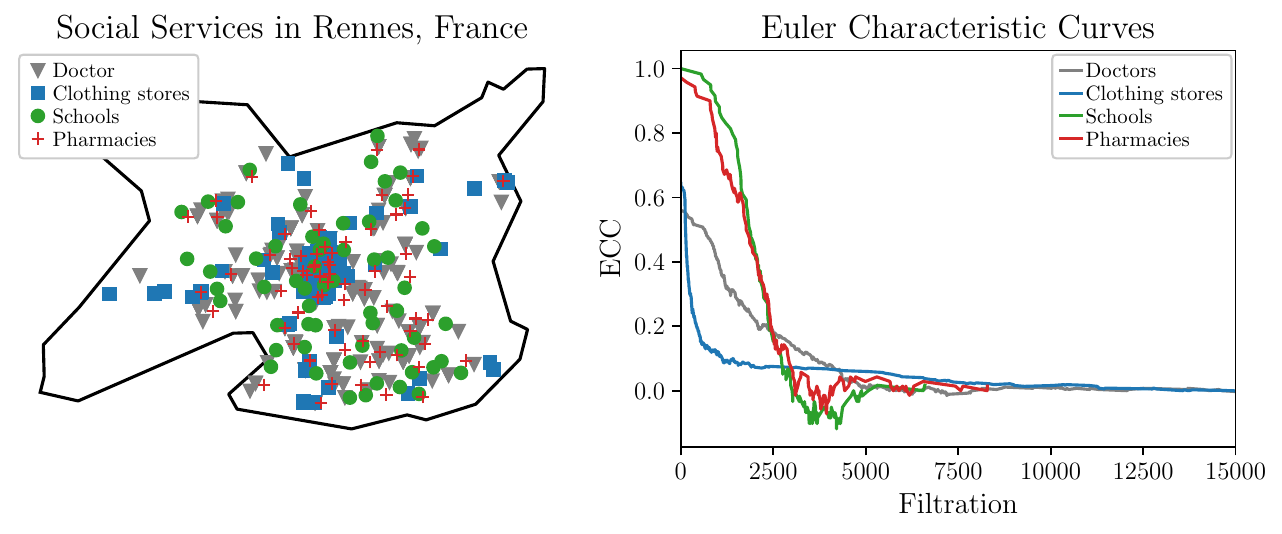}
\end{figure*} 

In our second example, we consider a dataset introduced in~\cite{Floch2018} consisting of a collection of  geographic locations of four distinct social services, i.e. doctor offices, clothing stores, schools, and pharmacies, in the municipality of Rennes, France.
It is visualized in Figure~\ref{fig:rennes} as a map. 
The two-sample TopoTest is used to detect if there are any significant differences in the distribution of those facilities. The test was conducted for all possible pairs. The $p$-values for all tests involving the distribution of clothing stores were below $10^{-4}$, meaning that in the Algorithm \ref{algo:2SampleTesting} in all of $K=10000$ iterations $d_{(p)} < D$, which indicates
that their geographic distribution is significantly different from the distribution of doctor offices, schools, and pharmacies. Such conclusion is supported by the plots of corresponding ECCs (c.f. Figure \ref{fig:rennes}, right panel): The curve computed for clothing stores (blue) is visually distinct from other curves. 
Contrary, no statistical differences were observed between the distribution of pharmacies and the distribution of schools -- the $p$-value of the TopoTest is 0.306. All the above conclusions are in agreement with the previous findings about that dataset made using the Fasano-Franceschi test \cite{Fasano1987, puritz2022fasanofranceschinitest}. 
However, in addition to that, the TopoTest rejects the hypothesis of equal geographical distributions of doctor offices vs. pharmacies and doctor offices vs. schools (in both cases the $p$-value is below $10^{-4}$), while the Fasano-Franceschi does not ($p$-value 0.881 and 0.435, respectively as computed using \texttt{fasano.franceschini.test} R package). 
This is an interesting observation in the context of previously discussed simulation study results, where we show that TopoTest is more powerful than the Kolmogorov-Smirnov test (closely related to the Fasano-Franceschi test) and hence more often correctly rejects the null hypothesis.

} 
\section{Discussion}
\label{sec:discussion}
Using Euler characteristic curves, we introduced a new framework for goodness-of-fit testing in arbitrary dimensions.
In addition, we provide a theoretical justification of the method.
Although the distribution of the test statistic is unknown for finite $n$, and contrary to the Kolmogorov-Smirnov test, depends on $F$, 
the asymptotic distribution is given by (\ref{eqn:GaussianProcess}), while theorem \ref{thm:TypeIIErrorAsymptotic} 
provides an upper bound on the type II error.

\red{A simulation} study was conducted to address the power of the TopoTest in comparison with Kolmogorov-Smirnov test.
A one- and two-sample setting was considered.
In both cases, the TopoTest in many cases yielded better performance than Kolmogorov-Smirnov.
It should be however highlighted that Kolmogorov-Smirnov test and TopoTests operate in slightly different frameworks -- the former in capable 
to distinguish between distributions that differ e.g. in location parameter while the TopoTests are insensitive to the distribution shifts,
rotations, reflections as described in Section~\ref{sec:properties}.

\bmhead{Acknowledgments}
    The second author is also with the University of Warsaw within the doctoral school of exact and natural sciences. N.H. thanks Wolfgang Polonik for a helpful discussion.

\bmhead{Funding}
    P.D., N.H. and R.T. were supported by the Dioscuri program initiated by the Max Planck Society,
    jointly managed with the National Science Centre (Poland), and mutually funded by the Polish
    Ministry of Science and Higher Education and the German Federal Ministry of Education and
    Research. {\L}.S. was supported by NCN grant no. 2020/37/B/HS4/00120.
    Computations reported in this paper were performed on the infrastructure provided by Google Cloud Higher Education Program.

\bibliography{references}      


\begin{thebibliography}{39}
\ifx \bisbn   \undefined \def \bisbn  #1{ISBN #1}\fi
\ifx \binits  \undefined \def \binits#1{#1}\fi
\ifx \bauthor  \undefined \def \bauthor#1{#1}\fi
\ifx \batitle  \undefined \def \batitle#1{#1}\fi
\ifx \bjtitle  \undefined \def \bjtitle#1{#1}\fi
\ifx \bvolume  \undefined \def \bvolume#1{\textbf{#1}}\fi
\ifx \byear  \undefined \def \byear#1{#1}\fi
\ifx \bissue  \undefined \def \bissue#1{#1}\fi
\ifx \bfpage  \undefined \def \bfpage#1{#1}\fi
\ifx \blpage  \undefined \def \blpage #1{#1}\fi
\ifx \burl  \undefined \def \burl#1{\textsf{#1}}\fi
\ifx \doiurl  \undefined \def \doiurl#1{\url{https://doi.org/#1}}\fi
\ifx \betal  \undefined \def \betal{\textit{et al.}}\fi
\ifx \binstitute  \undefined \def \binstitute#1{#1}\fi
\ifx \binstitutionaled  \undefined \def \binstitutionaled#1{#1}\fi
\ifx \bctitle  \undefined \def \bctitle#1{#1}\fi
\ifx \beditor  \undefined \def \beditor#1{#1}\fi
\ifx \bpublisher  \undefined \def \bpublisher#1{#1}\fi
\ifx \bbtitle  \undefined \def \bbtitle#1{#1}\fi
\ifx \bedition  \undefined \def \bedition#1{#1}\fi
\ifx \bseriesno  \undefined \def \bseriesno#1{#1}\fi
\ifx \blocation  \undefined \def \blocation#1{#1}\fi
\ifx \bsertitle  \undefined \def \bsertitle#1{#1}\fi
\ifx \bsnm \undefined \def \bsnm#1{#1}\fi
\ifx \bsuffix \undefined \def \bsuffix#1{#1}\fi
\ifx \bparticle \undefined \def \bparticle#1{#1}\fi
\ifx \barticle \undefined \def \barticle#1{#1}\fi
\bibcommenthead
\ifx \bconfdate \undefined \def \bconfdate #1{#1}\fi
\ifx \botherref \undefined \def \botherref #1{#1}\fi
\ifx \url \undefined \def \url#1{\textsf{#1}}\fi
\ifx \bchapter \undefined \def \bchapter#1{#1}\fi
\ifx \bbook \undefined \def \bbook#1{#1}\fi
\ifx \bcomment \undefined \def \bcomment#1{#1}\fi
\ifx \oauthor \undefined \def \oauthor#1{#1}\fi
\ifx \citeauthoryear \undefined \def \citeauthoryear#1{#1}\fi
\ifx \endbibitem  \undefined \def \endbibitem {}\fi
\ifx \bconflocation  \undefined \def \bconflocation#1{#1}\fi
\ifx \arxivurl  \undefined \def \arxivurl#1{\textsf{#1}}\fi
\csname PreBibitemsHook\endcsname

\bibitem[\protect\citeauthoryear{D'Agostino and Stephens}{1986}]{dagostino1986goodness}
\begin{bbook}
\beditor{\bsnm{D'Agostino}, \binits{R.B.}},
\beditor{\bsnm{Stephens}, \binits{M.A.}} (eds.):
\bbtitle{Goodness-of-fit Techniques}.
\bpublisher{Chapman \& Hall/CRC},
\blocation{{B}oca {R}aton}
(\byear{1986})
\end{bbook}
\endbibitem

\bibitem[\protect\citeauthoryear{Fasano and Franceschini}{1987}]{Fasano1987}
\begin{barticle}
\bauthor{\bsnm{Fasano}, \binits{G.}},
\bauthor{\bsnm{Franceschini}, \binits{A.}}:
\batitle{{A multidimensional version of the Kolmogorov–Smirnov test}}.
\bjtitle{Monthly Notices of the Royal Astronomical Society}
\bvolume{225}(\bissue{1}),
\bfpage{155}--\blpage{170}
(\byear{1987})
\doiurl{10.1093/mnras/225.1.155}
\end{barticle}
\endbibitem

\bibitem[\protect\citeauthoryear{Peacock}{1983}]{Peacock1983}
\begin{barticle}
\bauthor{\bsnm{Peacock}, \binits{J.A.}}:
\batitle{{Two-dimensional goodness-of-fit testing in astronomy}}.
\bjtitle{Monthly Notices of the Royal Astronomical Society}
\bvolume{202}(\bissue{3}),
\bfpage{615}--\blpage{627}
(\byear{1983})
\doiurl{10.1093/mnras/202.3.615}
\end{barticle}
\endbibitem

\bibitem[\protect\citeauthoryear{Justel et~al.}{1997}]{Justel1997}
\begin{barticle}
\bauthor{\bsnm{Justel}, \binits{A.}},
\bauthor{\bsnm{Peña}, \binits{D.}},
\bauthor{\bsnm{Zamar}, \binits{R.}}:
\batitle{A multivariate {K}olmogorov-{S}mirnov test of goodness of fit}.
\bjtitle{Statistics \& Probability Letters}
\bvolume{35}(\bissue{3}),
\bfpage{251}--\blpage{259}
(\byear{1997})
\doiurl{10.1016/S0167-7152(97)00020-5}
\end{barticle}
\endbibitem

\bibitem[\protect\citeauthoryear{Chiu and Liu}{2009}]{Chiu2009}
\begin{barticle}
\bauthor{\bsnm{Chiu}, \binits{S.N.}},
\bauthor{\bsnm{Liu}, \binits{K.I.}}:
\batitle{Generalized {C}ramér–von {M}ises goodness-of-fit tests for multivariate distributions}.
\bjtitle{Computational Statistics \& Data Analysis}
\bvolume{53}(\bissue{11}),
\bfpage{3817}--\blpage{3834}
(\byear{2009})
\doiurl{10.1016/j.csda.2009.04.004}
\end{barticle}
\endbibitem

\bibitem[\protect\citeauthoryear{Gonzalez and Wintz}{1977}]{gonzalez1977digital}
\begin{botherref}
\oauthor{\bsnm{Gonzalez}, \binits{R.C.}},
\oauthor{\bsnm{Wintz}, \binits{P.}}:
Digital image processing(book).
Reading, Mass., Addison-Wesley Publishing Co., Inc.(Applied Mathematics and Computation
(13),
451
(1977)
\end{botherref}
\endbibitem

\bibitem[\protect\citeauthoryear{Richardson and Werman}{2014}]{richardsonEfficientClassificationUsing2014}
\begin{barticle}
\bauthor{\bsnm{Richardson}, \binits{E.}},
\bauthor{\bsnm{Werman}, \binits{M.}}:
\batitle{Efficient classification using the {{Euler}} characteristic}.
\bjtitle{Pattern Recognition Letters}
\bvolume{49},
\bfpage{99}--\blpage{106}
(\byear{2014})
\doiurl{10/f6mz6s}
\end{barticle}
\endbibitem

\bibitem[\protect\citeauthoryear{Worsley}{1996}]{Worsley1996Geometry}
\begin{barticle}
\bauthor{\bsnm{Worsley}, \binits{K.J.}}:
\batitle{The geometry of random images}.
\bjtitle{CHANCE}
\bvolume{9}(\bissue{1}),
\bfpage{27}--\blpage{40}
(\byear{1996})
\doiurl{10.1080/09332480.1996.10542483}
\end{barticle}
\endbibitem

\bibitem[\protect\citeauthoryear{{Edelsbrunner} et~al.}{2002}]{edelsbrunner_topological_2002}
\begin{barticle}
\bauthor{\bsnm{{Edelsbrunner}}},
\bauthor{\bsnm{{Letscher}}},
\bauthor{\bsnm{{Zomorodian}}}:
\batitle{Topological {Persistence} and {Simplification}}.
\bjtitle{Discrete Comput Geom}
\bvolume{28}(\bissue{4}),
\bfpage{511}--\blpage{533}
(\byear{2002})
\doiurl{10.1007/s00454-002-2885-2}
\end{barticle}
\endbibitem

\bibitem[\protect\citeauthoryear{Zomorodian and Carlsson}{2005}]{zomorodian_computing_2005}
\begin{barticle}
\bauthor{\bsnm{Zomorodian}, \binits{A.}},
\bauthor{\bsnm{Carlsson}, \binits{G.}}:
\batitle{Computing {Persistent} {Homology}}.
\bjtitle{Discrete Comput Geom}
\bvolume{33}(\bissue{2}),
\bfpage{249}--\blpage{274}
(\byear{2005})
\doiurl{10.1007/s00454-004-1146-y}
\end{barticle}
\endbibitem

\bibitem[\protect\citeauthoryear{Wasserman}{2018}]{wasserman_topological_2018}
\begin{botherref}
\oauthor{\bsnm{Wasserman}, \binits{L.}}:
Topological {Data} {Analysis},
Rochester, NY
(2018).
\doiurl{10.1146/annurev-statistics-031017-100045} .
\url{https://papers.ssrn.com/abstract=3156968}
\end{botherref}
\endbibitem

\bibitem[\protect\citeauthoryear{Cericola et~al.}{}]{Cericola2016}
\begin{botherref}
\oauthor{\bsnm{Cericola}, \binits{C.}},
\oauthor{\bsnm{Johnson}, \binits{I.}},
\oauthor{\bsnm{Kiers}, \binits{J.}},
\oauthor{\bsnm{Krock}, \binits{M.}},
\oauthor{\bsnm{Purdy}, \binits{J.}},
\oauthor{\bsnm{Torrence}, \binits{J.}}:
Extending hypothesis testing with persistence homology to three or more groups
\doiurl{10.2140/involve.2018.11.27}
{\href{https://arxiv.org/abs/1602.03760v1}{{1602.03760v1}}}
\end{botherref}
\endbibitem

\bibitem[\protect\citeauthoryear{Robinson and Turner}{2017}]{Robinson2017}
\begin{botherref}
\oauthor{\bsnm{Robinson}, \binits{A.}},
\oauthor{\bsnm{Turner}, \binits{K.}}:
Hypothesis testing for topological data analysis.
Journal of Applied and Computational Topology
\textbf{1}(2)
(2017)
\doiurl{10.1007/s41468-017-0008-7}
\end{botherref}
\endbibitem

\bibitem[\protect\citeauthoryear{Vejdemo-Johansson and Mukherjee}{}]{vejdemo-johanssonMultipleTestingPersistent2022}
\begin{botherref}
\oauthor{\bsnm{Vejdemo-Johansson}, \binits{M.}},
\oauthor{\bsnm{Mukherjee}, \binits{S.}}:
Multiple testing with persistent homology
{\href{https://arxiv.org/abs/1812.06491v4}{{1812.06491v4}}}
\end{botherref}
\endbibitem

\bibitem[\protect\citeauthoryear{Edelsbrunner and Harer}{2010}]{edelsbrunnerComputationalTopologyIntroduction2010}
\begin{bbook}
\bauthor{\bsnm{Edelsbrunner}, \binits{H.}},
\bauthor{\bsnm{Harer}, \binits{J.L.}}:
\bbtitle{Computational Topology. {{An}} Introduction}.
\bpublisher{{American Mathematical Society (AMS)}},
\blocation{{Providence, RI}}
(\byear{2010})
\end{bbook}
\endbibitem

\bibitem[\protect\citeauthoryear{{de Berg} et~al.}{2008}]{debergDelaunayTriangulations2008}
\begin{bchapter}
\bauthor{\bsnm{{de Berg}}, \binits{M.}},
\bauthor{\bsnm{Cheong}, \binits{O.}},
\bauthor{\bsnm{{van Kreveld}}, \binits{M.}},
\bauthor{\bsnm{Overmars}, \binits{M.}}:
\bctitle{Delaunay {{Triangulations}}}.
In: \bbtitle{Computational {{Geometry}}: {{Algorithms}} and {{Applications}}},
pp. \bfpage{191}--\blpage{218}.
\bpublisher{{Springer}},
\blocation{{Berlin, Heidelberg}}
(\byear{2008}).
\doiurl{10.1007/978-3-540-77974-2\_9}
\end{bchapter}
\endbibitem

\bibitem[\protect\citeauthoryear{Edelsbrunner et~al.}{2017}]{edelsbrunnerExpectedSizesPoisson2017}
\begin{barticle}
\bauthor{\bsnm{Edelsbrunner}, \binits{H.}},
\bauthor{\bsnm{Nikitenko}, \binits{A.}},
\bauthor{\bsnm{Reitzner}, \binits{M.}}:
\batitle{Expected sizes of {{Poisson}}-{{Delaunay}} mosaics and their discrete {{Morse}} functions}.
\bjtitle{Advances in Applied Probability}
\bvolume{49}(\bissue{3}),
\bfpage{745}--\blpage{767}
(\byear{2017})
\doiurl{10.1017/apr.2017.20}
\end{barticle}
\endbibitem

\bibitem[\protect\citeauthoryear{Bobrowski and Kahle}{2018}]{bobrowskiTopologyRandomGeometric2017}
\begin{barticle}
\bauthor{\bsnm{Bobrowski}, \binits{O.}},
\bauthor{\bsnm{Kahle}, \binits{M.}}:
\batitle{Topology of random geometric complexes: A survey}.
\bjtitle{J Appl. and Comput. Topology}
\bvolume{1}(\bissue{3-4}),
\bfpage{331}--\blpage{364}
(\byear{2018})
\doiurl{10.1007/s41468-017-0010-0}
\end{barticle}
\endbibitem

\bibitem[\protect\citeauthoryear{Penrose}{2003}]{penroseRandomGeometricGraphs2003}
\begin{bbook}
\bauthor{\bsnm{Penrose}, \binits{M.}}:
\bbtitle{Random {Geometric} {Graphs}}.
\bsertitle{Oxford {{Studies}} in {{Probability}}}.
\bpublisher{{Oxford University Press}},
\blocation{{Oxford}}
(\byear{2003}).
\doiurl{10.1093/acprof:oso/9780198506263.001.0001}
\end{bbook}
\endbibitem

\bibitem[\protect\citeauthoryear{Krebs et~al.}{2021}]{krebsApproximationTheoremsEuler2021}
\begin{barticle}
\bauthor{\bsnm{Krebs}, \binits{J.}},
\bauthor{\bsnm{Roycraft}, \binits{B.}},
\bauthor{\bsnm{Polonik}, \binits{W.}}:
\batitle{On approximation theorems for the {{Euler}} characteristic with applications to the bootstrap}.
\bjtitle{Electronic Journal of Statistics}
\bvolume{15}(\bissue{2}),
\bfpage{4462}--\blpage{4509}
(\byear{2021})
\doiurl{10.1214/21-EJS1898}
\end{barticle}
\endbibitem

\bibitem[\protect\citeauthoryear{{Cohen-Steiner} et~al.}{2007}]{cohen-steinerStabilityPersistenceDiagrams2007}
\begin{barticle}
\bauthor{\bsnm{{Cohen-Steiner}}, \binits{D.}},
\bauthor{\bsnm{Edelsbrunner}, \binits{H.}},
\bauthor{\bsnm{Harer}, \binits{J.}}:
\batitle{Stability of {{Persistence Diagrams}}}.
\bjtitle{Discrete Comput Geom}
\bvolume{37}(\bissue{1}),
\bfpage{103}--\blpage{120}
(\byear{2007})
\doiurl{10.1007/s00454-006-1276-5}
\end{barticle}
\endbibitem

\bibitem[\protect\citeauthoryear{Turner et~al.}{2014}]{turnerFrechetMeansDistributions2014}
\begin{barticle}
\bauthor{\bsnm{Turner}, \binits{K.}},
\bauthor{\bsnm{Mileyko}, \binits{Y.}},
\bauthor{\bsnm{Mukherjee}, \binits{S.}},
\bauthor{\bsnm{Harer}, \binits{J.}}:
\batitle{Fr\'echet {{Means}} for {{Distributions}} of {{Persistence Diagrams}}}.
\bjtitle{Discrete Comput Geom}
\bvolume{52}(\bissue{1}),
\bfpage{44}--\blpage{70}
(\byear{2014})
\doiurl{10.1007/s00454-014-9604-7}
\end{barticle}
\endbibitem

\bibitem[\protect\citeauthoryear{Bobrowski and Adler}{2014}]{bobrowskiDistanceFunctionsCritical2014}
\begin{barticle}
\bauthor{\bsnm{Bobrowski}, \binits{O.}},
\bauthor{\bsnm{Adler}, \binits{R.J.}}:
\batitle{Distance functions, critical points, and the topology of random {{\v{C}ech}} complexes}.
\bjtitle{Homology, Homotopy and Applications}
\bvolume{16}(\bissue{2}),
\bfpage{311}--\blpage{344}
(\byear{2014})
\doiurl{10.4310/HHA.2014.v16.n2.a18}
\end{barticle}
\endbibitem

\bibitem[\protect\citeauthoryear{Bobrowski and Mukherjee}{2013}]{bobrowskiTopologyProbabilityDistributions2013}
\begin{botherref}
\oauthor{\bsnm{Bobrowski}, \binits{O.}},
\oauthor{\bsnm{Mukherjee}, \binits{S.}}:
The {{Topology}} of {{Probability Distributions}} on {{Manifolds}}.
Probability Theory and Related Fields
\textbf{161}
(2013)
\doiurl{10.1007/s00440-014-0556-x}
\end{botherref}
\endbibitem

\bibitem[\protect\citeauthoryear{Thomas and Owada}{2021}]{thomasFunctionalLimitTheorems2020b}
\begin{barticle}
\bauthor{\bsnm{Thomas}, \binits{A.M.}},
\bauthor{\bsnm{Owada}, \binits{T.}}:
\batitle{Functional limit theorems for the {E}uler {C}haracteristic process in the critical regime}.
\bjtitle{Advances in Applied Probability}
\bvolume{53}(\bissue{1}),
\bfpage{57}--\blpage{80}
(\byear{2021})
\doiurl{10.1017/apr.2020.46}
\end{barticle}
\endbibitem

\bibitem[\protect\citeauthoryear{Adler and Taylor}{}]{adler_random_2007}
\begin{botherref}
\oauthor{\bsnm{Adler}, \binits{R.J.}},
\oauthor{\bsnm{Taylor}, \binits{J.E.}}:
Random Fields and Geometry.
Springer Monogr. Math.
Springer.
\doiurl{10.1007/978-0-387-48116-6} .
{ISSN}: 1439-7382
\end{botherref}
\endbibitem

\bibitem[\protect\citeauthoryear{Adler}{}]{adler_new_2008}
\begin{botherref}
\oauthor{\bsnm{Adler}, \binits{R.J.}}:
Some new random field tools for spatial analysis
\textbf{22}(6),
809--822
\doiurl{10.1007/s00477-008-0242-6} .
Accessed 2023-03-20
\end{botherref}
\endbibitem

\bibitem[\protect\citeauthoryear{Bernardino et~al.}{}]{bernardino_test_2017}
\begin{botherref}
\oauthor{\bsnm{Bernardino}, \binits{E.D.}},
\oauthor{\bsnm{Estrade}, \binits{A.}},
\oauthor{\bsnm{León}, \binits{J.R.}}:
A test of gaussianity based on the euler characteristic of excursion sets
\textbf{11}(1),
843--890
\doiurl{10.1214/17-EJS1248} .
Publisher: Institute of Mathematical Statistics and Bernoulli Society.
Accessed 2023-03-20
\end{botherref}
\endbibitem

\bibitem[\protect\citeauthoryear{Cipriani et~al.}{2022}]{ciprianiTopologybasedGoodnessoffitTests2022a}
\begin{botherref}
\oauthor{\bsnm{Cipriani}, \binits{A.}},
\oauthor{\bsnm{Hirsch}, \binits{C.}},
\oauthor{\bsnm{Vittorietti}, \binits{M.}}:
Topology-based goodness-of-fit tests for sliced spatial data
(2022)
{\href{https://arxiv.org/abs/2201.04092}{{arXiv:2201.04092}}}
\end{botherref}
\endbibitem

\bibitem[\protect\citeauthoryear{Biscio et~al.}{2020}]{biscioTestingGoodnessFit2019}
\begin{barticle}
\bauthor{\bsnm{Biscio}, \binits{C.A.N.}},
\bauthor{\bsnm{Chenavier}, \binits{N.}},
\bauthor{\bsnm{Hirsch}, \binits{C.}},
\bauthor{\bsnm{Svane}, \binits{A.M.}}:
\batitle{Testing goodness of fit for point processes via topological data analysis}.
\bjtitle{Electronic Journal of Statistics}
\bvolume{14}(\bissue{1}),
\bfpage{1024}--\blpage{1074}
(\byear{2020})
\doiurl{10.1214/20-EJS1683}
\end{barticle}
\endbibitem

\bibitem[\protect\citeauthoryear{Botnan and Hirsch}{2021}]{botnanConsistencyAsymptoticNormality2021}
\begin{botherref}
\oauthor{\bsnm{Botnan}, \binits{M.B.}},
\oauthor{\bsnm{Hirsch}, \binits{C.}}:
On the consistency and asymptotic normality of multiparameter persistent {{Betti}} numbers.
arXiv:2109.05513 [math, stat]
(2021)
{\href{https://arxiv.org/abs/2109.05513}{{arXiv:2109.05513}}}
{[math, stat]}
\end{botherref}
\endbibitem

\bibitem[\protect\citeauthoryear{Vishwanath et~al.}{}]{vishwanath_limits_2022}
\begin{botherref}
\oauthor{\bsnm{Vishwanath}, \binits{S.}},
\oauthor{\bsnm{Fukumizu}, \binits{K.}},
\oauthor{\bsnm{Kuriki}, \binits{S.}},
\oauthor{\bsnm{Sriperumbudur}, \binits{B.}}:
On the limits of topological data analysis for statistical inference
({arXiv}:2001.00220)
\end{botherref}
\endbibitem

\bibitem[\protect\citeauthoryear{Arias-Castro}{2022}]{arias2022principles}
\begin{bbook}
\bauthor{\bsnm{Arias-Castro}, \binits{E.}}:
\bbtitle{Principles of Statistical Analysis: Learning from Randomized Experiments}.
\bsertitle{Institute of Mathematical Statistics Textbooks}.
\bpublisher{Cambridge University Press},
\blocation{Cambridge}
(\byear{2022}).
\doiurl{10.1017/9781108779197}
\end{bbook}
\endbibitem

\bibitem[\protect\citeauthoryear{Adler and Taylor}{2007}]{AdlerGaussianInequalities2007}
\begin{bbook}
\bauthor{\bsnm{Adler}, \binits{R.J.}},
\bauthor{\bsnm{Taylor}, \binits{J.E.}}:
\bbtitle{Gaussian Inequalities},
pp. \bfpage{49}--\blpage{64}.
\bpublisher{Springer},
\blocation{New York, NY}
(\byear{2007}).
\doiurl{10.1007/978-0-387-48116-6}
\end{bbook}
\endbibitem

\bibitem[\protect\citeauthoryear{Ledoux}{2005}]{ledouxConcentrationMeasurePhenomenon2005}
\begin{bbook}
\bauthor{\bsnm{Ledoux}, \binits{M.}}:
\bbtitle{The {Concentration} of {Measure} {Phenomenon}}.
\bsertitle{Mathematical {{Surveys}} and {{Monographs}}},
vol. \bseriesno{89}.
\bpublisher{{American Mathematical Society}},
\blocation{{Providence, Rhode Island}}
(\byear{2005}).
\doiurl{10.1090/surv/089}
\end{bbook}
\endbibitem

\bibitem[\protect\citeauthoryear{Casella and Berger}{2002}]{castella2002}
\begin{bbook}
\bauthor{\bsnm{Casella}, \binits{G.}},
\bauthor{\bsnm{Berger}, \binits{R.L.}}:
\bbtitle{Statistical Inference}.
\bpublisher{{Duxbury Press}},
\blocation{{Pacific Grove, CA}}
(\byear{2002})
\end{bbook}
\endbibitem

\bibitem[\protect\citeauthoryear{Dhar et~al.}{}]{Dhar2014}
\begin{botherref}
\oauthor{\bsnm{Dhar}, \binits{S.S.}},
\oauthor{\bsnm{Chakraborty}, \binits{B.}},
\oauthor{\bsnm{Chaudhuri}, \binits{P.}}:
Comparison of multivariate distributions using quantile-quantile plots and related tests
\textbf{20}(3),
1484--1506
\doiurl{10.3150/13-BEJ530}
\end{botherref}
\endbibitem

\bibitem[\protect\citeauthoryear{Floch et~al.}{2018}]{Floch2018}
\begin{botherref}
\oauthor{\bsnm{Floch}, \binits{J.-M.}},
\oauthor{\bsnm{Marcon}, \binits{E.}},
\oauthor{\bsnm{Puech}, \binits{F.}}:
Spatial distribution of points,
pp. 71--111.
Insee-Eurostat
(2018).
\url{https://www.insee.fr/en/information/3635545}
\end{botherref}
\endbibitem

\bibitem[\protect\citeauthoryear{Puritz et~al.}{2022}]{puritz2022fasanofranceschinitest}
\begin{botherref}
\oauthor{\bsnm{Puritz}, \binits{C.}},
\oauthor{\bsnm{Ness-Cohn}, \binits{E.}},
\oauthor{\bsnm{Braun}, \binits{R.}}:
fasano.franceschini.test: An Implementation of a Multidimensional KS Test in R
(2022)
\end{botherref}
\endbibitem

\end{thebibliography}

\end{document}